\documentclass[conference]{IEEEtran}
\IEEEoverridecommandlockouts

\usepackage{hyperref}
\usepackage{graphicx}
\usepackage{booktabs}
\usepackage{amsmath}
\usepackage{amsthm}
\usepackage{amssymb}
\usepackage{amsthm}
\usepackage{bbm}
\usepackage{amsfonts}
\usepackage{epsfig}
\usepackage[labelformat=simple]{subcaption}
\usepackage{url}
\usepackage{balance}
\usepackage[noend]{algorithmic}
\usepackage{algorithm}
\usepackage{multirow}
\usepackage{xspace}
\usepackage{xcolor}
\usepackage{textcomp}
\usepackage{pgfplots}
\usepackage{cleveref}

\usepackage{pifont}
\newtheorem{theor}{Theorem}

\newtheorem{corry}{Corollary}
\newtheorem{lem}{Lemma}
\newtheorem{clm}{Claim}

\newtheorem{defn}{Definition}

\newtheorem{exam}{Example}
\newtheorem{problem}{Problem}

\newcommand{\spara}[1]{\smallskip\noindent{\bf #1}}

\newcommand{\squishlist}{
 \begin{list}{$\bullet$}
  {  \setlength{\itemsep}{0pt}
     \setlength{\parsep}{3pt}
     \setlength{\topsep}{3pt}
     \setlength{\partopsep}{0pt}
     \setlength{\leftmargin}{2em}
     \setlength{\labelwidth}{1.5em}
     \setlength{\labelsep}{0.5em}
} }
\newcommand{\squishlisttight}{
 \begin{list}{$\bullet$}
  { \setlength{\itemsep}{0pt}
    \setlength{\parsep}{0pt}
    \setlength{\topsep}{0pt}
    \setlength{\partopsep}{0pt}
    \setlength{\leftmargin}{2em}
    \setlength{\labelwidth}{1.5em}
    \setlength{\labelsep}{0.5em}
} }

\newcommand{\squishdesc}{
 \begin{list}{}
  {  \setlength{\itemsep}{0pt}
     \setlength{\parsep}{3pt}
     \setlength{\topsep}{3pt}
     \setlength{\partopsep}{0pt}
     \setlength{\leftmargin}{1em}
     \setlength{\labelwidth}{1.5em}
     \setlength{\labelsep}{0.5em}
} }

\newcommand{\squishend}{
  \end{list}
}

\newcommand{\eat}[1]{}

\newcommand{\sharpP}{\ensuremath{\mathbf{\#P}}\xspace}

\newcounter{ccc}

\DeclareMathOperator*{\argmax}{arg\,max}
\newcommand{\bigO}{\mathcal{O}}

\crefformat{subfigure}{#2\onlyletter{#1}#3}
\crefrangeformat{subfigure}{(#3\onlyletter{#1}#4--#5\onlyletter{#2}#6)}
\crefmultiformat{subfigure}
  {(#2\onlyletter{#1}#3}
  {,~#2\onlyletter{#1}#3}
  {,~#2\onlyletter{#1}#3}
  {,~#2\onlyletter{#1}#3)}

\pgfplotsset{width=7cm, compat=1.9, scaled y ticks=false, yticklabel style={/pgf/number format/fixed, /pgf/number format/precision=4}}
\usetikzlibrary{patterns,shapes,arrows}

\usepackage{color,soul}

\begin{document}
\bstctlcite{IEEEexample:BSTcontrol}
\newcommand{\revise}[1]{#1}
\newcommand{\shepherd}[1]{#1}
\title{Most Probable Densest Subgraphs
\thanks{Xiangyu Ke is the corresponding author. Arijit Khan acknowledges support from the Novo Nordisk Foundation grant NNF22OC0072415. Cheng Long is supported by the Ministry of Education, Singapore, under its Academic Research Fund (Tier 1 Award (RG77/21)). Any opinions, findings, conclusions or recommendations expressed in this material are those of the author(s) and do not reflect the views of the funding agencies.}}

\author{\IEEEauthorblockN{Arkaprava Saha}
\IEEEauthorblockA{
\textit{NTU, Singapore}\\
saha0003@e.ntu.edu.sg}
\and
\IEEEauthorblockN{Xiangyu Ke}
\IEEEauthorblockA{
\textit{ZJU, China}\\
xiangyu.ke@zju.edu.cn}
\and
\IEEEauthorblockN{Arijit Khan}
\IEEEauthorblockA{
\textit{AAU, Denmark}\\
arijitk@cs.aau.dk}
\and
\IEEEauthorblockN{Cheng Long}
\IEEEauthorblockA{
\textit{NTU, Singapore}\\
c.long@ntu.edu.sg}
}

\maketitle

\begin{abstract}
Computing the densest subgraph is a primitive graph operation with critical applications in
detecting communities, events, and anomalies in biological, social, Web, and financial
networks. In this paper, we study the novel problem of Most Probable Densest Subgraph
({\sf MPDS}) discovery in uncertain graphs: Find the node set that is the {\em most likely} to
induce a densest subgraph in an uncertain graph. We further extend our problem by considering
various notions of density, e.g., clique and pattern densities, studying the top-$k$ {\sf MPDS}s,
and finding the node set with the largest containment probability
within densest subgraphs. \revise{We show that it is $\sharpP$-hard to
compute the probability of a node set inducing a densest subgraph.} We then
devise sampling-based efficient algorithms, with end-to-end accuracy guarantees, to
compute the {\sf MPDS}. Our thorough experimental results and real-world
case studies on brain and social networks validate the
effectiveness, efficiency, and usefulness of our solution.
\end{abstract}

\begin{IEEEkeywords}
uncertain graphs, densest subgraphs
\end{IEEEkeywords}

\section{Introduction}
\label{sec:intro}

The discovery of dense subgraphs has attracted extensive attention in the data management community \cite{G84, Charikar00, TsourakakisBGGT13, GT15, FYCLX19}.
They may correspond to communities \cite{DourisboureGP09}, filter bubbles and echo chambers \cite{asatani2021dense, L22} in social networks,
brain regions responding to stimuli \cite{legenstein_et_al:LIPIcs} or related to diseases \cite{wu2021extracting},
and commercial value motifs in financial domains \cite{DuJDLT09}.
They also have wide applications in graph compression and visualization \cite{BuehrerC08,ZhangP12,ZhaoT12},
indexing for reachability and distance queries \cite{CohenHKZ03,JinXRF09},
and social piggybacking \cite{GionisJLSW13}.
Densest subgraphs usually maximize some notion of density in a given graph, e.g., the {\em edge density} \cite{G84},
defined as the ratio of the number of induced edges to the number of nodes in a subgraph.
Although there are an exponential number of subgraphs, a densest subgraph can be found both exactly
and approximately in polynomial time \cite{G84, Charikar00}.
There also exist many other density metrics \cite{FaragoM19}, such as the edge ratio, edge surplus, discounted average degree, triangle density, clique density, pattern density,
etc. Their
algorithms are developed in \cite{T15, MPPTX15, YanagisawaH18, FYCLX19}.

Uncertainty is intrinsic in large graphs due to errors in measurements \cite{Aggarwal09},
edge imputation using inference and prediction models \cite{AdarR07,Liben-NowellK07},
and explicit manipulation including
privacy reasons \cite{BoldiBGT12}. An uncertain graph, where every edge is associated with a probability of existence,
is an expressive data model that has prompted a great deal of research
\cite{KYC18,KassianoGPT16,KhanC15}. Uncertain graphs are prevalent in many applications, such as biological
networks \cite{SevonEHKT06}, knowledge bases \cite{ZhangRCSWW17}, social networks \cite{ZouGL10}, influence
maximization \cite{KempeKT03}, road networks \cite{HP10}, and crowd sourcing \cite{YalavarthiKK17}.
Recently, researchers have extended several classic network problems to uncertain graphs,
e.g., nearest neighbors \cite{PBGK10},
shortest paths and centrality \cite{SahaBVKB21},
cliques \cite{MukherjeeXT17,ZouLGZ10,LiDWMQY19},
core and truss decomposition \cite{BonchiGKV14,HuangLL16},
clustering \cite{HanGXTHCH19},
and motif counting \cite{MaCLGF019}.
Surprisingly, except for maximum {\em expected edge density} \cite{Zou13, tsourakakis2020novel}, the study of densest subgraph discovery on uncertain graphs is still absent.
\begin{figure}
  \vspace{-1mm}
  \centering
  \includegraphics[scale=0.3]{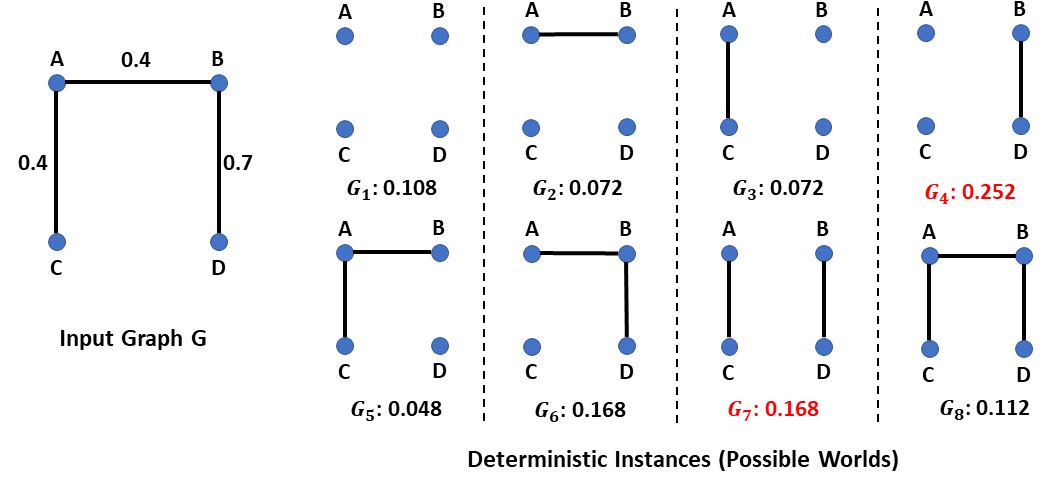}
    \vspace{-2mm}
  \captionof{figure}{\small Possible worlds of an uncertain graph and their probabilities}
  \label{fig:example_1}
    \vspace{-6mm}
\end{figure}
The expected edge density of an uncertain graph is defined as the expectation of the edge density value of a possible world (i.e., a deterministic graph)
of the uncertain graph, chosen at random \cite{Zou13}. However, a subgraph of the uncertain graph having the maximum expected edge density may induce densest subgraphs {\em only in a few} (even zero) possible worlds of that uncertain graph. Such a subgraph can be large with many low-probability
edges, or having nodes that are loosely connected (see Example~\ref{exam:expected}). This defeats the purpose of finding a densest subgraph.
Instead, many applications would require a densest subgraph with a high precision, such as being the densest with a high probability.
Specifically, given an uncertain graph $\mathcal{G}$, {\em our goal is to find the node set that is the most likely to induce a densest subgraph
in $\mathcal{G}$}, i.e., maximize the sum of the probabilities of those possible worlds of $\mathcal{G}$ in which this node set induces a densest subgraph. We refer to the uncertain subgraph induced by this node set as the {\em most probable densest subgraph} ({\sf MPDS}). To the best of our
knowledge, computing the {\sf MPDS} is a novel problem. We demonstrate real-world applications and case studies of our problem on uncertain brain (\S~\ref{sec:case_brain}) and
social (\S~\ref{sec:case_club}) networks, where our proposed {\sf MPDS} distinguishes healthy brains from those with autism and identifies meaningful communities in a social network.

\begin{table}[t]
	\scriptsize
	\centering
	\begin{center}
        \caption{\small Edge densities in possible worlds (PWs), {\em expected edge densities} (EEDs) and {\em densest subgraph probabilities} (DSPs) of node sets in the uncertain graph in Figure~\ref{fig:example_1}. The EED of a node set $U$ is the sum of the edge densities of the subgraphs induced by $U$ across all PWs, weighted by their probabilities. The DSP of $U$ is the sum of the probabilities of those PWs where $U$ induces a densest subgraph.}
		\begin{tabular} {c||c|c|c|c|c|c}
        {\bf PW:Pr.} & \{A,B\} & \{A,C\} & \{B,D\} & \{A,B,C\} & \{A,B,D\} & \{A,B,C,D\}
		\\ \hline \hline
		$G_1$:0.11 & 0 & 0 & 0 & 0 & 0 & 0 \\ \hline
		$G_2$:0.07 & {\bf 0.5} & 0 & 0 & 0.33 & 0.33 & 0.25 \\ \hline
		$G_3$:0.07 & 0 & {\bf 0.5} & 0 & 0.33 & 0 & 0.25 \\ \hline
		${\mathbf G_4}$:{\bf 0.25} & 0 & 0 & {\bf 0.5} & 0 & 0.33 & 0.25 \\ \hline
		$G_5$:0.05 & 0.5 & 0.5 & 0 & {\bf 0.67} & 0.33 & 0.5 \\ \hline
		$G_6$:0.17 & 0.5 & 0 & 0.5 & 0.33 & {\bf 0.67} & 0.5 \\ \hline
		${\mathbf G_7}$:{\bf 0.17} & 0 & {\bf 0.5} & {\bf 0.5} & 0.33 & 0.33 & {\bf 0.5} \\ \hline
		$G_8$:0.11 & 0.5 & 0.5 & 0.5 & 0.67 & 0.67 & {\bf 0.75} \\ \hline \hline
		{\bf EED} & 0.2 & 0.2 & 0.35 & 0.27 & 0.37 & {\bf 0.38} \\ \hline
		{\bf DSP} & 0.07 & 0.24  & {\bf 0.42} & 0.05  & 0.17 & 0.28
	\end{tabular}
		\vspace{-6mm}
		\label{tab:exp_density}
	\end{center}
\end{table}

\begin{exam}
Figure~\ref{fig:example_1} shows all possible worlds of an uncertain graph with their existence probabilities. It can be verified that, in each world, the connected component is also the densest subgraph. As Table~\ref{tab:exp_density} shows, the node set \{A,B,C,D\} has the maximum expected density (0.38), but it induces a densest subgraph {\em only} in possible worlds $G_7$ and $G_8$ with low existence probabilities (0.168 and 0.112). Thus, the probability of \{A,B,C,D\} inducing a densest subgraph is only 0.28. In contrast, the node set \{B,D\} has a lower expected density (0.35), but its probability of inducing a densest subgraph is much higher (0.42), since it induces a densest subgraph in possible worlds $G_4$ and $G_7$ with high existence probabilities.
\label{exam:expected}
\end{exam}

\spara{Challenges and our contributions.} We formulate and study the novel problem of {\em most probable densest subgraph} ({\sf MPDS}) discovery in uncertain graphs: Given an uncertain graph $\mathcal{G}$, find the node set that is the most likely to induce a densest subgraph in $\mathcal{G}$. Our contributions are summarized below.

\noindent$\bullet$ {\bf Novel problems:} To the best of our knowledge, the densest subgraph discovery in uncertain graphs has not been investigated before, other than expected edge density \cite{Zou13, tsourakakis2020novel}. \revise{Based on edge density, clique density, and pattern density \cite{FYCLX19}, we propose {\em densest subgraph probability} as a more sophisticated density metric.
We prove that computing the densest subgraph probability is $\sharpP$-hard.}
{\em We formulate and study the following novel problems of {\sf MPDS} discovery in uncertain graphs} (\S~\ref{sec:preliminaries}): {\sf MPDS}
based on edge density, clique density, and pattern density; for each of them, their top-$k$ variants and nucleus densest subgraph ({\sf NDS}) variants.
Real-world applications and case studies demonstrate the usefulness of our novel problems.

\noindent$\bullet$ {\bf Efficient approximate solution with end-to-end accuracy guarantees:}
\revise{In spite of the $\sharpP$-hardness of computing the densest subgraph probability, we design an efficient approximation algorithm for the top-$k$ densest subgraphs discovery, with an end-to-end accuracy guarantee.} Our solution for edge density-based {\sf MPDS} is built on independent sampling of possible worlds (e.g., via Monte-Carlo sampling) and, in each of them, efficient enumeration of all edge-densest subgraphs (via \cite{CQ20}). We provide time and space complexity analyses and theoretical accuracy guarantees of our method (\S \ref{sec:top_k_mpds_algorithm}).

\noindent$\bullet$ {\bf Extension to other density notions:} Besides edge density, our algorithm can be extended well to clique and pattern densities, while ensuring the same accuracy guarantee. We notice that, while there exist efficient algorithms to find {\em one} clique-densest and {\em one} pattern-densest subgraph in a deterministic graph \cite{FYCLX19}, the problems of enumerating {\em all} clique-densest and {\em all} pattern-densest subgraphs in a deterministic graph have not been studied earlier. However, such enumerations are required in our overall solution framework. Therefore, as additional technical contributions, we develop novel, {\em exact} algorithms for efficiently enumerating {\em all} clique-densest and {\em all} pattern-densest subgraphs in a deterministic graph (\S~\ref{sec:clique}, \ref{sec:pattern}).

\noindent$\bullet$ {\bf Practical nucleus densest subgraphs ({\sf NDS}):} In large graphs, we find that the densest subgraph probability of every possible node set may be quite small, due to the existence of many possible worlds, each having a smaller probability; and any two worlds might not have exactly the same densest subgraph. This defeats our purpose of identifying a node set that induces a {\em densest} subgraph with a {\em high probability}. In such cases, we propose to find those nodes which are most likely to form the ``nucleus'' of various densest subgraphs, i.e., whose containment probability within a densest subgraph is maximized. We develop an approximate solution and present theoretical analyses about its accuracy-efficiency trade-offs. {\em The novelty of our solution is that, by finding the maximum-sized densest subgraph in each sampled world, we reduce this problem to the \revise{closed} frequent itemset mining problem, for which efficient algorithms like \revise{\textsf{TFP} \cite{wang2005tfp}} exist} (\S~\ref{sec:nucleus}).

\noindent$\bullet$ {\bf Experiments and case studies:} Our rigorous experiments (\S~\ref{sec:exp}) show that our \textsf{MPDS} and \textsf{NDS} are different from existing notions of dense subgraphs in uncertain graphs (\S~\ref{sec:comp_others}). Also, our methods are very efficient even on large graphs (\S~\ref{sec:large_graph_results}) and return reasonably accurate results when compared to the exact methods (\S~\ref{sec:comp_exact}). Moreover, our case studies on brain (\S~\ref{sec:case_brain}) and social (\S~\ref{sec:case_club}) networks demonstrate useful real-world applications of the \textsf{MPDS}.


\section{Preliminaries}
\label{sec:preliminaries}
An uncertain graph $\mathcal{G}$ is a triple $(V,E,p)$, where $V$ is a set of $n$ nodes, $E\subseteq V\times V$ is a set of $m$ undirected unweighted edges, and the function $p: E \rightarrow (0,1]$ assigns a probability of existence to each edge. Following the bulk of the literature on uncertain graphs \cite{KYC18,PBGK10,SahaBVKB21,V79,B86,KBGN18}, we assume that the edges exist independent of each other: 
The uncertain graph $\mathcal{G}$ can be interpreted as a probability distribution over $2^{m}$ deterministic instances (possible worlds) $G=(V,E_G) \sqsubseteq \mathcal{G}$ obtained by independently sampling the edges. The probability of a possible world $G=(V,E_G)$ being observed is:
\begin{small}
\begin{align}
\Pr(G)=\prod_{e\in E_G}p(e)\prod_{e\in E \setminus E_G}(1-p(e))
\end{align}
\end{small}
%
In the rest of this section, we first revisit
several notions of {\em graph density} \cite{GT15,T15,FYCLX19}
in deterministic graphs (\S \ref{sec:density_notions}).
We next extend these notions to uncertain graphs based on the possible world semantics (\S \ref{sec:dsp}).
Then, our novel {\sf Most Probable Densest Subgraph (MPDS)} problem is formally introduced,
together with several practical variants (\S \ref{sec:problem}).
Finally, we discuss the hardness of our problem (\S \ref{sec:hardness}).
%
\subsection{Density Notions over Deterministic Graphs}
\label{sec:density_notions}
\subsubsection{Edge Density}
\label{sec:e_density}
The edge density \cite{G84} measures 
the average degree per node,
which can be used for community detection \cite{JY12, TsourakakisBGGT13} in social networks.
\begin{defn}
[Edge Density \cite{G84}]
The edge density $\rho_e$ of a deterministic graph $G = (V, E)$ is defined as:
\begin{small}
\begin{equation}
    \rho_e(G) = \frac{|E|}{|V|}
\label{eq:edge_density}
\end{equation}
\end{small}
\label{def:edge_density}
\end{defn}
\vspace{-6mm}
\subsubsection{Clique Density}
\label{sec:c_density}
An $h$-clique ($h \ge$ 2) is a complete graph of $h$ nodes. The clique density is formally defined below.
%
\begin{defn}
[$h$-Clique Density \cite{T15}]
Given an integer $h \geq 2$, the $h$-clique density $\rho_h$ of a deterministic graph $G = (V, E)$, with the number of $h$-cliques $\mu_h(G)$, is defined as:
\begin{small}
\vspace{-1mm}
\begin{equation}
    \rho_h(G) = \frac{\mu_h(G)}{|V|}
\label{eq:clique_density}
\end{equation}
\vspace{-4mm}
\end{small}
\label{def:clique_density}
\end{defn}
Notice that a 2-clique is an edge.
Thus, edge density is a special case of clique density when $h=2$.
Clique density benefits in higher-order community discovery and finding subgraphs which are large near-cliques \cite{YinBLG17,T15}.

\subsubsection{Pattern Density}
\label{sec:p_density}
Given an arbitrary pattern graph, the pattern density measures the average number of such patterns per node occurring in a subgraph.
%
\begin{defn}
[Pattern Density \cite{FYCLX19}]
Given a pattern graph $\psi$, the pattern density $\rho_\psi$ (w.r.t. $\psi$) of a deterministic graph $G = (V, E)$, with the number of $\psi$-instances $\mu_\psi(G)$, is defined as:
\begin{small}
\vspace{-1mm}
\begin{equation}
    \rho_\psi(G) = \frac{\mu_\psi(G)}{|V|}
    \label{eq:pattern_density}
\vspace{-4mm}
\end{equation}
\end{small}
\label{def:pattern_density}
\end{defn}
Clearly, clique density is a special case of pattern density when the input pattern is a clique. Pattern density
can be more expressive in real-world applications. For instance, in the LinkedIn social network, the ``employer'' nodes (e.g., companies)
cannot directly link to the ``education'' nodes (e.g., universities). Thus, ``employer'' nodes and ``education'' nodes never form a clique.
However, they can be connected via nodes representing ``employee''. A subgraph which is dense w.r.t. the
``employer-employee1-education-employee2-employer'' diamond pattern may identify a group of employees
with common work experiences and educational backgrounds.
\vspace{-4mm}
\subsection{Extending Density Notions to Uncertain Graphs}
\label{sec:dsp}
We define the probability of a node set inducing a densest subgraph in an uncertain graph using possible world semantics.
%
\begin{defn}
[Densest Subgraph Probability]
Given an uncertain graph $\mathcal{G} = (V, E, p)$ and a node set $U \subseteq V$, the densest subgraph probability of $U$,
denoted by $\tau(U)$, is the sum of the probabilities of all possible worlds where the subgraph induced by $U$ has the largest density.
Formally,
\begin{small}
\begin{equation}
\vspace{-1mm}
    \tau(U) = \sum_{G \sqsubseteq \mathcal{G}} \Pr(G) \times \mathbbm{1} \left[ \rho(G[U]) = \max_{W \subseteq V} \rho(G[W]) \right]
\label{eq:densest_subgraph_probability}
\vspace{-3mm}
\end{equation}
\end{small}
\label{def:densest_subgraph_probability}
\end{defn}
The above equation verifies, in each possible world of the uncertain graph $\mathcal{G}$, whether the node set $U$ induces a subgraph with the maximum density. $G[W]= (W, E_G[W])$ denotes the subgraph of $G$ induced by a node set $W\subseteq V$, where $E_G[W] = \{(u, v) \in E_G: u \in W, v \in W \}$. The indicator function $\mathbbm{1} \left[\cdot\right]$ returns 1 if the inner condition is true (i.e., the subgraph induced by $U$ has the largest density in $G$), and 0 otherwise. Note that the node set whose induced subgraph has the maximum edge density in $G$ may not be unique. The density metric $\rho$ can follow any of the density notions (\S~\ref{sec:density_notions}) based on the real application demand. In the following, without loss of generality, {\em densest subgraph probability} is coupled with edge density $\rho_e$ by default. For $h$-clique density $\rho_h$ and pattern density $\rho_\psi$, we refer to {\em $h$-clique densest subgraph probability} $\tau_h$ and {\em pattern densest subgraph probability} $\tau_\psi$ respectively.
\begin{figure}
  \centering
  \includegraphics[scale=0.27]{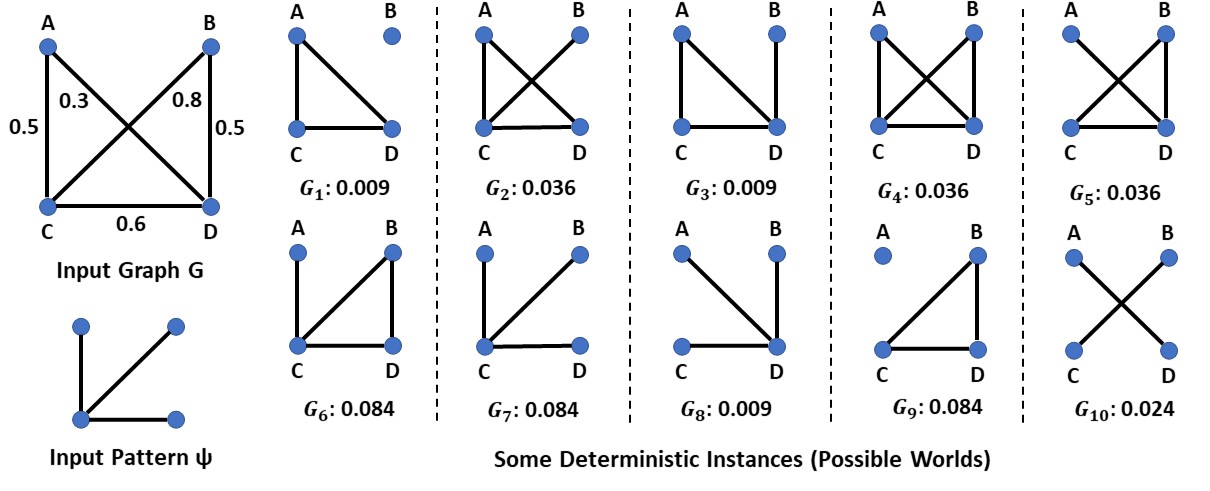}
    \vspace{-3mm}
  \captionof{figure}{\small 3-Clique and pattern densities in an uncertain graph}
  \label{fig:example_3_pattern}
    \vspace{-5mm}
\end{figure}
\begin{exam}
Figure~\ref{fig:example_3_pattern} shows an input uncertain graph and some of its possible worlds.
The subgraphs induced by the node set $\{A,C,D\}$ contain a 3-clique (i.e., triangle) only in possible worlds $G_1$, $G_2$, $G_3$, and $G_4$.
In fact, $\{A,C,D\}$ induces
3-clique densest subgraphs in possible worlds $G_1$, $G_2$, and $G_3$, thus the
3-clique densest subgraph probability of $\{A,C,D\}$ is 0.009 + 0.036 + 0.009 = 0.054.

For the pattern $\psi$, we notice that the subgraphs induced by node set $\{A,B,C,D\}$ contain $\psi$ only in possible worlds $G_2$, $G_3$, $G_4$, $G_5$, $G_6$, $G_7$, and $G_8$.
Moreover, $\{A,B,C,D\}$ induces the $\psi$-densest subgraphs in all these six possible worlds, thus the $\psi$-densest subgraph probability of $\{A,B,C,D\}$ is 0.294.
\label{exam:clique}
\end{exam}
\vspace{-6mm}
\subsection{Problem Formulations}
\label{sec:problem}
We study the novel problem of finding the node set which is the most likely to induce a densest subgraph in an uncertain graph, formally defined as follows.
\begin{problem}
[Most Probable Densest Subgraph (\textsf{MPDS})]
Given an uncertain graph $\mathcal{G} = (V, E, p)$, find the node set $U^* \subseteq V$ that has the highest densest subgraph probability. 
\begin{small}
\begin{equation}
    U^* = \argmax_{U \subseteq V} \tau(U)
\label{eq:mpds}
\vspace{-5mm}
\end{equation}
\end{small}
\label{prob:mpds}
\end{problem}
The {\sf $h$-Clique-MPDS} and {\sf Pattern-MPDS} problems can be defined analogously.
In the following, we provide two other variants of our {\sf MPDS} problem.
First, the user may be interested in exploring more possible choices besides the best node set.
Thus, we provide the top-$k$ variant below.
\begin{problem}
[Top-$k$ Most Probable Densest Subgraphs (\textsf{Top-$k$ MPDSs})]
Given an uncertain graph $\mathcal{G} = (V, E, p)$ and a positive integer $k$, find the top-$k$ distinct node sets $U_1^*, U_2^*, \ldots, U_k^*$
(where each $ U_i^* \subseteq V$, $1 \leq i \leq k$) having the highest densest subgraph probabilities. Formally,
\begin{small}
\begin{align}
   & \tau(U_i^*) \ge \tau(U_{i+1}^*) \quad \forall i \in [1, k) \nonumber \\
   & \tau(U_i^*) \ge \tau(U) \; \forall i \in [1, k] \; \text{\&} \; \forall U \subseteq V \setminus \{U_1^*, \ldots, U_k^*\}
   \label{eq:top_k_mpds}
\end{align}
\end{small}
\label{prob:top_k_mpds}
\end{problem}
\vspace{-6mm}
Second, in large graphs, we observe that the densest subgraph probability of every node set can be quite small, e.g., below $3.91 \times 10^{-5}$ in \textsf{Homo Sapiens}, \textsf{Biomine}, and \textsf{Twitter} (\S~\ref{sec:exp}). In such cases, reporting {\sf MPDS}s contradicts our goal of identifying node sets that can induce densest subgraphs {\em with high probabilities}.
Instead, we find those sets with the highest containment probabilities within a densest subgraph. Such node sets form the ``nuclei'' of various densest subgraphs across different possible worlds of the uncertain graph.
\begin{defn}
[Densest Subgraph Containment Probability]
\label{def:pr}
The densest subgraph containment probability $\gamma(U)$ of $U\subseteq V$ is the sum of the probabilities of all possible worlds $G$ of the uncertain graph $\mathcal{G}=(V,E,p)$ such that $U$ is contained in a densest subgraph of $G$. Formally,
\begin{small}
\begin{equation*}
    \gamma(U) = \sum_{G \sqsubseteq \mathcal{G}} \Pr(G) \times \mathbbm{1} \left[ \exists \, U' \supseteq U \, : \, \rho(G[U']) = \max_{W \subseteq V} \rho(G[W]) \right]
\end{equation*}
\end{small}
\end{defn}
\begin{exam}
Considering the input uncertain graph in Figure~\ref{fig:example_1}, the node set $\{B,D\}$ can induce a densest subgraph with probability 0.42 (in possible worlds $G_4$ and $G_7$). We notice that $\{B,D\}$ is also contained in densest subgraphs of other possible worlds ($G_6$ and $G_8$),
even though $\{B,D\}$ alone does not induce densest subgraphs in these worlds. The overall densest subgraph containment probability of $\{B,D\}$ is 0.7 (due to $G_4$, $G_6$, $G_7$, and $G_8$), which implies that this node set is a critical component in forming densest subgraphs.
\label{exam:nucleus}
\end{exam}

\vspace{-6mm}
\revise{
We aim to find the top-$k$ node sets having the highest densest subgraph containment probabilities. However, there are two caveats. First, a very small node set (as small as a singleton) can also have a high containment probability; however, such sets do not represent meaningful graph communities. Second, if a node set $U$ and one of its supersets $U'$ have equal containment probabilities $\gamma(\cdot)$, then it makes more sense to report $U'$ only (in order to avoid redundancy in the result). To mitigate these, we respectively impose the restrictions that all returned node sets must have a minimum specified size and must be closed w.r.t. $\gamma(\cdot)$. A node set is closed w.r.t. $\gamma(\cdot)$ if none of its supersets has the same value of $\gamma(\cdot)$. We now define our \textsf{NDS} problem.
\begin{problem}
[Top-$k$ Nucleus Densest Subgraphs (\textsf{NDS}s)]
\label{prob:cont_min}
Given an uncertain graph $\mathcal{G} = (V, E, p)$ and positive integers $k$ and $l_m$, let $\mathcal{V}_c^{\geq l_m}$ denote the set of all node sets of size at least $l_m$ that are closed w.r.t. $\gamma(\cdot)$. Find the top-$k$ closed node sets $U_1^*, U_2^*, \ldots, U_k^*$ (where $U_i^* \in \mathcal{V}_c^{\geq l_m}$, $i \in [1, k]$) having the highest densest subgraph containment probabilities. Formally,
\begin{small}
\begin{align}
    & \gamma(U_i^*) \geq \gamma(U_{i+1}^*) \quad \forall i \in [1, k) \nonumber \\
    & \gamma(U_i^*) \geq \gamma(U) \quad \forall i \in [1, k] \; \text{\&} \; \forall U \in \mathcal{V}_c^{\geq l_m} \setminus \left\{ U_1^*, \ldots, U_k^* \right\}
\end{align}
\end{small}
\end{problem}
}
Notice that the {\sf Top-$k$ MPDS} and the {\sf NDS} problems can be analogously extended to their clique and pattern versions.
\subsection{Hardness}
\label{sec:hardness}
%
%
\begin{theor}
\label{th:hard}
Computing the densest subgraph probability of a node set $U$ in an uncertain graph $\mathcal{G} = (V, E, p)$ is $\sharpP$-hard.
\end{theor}
\begin{proof}
We prove by a reduction from the $\sharpP$-hard problem of finding the number of matchings in a graph \cite{V79}. A matching in a deterministic graph $G = (V, E)$ is an edge
set $M \subseteq E$ without any common nodes.

Consider a deterministic graph $G = (V, E)$. This graph is transformed, by adding two new nodes $v_1$ and $v_2$ along with an edge between them, into an uncertain graph $\mathcal{G} = (V \cup \{v_1, v_2\}, E \cup \{(v_1, v_2)\}, p)$, where the probability of each edge is $0.5$, except the new edge $(v_1, v_2)$ which has probability 1. Clearly, this reduction takes $\bigO(|E|)$ time, which is polynomial in the size of $G$. It can be shown that:
    $\bullet$ Any possible world $G' \sqsubseteq \mathcal{G}$ with non-zero probability has $\Pr(G') = (0.5)^{|E|}$.
    $\bullet$ There is a bijection between the set of subsets of $E$ and the set of possible worlds of $\mathcal{G}$ with non-zero probability.
    $\bullet$ The node set $\{v_1, v_2\}$ induces a densest subgraph in a possible world iff every node has degree at most 1 in that world, i.e., the edges in the world excluding $(v_1, v_2)$ form a matching in $G$.
Thus,
\begin{small}
\begin{align*}
    &\tau(\{v_1, v_2\}) = \sum_{G' \sqsubseteq \mathcal{G}} \Pr(G') \times \mathbbm{1} \left[ \{v_1, v_2\} = \argmax_{W \subseteq V} \rho(G'[W]) \right] \\
    &= \sum_{G' \sqsubseteq \mathcal{G}} \Pr(G') \times \mathbbm{1} \left[ \text{each node has degree at most 1 in }G' \right] \\
    &= (0.5)^{|E|} \sum_{G' \sqsubseteq \mathcal{G} : \Pr(G') \neq 0}\!\!\!\!\!\! \mathbbm{1} \left[ \text{each node has degree at most 1 in }G' \right] \\
    &= (0.5)^{|E|} \sum_{M \subseteq E} \mathbbm{1} \left[ \text{M is a matching in $G$} \right]
\end{align*}
\end{small}
The sum in the last line above is the number of matchings in $G$. Thus, a solution to our problem on $\mathcal{G}$ provides a solution to the matching counting problem on $G$.
\end{proof}

Since the computation of $\tau(U)$, for a given $U$, is $\sharpP$-hard, the computations of its generalizations $\tau_h(U)$ and $\tau_\psi(U)$ are also $\sharpP$-hard.
\revise{Thus, finding the node sets with the top-$k$ densest subgraph probabilities, as well as computing the {\sf NDS}, are also very difficult.}
{\em Given such computational challenges, we design approximate algorithms, with end-to-end accuracy guarantees, to find the most probable densest subgraphs in an uncertain graph, based on various graph density notions}. 
\section{Approximate Solutions for Densest Subgraphs}
\label{sec:mpds_algo}
In this section, we develop approximation algorithms for detecting the top-$k$ \textsf{MPDS}s, along with end-to-end theoretical accuracy guarantees.
Our {\bf technical contributions} are as follows: {\bf (1)} We design {\em novel approximation methods to compute the top-$k$ \textsf{MPDS}s in an uncertain graph for all density notions: edge (\S~\ref{sec:top_k_mpds_algorithm}), clique (\S~\ref{sec:clique}), and pattern (\S~\ref{sec:pattern})}.
{\bf (2)} As building blocks of the algorithms for clique and pattern densities, we also design {\em novel algorithms to discover all clique and pattern densest subgraphs in a deterministic graph}\footnote{\revise{\scriptsize Due to the additional nodes for cliques/patterns in flow network construction, the definitions and theorems do not trivially follow  \cite{CQ20}, more details can be found in the remark, \S~\ref{sec:clique}}}\footnote{\revise{\scriptsize Our empirical study (\S~\ref{sec:allvs1}) validates that considering {\em all} densest subgraphs can significantly outperform (e.g., up to 20$\times$ in {\sf LastFM}) the method that considers only one randomly chosen densest subgraph.}}(\S~\ref{sec:clique} and \S~\ref{sec:pattern}). 
{\bf (3)}  Additionally, we use these methods to design {\em approximation algorithms, with end-to-end theoretical quality guarantees,
for computing the corresponding \textsf{NDS} (\S~\ref{sec:nucleus})}.

\vspace{-1mm}
\subsection{\textsf{Top-$k$ MPDS}: Approximate Algorithm}
\label{sec:top_k_mpds_algorithm}

\begin{algorithm}[t]
	\renewcommand{\algorithmicrequire}{\textbf{Input:}}
	\renewcommand{\algorithmicensure}{\textbf{Output:}}
	\scriptsize
	\caption{\small \textsf{Top-$k$ MPDS} estimation}
	\label{alg:exact}
	\begin{algorithmic}[1]
		\REQUIRE Uncertain graph $\mathcal{G} = (V,E,p)$, positive integer $k$, and number of samples $\theta$
		\ENSURE (Approximate) \textsf{Top-$k$ MPDS}
		\FORALL{$U \subseteq V$}
		\STATE $\widehat{\tau}(U) \gets 0$
		\ENDFOR
		\FOR{$i = 1$ \TO $\theta$}
		\STATE Sample a possible world $G \sqsubseteq \mathcal{G}$
		\STATE $S \gets $ All densest subgraphs in $G$ via \cite{CQ20} \label{line:all_densest}
		\FORALL{$U \in S$}
		\STATE $\widehat{\tau}(U) \gets \widehat{\tau}(U) + \frac{1}{\theta}$
		\ENDFOR
		\ENDFOR
		\RETURN Top-$k$ $U$'s having the highest $\widehat{\tau}(U)$ 
	\end{algorithmic}
\end{algorithm}

The proposed solution (Algorithm \ref{alg:exact}) runs $\theta$ independent
iterations as follows: Sample a possible world $G \sqsubseteq \mathcal{G}$ and find {\em all} the node sets
inducing the densest subgraphs in $G$ (Line \ref{line:all_densest}).
$\widehat{\tau}(U)$ denotes the estimated densest subgraph probability, which is computed as the average frequency that a node set $U$
induces a densest subgraph across $\theta$ rounds. Finally, we return the top-$k$ node sets having the highest $\widehat{\tau}(\cdot)$.

\vspace{-2mm}
\begin{lem}
\label{lem:unbiased}
$\widehat{\tau}(U)$ is an unbiased estimator for $\tau(U)$. Formally,
    $\mathbbm{E}[\widehat{\tau}(U)] = \tau(U)$.
\end{lem}
\begin{proof}
Let $X_i(U)$ be a binary random variable denoting whether $U$ induces a densest subgraph in the $i^{th}$ possible world; thus $\widehat{\tau}(U) = \frac{1}{\theta} \sum_{i=1}^\theta X_i(U)$. Clearly, $\mathbbm{E}[X_i(U)] = \Pr \left( X_i[U] = 1 \right) = \tau(U)$, and hence $\mathbbm{E}[\widehat{\tau}(U)] = \tau(U)$.
\end{proof}
\vspace{-2mm}

The unbiasedness ensures that the estimated $\widehat{\tau}(U)$ goes closer to the true value $\tau(U)$ as the sample size $\theta$ increases.

The technique in \cite{CQ20} for computing all densest subgraphs in a deterministic graph (Line \ref{line:all_densest}) involves reducing the graph to its $\lceil \Tilde{\rho} \rceil$-core \cite{BZ03}, where $\Tilde{\rho}$ is a lower bound on the
maximum edge density $\rho_e^*$ of any subgraph. This is followed by computing $\rho_e^*$ using the
state-of-the-art Goldberg's algorithm \cite{G84}. During each iteration of its binary search, the Goldberg's algorithm tries to find a subgraph with density larger than a guessed value $\alpha$ by computing the minimum cut in a flow network parameterized by $\alpha$.
Once $\rho_e^*$ is found, all densest subgraphs are enumerated by
traversing the strongly connected components (SCCs) in the residual graph (under a maximum flow)
of the flow network with $\alpha = \rho_e^*$; the details can be found in \cite{CQ20} and in the example below.

\vspace{-2mm}
\begin{exam}
We shall compute all densest subgraphs in a possible world $G$ (Figure~\ref{fig:example_5_edge}(b)) of an uncertain graph $\mathcal{G}$ ((Figure~\ref{fig:example_5_edge}(a)). A flow network $G_{\alpha}$ (Figure~\ref{fig:example_5_edge}(c)) is constructed as follows. (1) Add a source node $s$ and a sink node $t$. (2) If an edge $(u,v)$ exists in $G$, add an edge from $u$ to $v$ and one from $v$ to $u$ in $G_{\alpha}$, both with capacity 1. (3) Add an edge from $s$ to each node $v$ in $G$ with capacity equal to the degree of $v$ in $G$. (4) Add an edge from each node in $G$ to $t$ with capacity $2\alpha$. Goldberg's algorithm \cite{G84} conducts a binary search with $[0,m]$ as the initial range of $\alpha$. In each iteration, it guesses an $\alpha$ and computes the maximum flow (minimum cut) in the flow network.
Once terminated, 
the optimal density is assigned to be $\rho^*=\alpha$.
In this example, we get $\rho^*=1$ and a densest subgraph $\{A,B,C,D\}$ (which corresponds to the minimum cut).
To find the other densest subgraphs, we create the residual graph ((Figure~\ref{fig:example_5_edge}(d)) by removing all the edges with zero residual capacity.
The densest subgraphs are enumerated by exploring the SCCs in this residual graph. In this case, we find one additional densest subgraph $\{B,C,D\}$.
\label{example:edge_density_flow}
\vspace{-1mm}
\end{exam}

\vspace{-6mm}
\spara{Space complexity.} The majority of the memory is consumed by the flow network in each iteration of
the binary search, and hence the overall space complexity is $\bigO(m + n)$.

\begin{figure}
  \centering
  \vspace{-2mm}
  \includegraphics[scale=0.55]{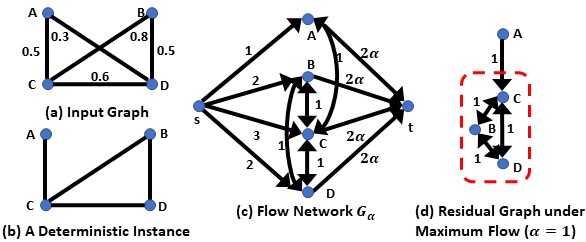}
    \vspace{-2mm}
  \captionof{figure}{\small Finding all densest subgraphs in a possible world
  }
  \label{fig:example_5_edge}
    \vspace{-6mm}
\end{figure}

\vspace{-1mm}
\spara{Time complexity.} Recall that we sample $\theta$ possible worlds. The computation of the $\lceil \Tilde{\rho} \rceil$-core of each possible world takes $\bigO(m)$ time \cite{BZ03}. Let $n_c$ and $m_c$ denote the number of nodes and edges respectively in the $\lceil \Tilde{\rho} \rceil$-core of a sampled possible world $G \sqsubseteq \mathcal{G}$. As shown in \cite{CQ20},
the computation of $\rho_e^*$ takes $\bigO \left( n_c m_c \log \left( \frac{n_c^2}{m_c} \right) \right)$ time,
while the enumeration of the densest subgraphs takes $\bigO(L)$ time per subgraph, where $L$ is the
number of nodes in that densest subgraph (note that $L \leq n$). Denoting respectively by $n_c^*$, $m_c^*$, $d^*$,
and $L^*$ the maximum number of nodes, edges, densest subgraphs, and nodes in a densest subgraph
in any possible world, the overall time complexity of our algorithm is
$\bigO \left( \theta \left( m + n_c^* m_c^* \log \left( \frac{n_c^{*2}}{m_c^*} \right) + d^*L^* \right) \right)$.
Practically, $n_c^* \ll n$, $m_c^* \ll m$ and $L^* \ll n$, as validated in our experiments (\S~\ref{sec:exp}).
For instance, $\frac{n_c^*}{n} = 2.45 \times 10^{-4}$, $\frac{L^*}{n} = 1.8 \times 10^{-4}$,
$\frac{m_c^*}{m} = 5.43 \times 10^{-3}$, and $d^* = 1$
in our large \textsf{Twitter} dataset, making Algorithm~\ref{alg:exact} efficient even with large-scale graphs.



\spara{Accuracy guarantee.}
We theoretically analyze the sample size $\theta$ to return the true top-$k$ node sets with a high probability.
First, the true top-$k$ node sets are proved likely to be included in the candidate node sets\footnote{\scriptsize{A node set $U$ is called a candidate node set if its estimated densest subgraph probability $\widehat{\tau}(U)$ after
$\theta$ rounds is non-zero, i.e., $U$ induces a densest subgraph in at least one of $\theta$ possible world samples of $\mathcal{G}$.}}
after $\theta$ rounds (Theorem~\ref{th:cv}). Second, we ensure that those sets indeed have the $k$ largest {\em estimated} densest subgraph probabilities (Theorem~\ref{th:topk}).
\begin{theor}
Let $V_1, \ldots, V_{k}$ denote the true top-$k$ node sets having the highest densest subgraph probabilities, and let $CV$ denote the set of candidate node sets after $\theta$ rounds. Then,
\vspace{-1mm}
\begin{small}
\begin{equation}
    \Pr \left( \left\{V_1, \ldots, V_k \right\} \subseteq CV \right) \geq 1 - \sum_{i=1}^k \left( 1 - \tau(V_i) \right)^\theta
    \label{eq:cv}
    \vspace{-4mm}
\end{equation}
\end{small}
\label{th:cv}
\end{theor}
\begin{proof} \renewcommand\qedsymbol{}
Since we compute all densest subgraphs in each of $\theta$ independently sampled possible worlds, $\Pr \left( V_i \notin CV \right) = \left( 1 - \tau(V_i) \right)^\theta \, \forall i \in \{1, \ldots, k\}$. From the union bound,
\begin{footnotesize}
\begin{align*}
    &\Pr \left( \left\{V_1, \ldots, V_k \right\} \subseteq CV \right)
    = 1 - \Pr \left( \exists{i \in \{1, \ldots, k\}} \, : \, V_i \notin CV \right) \\
    & \geq 1 - \sum_{i=1}^k \Pr \left( V_i \notin CV \right) = 1 - \sum_{i=1}^k \left( 1 - \tau(V_i) \right)^\theta
\end{align*}
\end{footnotesize}
\end{proof}
\vspace{-4mm}
From Theorem \ref{th:cv}, if the densest subgraph probabilities of the true top-$k$ node sets and the sample size $\theta$ are reasonably large, they are all highly likely to be included in the candidate node set, which is necessary for all of them to be returned.
\begin{theor}
Let $V_1, \ldots, V_{k+1}$ denote the true top-$(k+1)$ node sets having the highest densest subgraph probabilities, and let $CV$ denote the set of candidate node sets after $\theta$ rounds. Define $mid = \frac{1}{2} \left[ \tau\left(V_k\right) + \tau\left(V_{k+1}\right) \right]$ and
\begin{small}
\begin{equation}
    d_U =
    \begin{cases}
    \tau(U) - mid & \text{if} \, \, U \in \{V_1, \ldots, V_k\} \\
    mid - \tau(U) & \text{otherwise}
    \end{cases}
\end{equation}
\end{small}
Then, the probability that $V_1, \ldots, V_k$ are returned by Algorithm~\ref{alg:exact} is at least
\begin{small}
\begin{equation}
    \small \left[ 1 - \sum_{i=1}^k \left( 1 - \tau(V_i) \right)^\theta \right] \left[ 1 - \sum_{U \in CV} \exp \left( -2d_U^2\theta \right) \right] \label{eq:prf}
\end{equation}
\end{small}
\label{th:topk}
\end{theor}
\vspace{-6mm}
\begin{proof}
Let $CV'=CV \setminus \{V_1, \ldots, V_k\}$. We have:
\begin{footnotesize}
\begin{align}
    &\Pr(V_1, \ldots V_k \text{ are returned}) \geq \Pr(\{V_1, \ldots, V_k\} \subseteq CV) \times \nonumber \\
    &\quad \Pr \left( \left(\bigwedge_{U \in \{V_1, \ldots, V_k\}} \!\!\!\!\!\! \widehat{\tau}(U) > mid  \right)\wedge \left(\bigwedge_{U \in CV'} \!\!\!\!\!\! \widehat{\tau}(U) < mid \right) \right)\label{eq:pr1}
\end{align}
\end{footnotesize}
Now, using the union bound and Hoeffding's inequality,
\begin{footnotesize}
\begin{align}
    &\Pr \left( \left(\bigwedge_{U \in \{V_1, \ldots, V_k\}} \widehat{\tau}(U) > mid \right)\wedge \left(\bigwedge_{U \in CV'} \widehat{\tau}(U) < mid \right) \right) \nonumber \\
    &= 1 - \Pr \left( \left(\bigvee_{U \in \{V_1, \ldots, V_k\}} \!\!\! \widehat{\tau}(U) \leq mid \right) \vee \left(\bigvee_{U \in CV'}  \widehat{\tau}(U) \geq mid \right)\right) \nonumber \\
    &\geq 1 - \sum_{U \in \{V_1, \ldots, V_k\}}  \Pr \left( \widehat{\tau}(U) \leq mid \right) - \sum_{U \in CV'} \!\! \Pr \left( \widehat{\tau}(U) \geq mid \right) \nonumber \\
    &= 1 - \!\!\!\! \sum_{U \in \{V_1, \ldots, V_k\}} \!\!\!\!\!\!\!\!\!\!\!\!\! \Pr \left( \widehat{\tau}(U) - \tau(U) \leq -d_U \right) - \!\! \sum_{U \in CV'} \!\!\!\!\! \Pr \left( \widehat{\tau}(U) - \tau(U) \geq d_U \right) \nonumber \\
    &\geq 1 - \sum_{U \in CV} \exp \left( -2d_U^2\theta \right) \label{eq:pr2}
\end{align}
\end{footnotesize}
Finally, plugging \eqref{eq:cv} and \eqref{eq:pr2} into \eqref{eq:pr1}, we obtain \eqref{eq:prf}.
\end{proof}
%

From Theorem \ref{th:topk}, if the densest subgraph probabilities of the true top-$k$ node sets are reasonably large
in contrast to the others
and if the sample size $\theta$ is sufficiently large, the true top-$k$ sets are returned by Algorithm \ref{alg:exact} with a high probability.

\vspace{-1mm}
\spara{Remarks.} \textbf{(1)} Notice that the algorithmic framework and accuracy guarantees can be easily adapted to solve the top-$k$ {\sf Clique-MPDS} (resp. {\sf Pattern-MPDS}) problems. However, we need to develop an efficient algorithm for detecting {\em all} clique (resp. pattern)-densest subgraphs in each sampled possible world (Line \ref{line:all_densest} of Algorithm \ref{alg:exact}), which is our {\em novel technical contribution} in \S~\ref{sec:clique} (resp. \S~\ref{sec:pattern}).
\textbf{(2)} 
The analyses in this section are based on the assumption that we use {\em Monte Carlo} to sample possible worlds. There also exist other sampling techniques such as {\em  Lazy Propagation} \cite{LFZT17} and {\em Recursive Stratified Sampling} \cite{RSS16}.
We empirically show that, for our problem, these three sampling strategies result in {\em similar sample sizes} $\theta$ and have {\em comparable running times}, while {\em Monte Carlo} consumes {\em much less memory} (\S~\ref{sec:large_graph_results}).

\subsection{\textsf{$h$-Clique-MPDS}: Approximate Algorithm}
\label{sec:clique}
Inspired by \cite{CQ20}, we develop a novel, exact, and efficient solution to discover {\em all} clique-densest subgraphs in a deterministic graph (Algorithm \ref{alg:clique_densest}). {\em This is a novel problem, and
no existing work has
studied it. Therefore, Algorithm \ref{alg:clique_densest} is one of our novel technical contributions}.
In the following, we first revisit the concepts of clique degree (Definition \ref{def:clique_deg}) and clique-based core (Definition \ref{def:core}) in deterministic graphs. Then, we illustrate the technical details and the intuitions of our algorithm, together with a running example. Finally, we provide theoretical analyses about its efficiency and correctness.

\begin{algorithm}[t]
	\renewcommand{\algorithmicrequire}{\textbf{Input:}}
	\renewcommand{\algorithmicensure}{\textbf{Output:}}
	\scriptsize
	\caption{\small Find all clique-densest subgraphs} \label{alg:clique_densest}
	\begin{algorithmic}[1]
		\REQUIRE Deterministic graph $G = (V,E)$, positive integer $h$
		\ENSURE All $h$-clique densest subgraphs in $G$
		\STATE $\Tilde{\rho} \gets $ Density returned by the peeling method \cite{T15, FYCLX19} \label{line:lb}
		\STATE $G_c \gets (\lceil \Tilde{\rho} \rceil, h)$-core of $G$ \cite{FYCLX19} \label{line:core}
		\STATE $\Lambda \gets $ All $(h-1)$-cliques contained in $h$-cliques in $G_c$ \cite{DBS18} \label{line:all_instance}
		\STATE $\rho_h^* \gets \max_{S \subseteq V} \rho_h(S)$ \cite{SDCS20} \label{line:val}
		\STATE $\mathcal{H} \gets $ Algorithm \ref{alg:flow} ($G_c$, $\Lambda$, $\rho_h^*$) \cite{MPPTX15, FYCLX19} \label{line:begin_enum_all}
		\STATE $f^* \gets $ Maximum flow in $\mathcal{H}$ \label{line:max_flow}
		\STATE $\mathcal{C} \gets $ SCCs of the residual graph $\mathcal{H}_{f^*}$, excluding those of $s$ and $t$ \label{line:scc}
		\RETURN Algorithm \ref{alg:enum} ($\emptyset$, $\mathcal{C}$, $V$) \label{line:end_enum_all}
	\end{algorithmic}
\end{algorithm}

\vspace{-1mm}
\begin{defn}
[$h$-Clique Degree \cite{FYCLX19}]
\label{def:clique_deg}
The $h$-clique degree ($h \geq 2$) of a node $v$ in a deterministic graph $G$, denoted by $deg_G(v, h)$, is the number of $h$-cliques in $G$ containing $v$.
\end{defn}

\vspace{-2mm}
\begin{defn}
[$(k, h)$-Core \cite{FYCLX19}]
\label{def:core}
Given a deterministic graph $G$ and two integers $h \geq 2$ and $k \geq 0$, the $(k, h)$-core of $G$, denoted by $\mathcal{R}_k$, is the largest subgraph of $G$ such that, for every node $v$ in $\mathcal{R}_k$, $deg_{\mathcal{R}_k} (v, h) \geq k$.
\end{defn}
\vspace{-2mm}

Armed with these definitions, we proceed to the details of Algorithm \ref{alg:clique_densest}, which consists of two general steps: {\bf (1)} The technique in \cite{SDCS20} is applied to compute the maximum density of any subgraph of $G$. {\bf (2)} A flow network $\mathcal{H}$ is constructed following \cite{MPPTX15, CQ20}. The SCCs of the residual graph under a maximum flow in $\mathcal{H}$ indicate all densest subgraphs of $G$.

In Line \ref{line:lb}, it runs the peeling method of \cite{T15}, which iteratively removes the node with the smallest $h$-clique degree and returns the maximum density among all the resultant subgraphs, denoted by $\Tilde{\rho}$. Then, in Line \ref{line:core}, it replaces $G$ with its $(\lceil \Tilde{\rho} \rceil, h)$-core, i.e., the subgraph induced by those nodes which have $h$-clique degree at least $\Tilde{\rho}$. After that, Line \ref{line:all_instance} computes the set $\Lambda$ of all $(h-1)$-cliques contained in $h$-cliques in $G$, which are enumerated using the method in \cite{DBS18}. Line \ref{line:val} computes $\rho_h^*$, the maximum $h$-clique density of any subgraph, by the method in \cite{SDCS20}, which iteratively computes a (predicted) clique-densest subgraph via optimizing a convex program, till the computed subgraph is deemed to be indeed clique-densest.
After that, the clique density of the computed subgraph is returned.

Next, Line \ref{line:begin_enum_all} constructs a flow network $\mathcal{H}$ following \cite{MPPTX15, FYCLX19}, which contains one node for each $(h-1)$-clique in $\Lambda$ and one for each node in $V$, in addition to a source node $s$ and a sink node $t$ (see Appendix \ref{sec:clique_proof} for the pseudocode). Once $\mathcal{H}$ is constructed, Algorithm \ref{alg:clique_densest} computes a maximum flow $f^*$ in $\mathcal{H}$ (Line \ref{line:max_flow}) and then identifies the strongly connected components (SCCs) of the residual graph $\mathcal{H}_{f^*}$ of $\mathcal{H}$ under $f^*$ (Line \ref{line:scc}) after removing the edges with zero residual capacity.

\begin{algorithm}[t]
	\renewcommand{\algorithmicrequire}{\textbf{Input:}}
	\renewcommand{\algorithmicensure}{\textbf{Output:}}
	\scriptsize
	\caption{\small Enumerate all clique densest subgraphs} \label{alg:enum}
	\begin{algorithmic}[1]
		\REQUIRE Component sets $\mathcal{C}_1$ and $\mathcal{C}_2$, node set $V$
		\ENSURE All clique-densest subgraphs
		\STATE $R \gets \emptyset$
		\IF{$\mathcal{C}_1 \neq \emptyset$}
		\STATE $R \gets R \cup \left( \bigcup_{C \in \mathcal{C}_1 \cup des(\mathcal{C}_1)} C \cap V \right)$
		\ENDIF
		\FORALL{$C \in \mathcal{C}_2$}
		\IF{$C \cap V \neq \emptyset$}
		\STATE $\mathcal{C}_2 \gets \mathcal{C}_2 \setminus \{C\}$ \label{line:rem}
		\STATE $S \gets $ Algorithm \ref{alg:enum} ($\mathcal{C}_1 \cup \{C\}$, $\mathcal{C}_2 \setminus (des(C) \cup anc(C))$, $V$) \label{line:rec}
		\STATE $R \gets R \cup S$
		\ENDIF
		\ENDFOR
		\RETURN $R$
	\end{algorithmic}
\end{algorithm}

Finally, in Line \ref{line:end_enum_all},
Algorithm \ref{alg:enum} enumerates one densest subgraph of $G$ for every value of $C \cup des(C)$, where $C$ is an SCC of $\mathcal{H}_{f^*}$, and the set $des(C)$ (resp. $anc(C)$) denotes the set of SCCs having a directed path from (resp. to) $C$ in the SCC graph of $\mathcal{H}_{f^*}$. The detailed theoretical analyses for Algorithms \ref{alg:clique_densest} and \ref{alg:enum} are given in Appendix \ref{sec:clique_proof}.

\begin{figure}
  \centering
  \includegraphics[scale=0.32]{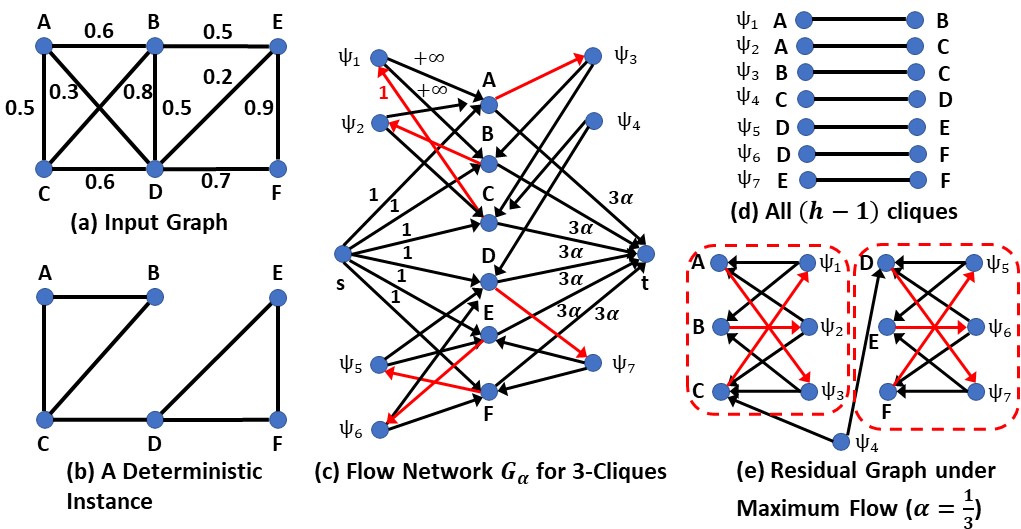}
   \vspace{-1mm}
  \captionof{figure}{\small Finding all $3$-clique densest subgraphs in a possible world
  }
  \label{fig:example_6_clique}
    \vspace{-5mm}
\end{figure}

\vspace{-2mm}
\begin{exam}
We shall compute all $h$-clique densest subgraphs ($h = 3$) in a possible world $G$ (Figure~\ref{fig:example_6_clique}(b)) of an uncertain graph $\mathcal{G}$ ((Figure~\ref{fig:example_6_clique}(a)). The maximum 3-clique density $\rho_h^*$ of any subgraph of $G$ can be easily computed as $\frac{1}{3}$. A flow network $G_{\alpha}$ is constructed as in Figure~\ref{fig:example_6_clique}(c). Instead of directly adding edges between nodes as for edge density (Example~\ref{example:edge_density_flow}), a new set of nodes representing the $(h-1)$-cliques is added. A $(h-1)$-clique node $\psi_i$ has a directed edge to each node contained in this $(h-1)$-clique, with infinite capacity. If a node $v$ forms an $h$-clique with the $(h-1)$-clique $\psi_i$, a directed edge from $v$ to $\psi_i$ is added with capacity 1. For simplicity, we only show the capacities of the edges entering and leaving $\psi_1$ in Figure~\ref{fig:example_6_clique}(c). In this example, the $(h-1)$-cliques are all edges in the possible world, as shown in Figure~\ref{fig:example_6_clique}(d). After plugging in $\rho_h^* = \frac{1}{3}$ and computing the maximum flow, we identify a 3-clique densest subgraph $\{A,B,C,D,E,F\}$, and obtain a residual graph as shown in Figure~\ref{fig:example_6_clique}(e). The remaining 3-clique densest subgraphs are $\{A,B,C\}$ and $\{D,E,F\}$. Each of them corresponds to an SCC of the residual graph.
\label{exam:clique_densest}
\end{exam}

\vspace{-1mm}
\spara{Remark.} Our Algorithm \ref{alg:clique_densest} for computing all clique-densest subgraphs has similarities
to the method of computing all edge-densest subgraphs in \cite{CQ20}.
\revise{{\em However, there are also major differences from \cite{CQ20}:} 
As demonstrated in Example \ref{exam:clique_densest}, our flow network $\mathcal{H}$ has one additional node for each $(h-1)$-clique in $\Lambda$ (in addition to one for each node in $V$ as in \cite{CQ20}), and the edges are from nodes in $V$ to nodes in $\Lambda$ (unlike \cite{CQ20}, where the edges only exist between nodes in $V$). Thus, some non-trivial additions\footnote{\revise{\scriptsize In Appendix \ref{sec:clique_proof}, Definitions~\ref{def:ic}-\ref{def:dcc} and the proofs of Lemmas~\ref{lem:ic_dcc}-\ref{lem:dcc_densest} are different from their counterparts in \cite{CQ20}. Moreover, Lemma~\ref{lem:tc} is newly derived and serves as a critical reason for not considering SCCs of the sink node.}} need to be incorporated into the definitions and proofs in \cite{CQ20} to prove the correctness of our Algorithms \ref{alg:clique_densest} and \ref{alg:enum}. {\em This forms one of our novel technical contributions}.}

\vspace{-1mm}
\spara{Space complexity.} Let $G_c$ be the $(\lceil \Tilde{\rho} \rceil, h)$-core in Lines \ref{line:lb}-\ref{line:core} of Algorithm \ref{alg:clique_densest}, with $n_c$ and $m_c$ the corresponding node and edge counts respectively. As each $h$-clique of $G_c$ contains $h$ distinct $(h-1)$-cliques, $|\Lambda| = \bigO \left( h \mu_h(G_c) \right)$, where $\mu_h(G_c)$ is the number of $h$-cliques in $G_c$. Thus, the number of nodes in $\mathcal{H}$ is $\bigO \left( n_c + h \mu_h(G_c) \right)$ and the number of edges is $\bigO \left( n_c + \sum_{v \in V} deg_{G_c}(v, h) + (h - 1) h \mu_h(G_c) \right) = \bigO \left( n_c + h^2 \mu_h(G_c) \right)$.
Since these consume the most memory, the space complexity of Algorithm \ref{alg:clique_densest} is $\bigO \left( n_c + h^2 \mu_h(G_c) \right)$.

When we adapt Algorithm \ref{alg:exact} for \textsf{$h$-Clique-MPDS}, in addition to the memory required for storing the uncertain graph,
the majority of the memory is consumed by Line \ref{line:all_densest}, which invokes Algorithm \ref{alg:clique_densest}. Denoting by $\mu_h^*$ the maximum number of $h$-cliques in the $(\lceil \Tilde{\rho} \rceil, h)$-core of any possible world of $\mathcal{G}$, the overall space complexity of our method is $\bigO \left( m + n + h^2 \mu_h^* \right)$.

\vspace{-1mm}
\spara{Time complexity.} For Lines \ref{line:lb}-\ref{line:all_instance} in Algorithm \ref{alg:clique_densest}, the major step is enumerating all $h$-cliques in $G$, which takes $\bigO \left(hm \left( \frac{1}{2} \cdot \mu_h(G) \right)^{h-2} \right)$ time, where $\mu_h(G)$ is the number of $h$-cliques in $G$ \cite{DBS18}. 
%
%
If $T$ denotes the number of iterations involved in computing $\rho_h^*$ using \cite{SDCS20}, the running time of Line \ref{line:val} consists of optimizing the convex program and computing the maximum flow in each iteration \cite{SDCS20}, and requires $\bigO \left( 2^T h \mu_h(G_c) + T \left( n_c h \mu_h(G_c) + n_c^3 \right) \right)$ time, where $n_c$ is as defined above.
Lines \ref{line:begin_enum_all}-\ref{line:max_flow} can be done in $\bigO \left( \left( n_c h \mu_h(G_c) + n_c^3 \right) \right)$ time \cite{FYCLX19}.
%
%
Line \ref{line:scc} takes $\bigO \left( n_c + h^2 \mu_h(G_c) \right)$ time, as it involves finding the SCCs of $\mathcal{H}_{f^*}$, whose node and edge counts are stated above.
%
%
In Line \ref{line:end_enum_all}, each densest subgraph is enumerated exactly once (as proved in Appendix \ref{sec:clique_proof}), and each subgraph enumeration takes time linear in the corresponding number of nodes.

Denoting respectively by $n_c^*$, $d^*$, $L^*$, and $T^*$, the maximum number (in the $(\lceil \Tilde{\rho} \rceil, h)$-core of any possible world) of nodes, densest subgraphs, nodes in a densest subgraph
and iterations of \cite{SDCS20} to compute $\rho_h^*$, and by $\mu_h^*$ the maximum number of $h$-cliques in any possible world, the overall time complexity of our method is \small{$\bigO \left( \theta \left( h m \left( \frac{1}{2} \cdot \mu_h^* \right)^{h-2} + 2^{T^*} h \mu_h^* + T^* \left( n_c^* h \mu_h^* + n_c^{*3} \right) + d^* L^* \right) \right)$}. \normalsize Practically, $n_c^* \ll n$ and $T^*$ is very small, as validated in our experiments (\S~\ref{sec:exp}). For instance, $T^* = 11$, $\frac{n_c^*}{n} = 4.731 \times 10^{-4}$, $\frac{L^*}{n} = 2.049 \times 10^{-5}$,  $\frac{\mu_h^*}{m} = 0.181$, and $d^* = 1$ for our large \textsf{Twitter} dataset, making our algorithm quite efficient with large-scale graphs.

\vspace{-1mm}
\spara{Accuracy guarantee.} We show (briefly in the following and in details in Appendix \ref{sec:clique_proof}) that Algorithm \ref{alg:clique_densest} correctly computes all clique-densest subgraphs in a deterministic graph. Thus, our overall accuracy guarantees for finding the top-$k$
\textsf{$h$-Clique-MPDS}s in an uncertain graph remain the same as in \S \ref{sec:top_k_mpds_algorithm}.

\vspace{-1mm}
\begin{theor}
Algorithm \ref{alg:clique_densest} enumerates each $h$-clique densest subgraph of a deterministic graph $G$ exactly once.
\end{theor}
\vspace{-4mm}
\begin{proof}
[Proof Sketch]
Lines \ref{line:lb}-\ref{line:core} of Algorithm \ref{alg:clique_densest} are justified by the fact that the densest subgraph is in the $(\lceil \rho_h^* \rceil, h)$-core, and hence in the $(\lceil \Tilde{\rho} \rceil, h)$-core \cite{FYCLX19}.
The correctness of Line \ref{line:val} follows from \cite{SDCS20}.
The main idea behind the remaining lines is that 
all $h$-clique densest subgraphs of $G$ are hidden in the SCCs of the residual graph $\mathcal{H}_{f^*}$ of $\mathcal{H}$ under a maximum flow $f^*$.
Any $h$-clique densest subgraph of $G$ constitutes a minimum $s$-$t$ cut in $\mathcal{H}$.
Thus, this subgraph has no outgoing edge in $\mathcal{H}_{f^*}$.
\end{proof}

\vspace{-2mm}
\subsection{\textsf{Pattern-MPDS}: Approximate Algorithm}
\label{sec:pattern}

Algorithm \ref{alg:pattern_densest} for finding {\em all} pattern-densest subgraphs (w.r.t. a given pattern $\psi = (V_\psi, E_\psi)$) is inspired by \cite{CQ20}.
As earlier, {\em finding all pattern-densest subgraphs in a deterministic graph is a novel problem, and Algorithm \ref{alg:pattern_densest}
is a novel contribution.}

Algorithm \ref{alg:pattern_densest} is similar to Algorithm \ref{alg:clique_densest} (\S~\ref{sec:clique}); in fact, the $(k, h)$-core (Definition \ref{def:core}) can be easily extended to the $(k, \psi)$-core (Line \ref{line:core_pat}), and $\rho_\psi^*$ can be computed in Line \ref{line:val_pat} by extending \cite{SDCS20}. \revise{In Line \ref{line:all_instance}, $\Lambda$ refers to the set of all $\psi$-instances in $G$ instead of cliques, which can be enumerated using the method in \cite{QZC17}.
Moreover, the flow network $\mathcal{H}$ (constructed as in \cite{FYCLX19}; the pseudocode is included in Appendix \ref{sec:pattern_proof}) contains one node for each group of $\psi$-instances with a common node set, instead of one for each instance, in order to reduce the memory footprint and running time.} This forms the main difference between our method and the one in \cite{CQ20}, resulting in some non-trivial additions which constitute one of our novel contributions (see Appendix \ref{sec:pattern_proof}).

\begin{algorithm}[t]
	\renewcommand{\algorithmicrequire}{\textbf{Input:}}
	\renewcommand{\algorithmicensure}{\textbf{Output:}}
	\scriptsize
	\caption{\small Find all pattern-densest subgraphs} \label{alg:pattern_densest}
	\begin{algorithmic}[1]
		\REQUIRE Deterministic graph $G = (V,E)$, pattern $\psi = (V_\psi, E_\psi)$
		\ENSURE All pattern-densest subgraphs w.r.t. $\psi$ in $G$
		\STATE $\Tilde{\rho} \gets $ Density returned by the peeling method \cite{FYCLX19} \label{line:lb_pat}
		\STATE $G_c \gets (\lceil \Tilde{\rho} \rceil, \psi)$-core of $G$ \cite{FYCLX19} \label{line:core_pat}
		\STATE $\Lambda \gets $ All $\psi$-instances in $G_c$ \cite{QZC17} \label{line:all_instance_pat}
		\STATE $\rho_\psi^* \gets \max_{S \subseteq V} \rho_\psi(S)$ \cite{SDCS20} \label{line:val_pat}
		\STATE $\mathcal{H} \gets $ Algorithm \ref{alg:flow_pat} ($G_c$, $\psi$, $\Lambda$, $\rho_h^*$) \cite{FYCLX19} \label{line:begin_enum_all_pat}
		\STATE $f^* \gets $ Maximum flow in $\mathcal{H}$ \label{line:max_flow_pat}
		\STATE $\mathcal{C} \gets $ SCCs of the residual graph $\mathcal{H}_{f^*}$, excluding those of $s$ and $t$ \label{line:scc_pat}
		\RETURN Algorithm \ref{alg:enum} ($\emptyset$, $\mathcal{C}$, $V$) \label{line:end_enum_all_pat}
	\end{algorithmic}
\end{algorithm}

\vspace{-1mm}
\spara{Accuracy guarantee.} By a similar analysis as in Appendix \ref{sec:clique_proof}, it can be shown that Algorithm \ref{alg:pattern_densest} correctly computes {\em} pattern-densest subgraphs w.r.t. $\psi$ in a deterministic graph. \revise{The main difference is in the derivation of the capacity of a minimum cut in $\mathcal{H}$, which is quite different from the one in \S~\ref{sec:clique} as mentioned above; this derivation is shown in Appendix \ref{sec:pattern_proof}.} Finally, since Algorithm \ref{alg:pattern_densest} correctly computes all pattern-densest subgraphs w.r.t. $\psi$ in a deterministic graph, our overall accuracy guarantees for finding the top-$k$
\textsf{Pattern-MPDS} in an uncertain graph remain the same as in \S \ref{sec:top_k_mpds_algorithm}.

\vspace{-1mm}
\spara{Space complexity.} Let $G_c$ be the $(\lceil \Tilde{\rho} \rceil, \psi)$-core in Lines \ref{line:lb}-\ref{line:core} of Algorithm \ref{alg:pattern_densest}, with $n_c$ and $m_c$ the corresponding node and edge counts respectively. Clearly $|\Lambda| = \bigO \left( \mu_\psi(G_c) \right)$, where $\mu_\psi(G_c)$ is the number of $\psi$-instances in $G_c$. Thus, the number of nodes in $\mathcal{H}$ is $\bigO \left( n_c + \mu_\psi(G_c) \right)$ and the number of edges is $\bigO \left( n_c + \left| V_\psi \right| \mu_\psi(G_c) \right)$. Also, computing $\rho_\psi^*$ (Line \ref{line:val_pat}) requires $\bigO(n_c + |\Lambda|) = \bigO(n_c + \mu_\psi(G_c))$ space \cite{FYCLX19}. Since these constitute most of the memory consumed, the total space complexity of Algorithm \ref{alg:pattern_densest} is $\bigO \left( n_c + \left| V_\psi \right| \mu_\psi(G_c) \right)$.

When we adapt Algorithm \ref{alg:exact} for \textsf{Pattern-MPDS}, in addition to the memory required for storing the uncertain graph, the majority of the memory is consumed by Line \ref{line:all_densest}, which invokes Algorithm \ref{alg:pattern_densest}. Denoting by $\mu_\psi^*$ the maximum number of $\psi$-instances in the $(\lceil \Tilde{\rho} \rceil, h)$-core of any possible world of $\mathcal{G}$, the overall space complexity is $\bigO \left( m + n + \left| V_\psi \right| \mu_\psi^* \right)$.

\vspace{-1mm}
\spara{Time complexity.} Assume that the enumeration of all $\psi$-instances in a possible world takes $\bigO(t_{\psi})$ time. By a similar analysis as in \S~\ref{sec:clique}, denoting respectively by $n_c^*$, $d^*$, $L^*$, $\mu_\psi^*$, and $T^*$, the maximum number (in the $(\lceil \Tilde{\rho} \rceil, \psi)$-core of any possible world) of nodes, densest subgraphs, nodes in a densest subgraph,
$\psi$-instances, and iterations of \cite{SDCS20} to compute $\rho_\psi^*$, the overall time complexity of our method is $\bigO \left( \theta \left( t_{\psi} + 2^{T^*} \left| V_\psi \right| \mu_\psi^* + T^* \left( n_c^* \mu_\psi^* + n_c^{*3} \right) + d^* L^* \right) \right)$.

\vspace{-1mm}
\spara{Remark.} For larger graphs and bigger patterns $\psi$, we find that the enumeration of all $\psi$-instances (which is necessary to compute the densest subgraphs in a possible world) can be expensive. In such cases, we resort to a heuristic method in which we enumerate some reasonably dense subgraphs (instead of all densest ones) using \cite{FYCLX19}.
Specifically, we use a method different from Algorithm \ref{alg:pattern_densest}, which runs core decomposition w.r.t. $\psi$. If $k_{max}$ denotes the maximum core number, then
the $(k_{max}, \psi)$-core is a reasonably dense subgraph. In particular,
the $(k_{max}, \psi)$-core's density is at least $\frac{1}{\left| V_{\psi} \right|}$ times the maximum density of any subgraph \cite{FYCLX19}. Based on this, we return the $(k_{max}, \psi)$-core and all intermediate subgraphs (obtained during core decomposition) having greater densities.
Experimental results (\S~\ref{sec:large_graph_results}) show that this heuristic method yields good-quality solutions with higher efficiency.
 
\section{Approximate Solution for \\ Nucleus Densest Subgraphs}
\label{sec:nucleus}

In this section, we {\em convert the {\sf NDS} (nucleus densest subgraphs) problem into the widely-studied closed frequent itemset mining problem} and develop an {\em approximate method} (Algorithm \ref{alg:exact_nucleus}) to find the top-$k$ {\sf NDS} for all three notions of density:
edge, clique, and pattern density.

\begin{algorithm}[t]
	\renewcommand{\algorithmicrequire}{\textbf{Input:}}
	\renewcommand{\algorithmicensure}{\textbf{Output:}}
	\scriptsize
	\caption{\small Estimate all \textsf{NDS} in an uncertain graph} \label{alg:exact_nucleus}
	\begin{algorithmic}[1]
		\REQUIRE Uncertain graph $\mathcal{G} = (V,E,p)$, positive integers $k$ and $l_m$, no. of samples $\theta$
		\ENSURE (Approximate) \textsf{NDS}
		\STATE $CV \gets \emptyset$
		\FOR{$i = 1$ \TO $\theta$}
		\STATE Sample a possible world $G \sqsubseteq \mathcal{G}$
		\STATE $S \gets$ Maximum-sized densest subgraph in $G$ \label{line:all_densest_nucleus}
		\STATE $CV \gets CV \cup \{S\}$
		\ENDFOR
		\RETURN \revise{\textsf{TFP}($CV$, $k$, $l_m$) \cite{wang2005tfp}}
	\end{algorithmic}
\end{algorithm}

\revise{The algorithm first runs $\theta$ independent rounds: Sample a possible world $G \sqsubseteq \mathcal{G}$ and insert, into the set of candidate node sets $CV$, the maximum-sized densest subgraph of $G$\footnote{\scriptsize By a trivial generalization of \cite{BBCGS15} to all density notions, the union of the node sets of all densest subgraphs of a deterministic graph $G$ induces the maximum-sized (w.r.t. node count) densest subgraph in $G$. Thus, a node set is contained in a densest subgraph of $G$ if and only if it is contained in the maximum-sized densest subgraph of $G$.}. For a node set $U \subseteq V$, let $\widehat{\gamma}(U)$ denote the estimated densest subgraph containment probability of $U$, which is computed as the fraction of node sets in $CV$ which contain $U$. Then a \revise{closed} frequent itemset mining algorithm (e.g., \textsf{TFP} \cite{wang2005tfp}) is applied to compute the top-$k$ closed node sets in $CV$ of size at least $l_m$ having the largest values of $\widehat{\gamma}(\cdot)$. A node set is said to be closed w.r.t. $\widehat{\gamma}(\cdot)$ if it has no superset with the same value of $\widehat{\gamma}(\cdot)$ (i.e., is contained in the same number of node sets in $CV$). 
Here, $l_m$ is a user input decided based on the minimum desired size of a returned subgraph.}



The maximum-sized densest subgraph in a deterministic graph (Line \ref{line:all_densest_nucleus}) can be computed using parts of the methods in \S~\ref{sec:mpds_algo}. For edge density, we terminate after computing the SCCs of the residual graph under a maximum flow \cite{CQ20}. For clique and pattern densities, we terminate after computing the value of the maximum density of a subgraph, since we also get the maximum-sized densest subgraph in this process \cite{SDCS20}.


\vspace{-1mm}
\spara{Space complexity.} The major memory cost is due to the flow network in each iteration of Algorithm \ref{alg:exact_nucleus}. Therefore, the overall space complexity is the same as in \S~\ref{sec:mpds_algo}.

\vspace{-1mm}
\spara{Time complexity.} The time complexity is similar to the ones in \S~\ref{sec:mpds_algo},
plus that for computing the closed frequent node sets by \textsf{TFP}, which is reasonable in our experiments (\S~\ref{sec:parameter}).
Note that, for \textsf{Pattern-NDS} on our larger graphs, we use the heuristic method of \S~\ref{sec:pattern} in Line \ref{line:all_densest_nucleus} of Algorithm \ref{alg:exact_nucleus}.

\vspace{-1mm}
\spara{\revise{Accuracy guarantee.}} 
\revise{We theoretically analyze the sample size $\theta$ to return the true top-$k$ node sets with a high probability.
First, we prove that the true top-$k$ node sets are likely to be closed w.r.t. $\widehat{\gamma}(\cdot)$ after $\theta$ rounds (Theorem~\ref{th:cv_nucleus}), which is necessary (but not sufficient) for them to be finally returned. Second, we ensure that those sets indeed have the $k$ largest {\em estimated} densest subgraph containment probabilities (Theorem~\ref{th:topk_nucleus}), which guarantees that they are finally returned.
\begin{theor}
Given an uncertain graph $\mathcal{G}$ and positive integers $k$ and $l_m$, let $V_1, \ldots, V_{k}$ denote the true top-$k$ closed node sets w.r.t. $\gamma(\cdot)$ of size at least $l_m$ having the highest densest subgraph containment probabilities. For each $i \in [1, k]$, let $\mathbf{G}(V_i)$ denote the set of all possible worlds of $\mathcal{G}$ whose densest subgraphs contain $V_i$. Define $\mathbbm{G} = \bigcup_{i=1}^k \mathbf{G}(V_i)$. Then
\begin{footnotesize}
\begin{equation}
    \Pr \left( V_1, \ldots, V_k \text{ are closed w.r.t. } \widehat{\gamma} (\cdot) \right) \geq 1 - \sum_{G \in \mathbbm{G}} \left( 1 - \Pr(G) \right)^\theta
    \label{eq:cv_nucleus}
\end{equation}
\end{footnotesize}
\label{th:cv_nucleus}
\end{theor}
\begin{proof}
For each $i \in [1, k]$, if all possible worlds in $\mathbf{G}(V_i)$ are sampled at least once, then for each set $V' \supset V_i$, there is at least one node set (densest subgraph) in $CV$ containing $V_i$ but not $V'$, since $V_i$ is closed w.r.t. $\gamma(\cdot)$. In that case, $V_i$ is also closed w.r.t. $\widehat{\gamma}(\cdot)$. Thus, using this and the union bound,
\begin{small}
\begin{align*}
    &\Pr \left( V_1, \ldots, V_k \text{ are closed w.r.t. } \widehat{\gamma} (\cdot) \right) \geq \Pr \left( \bigwedge_{G \in \mathbbm{G}} G \text{ is sampled} \right) \\
    &= 1 - \Pr \left( \bigvee_{G \in \mathbbm{G}} G \text{ is never sampled} \right) \geq 1 - \sum_{G \in \mathbbm{G}} \left( 1 - \Pr(G) \right)^\theta
\end{align*}
\end{small}
\end{proof}
From Theorem \ref{th:cv_nucleus}, if the existence probabilities of the possible worlds whose densest subgraphs contain the true top-$k$ node sets 
and the sample size $\theta$ are reasonably large, they are all highly likely to be closed w.r.t. $\widehat{\gamma}(\cdot)$, which is necessary for all of them to be returned.
\begin{theor}
Let $V_1, \ldots, V_{k+1}$ denote the true top-$(k+1)$ closed node sets of size at least $l_m$ having the highest densest subgraph containment probabilities, and let $CV$ denote the set of candidate node sets after $\theta$ rounds. Define $mid = \frac{1}{2} \left[ \gamma\left(V_k\right) + \gamma\left(V_{k+1}\right) \right]$ and
\begin{small}
\begin{equation}
    d_U =
    \begin{cases}
    \gamma(U) - mid & \text{if} \, \, U \in \{V_1, \ldots, V_k\} \\
    mid - \gamma(U) & \text{otherwise}
    \end{cases}
\end{equation}
\end{small}
For each $i \in [1, k]$, let $\mathbf{G}(V_i)$ denote the set of all possible worlds of $\mathcal{G}$ whose densest subgraphs contain $V_i$. Define $\mathbbm{G} = \bigcup_{i=1}^k \mathbf{G}(V_i)$ and $\mathcal{CV}$ as the set of all closed node sets w.r.t. $\widehat{\gamma}(\cdot)$ of size at least $l_m$.
Then, the probability that $V_1, \ldots, V_k$ are returned by Algorithm~\ref{alg:exact_nucleus} is at least
\begin{small}
\begin{equation}
    \small \left[ 1 - \sum_{G \in \mathbbm{G}} \left( 1 - \Pr(G) \right)^\theta \right] \left[ 1 - \sum_{U \in \mathcal{CV}} \exp \left( -2d_U^2\theta \right) \right] \label{eq:prf_nucleus}
\end{equation}
\end{small}
\label{th:topk_nucleus}
\end{theor}
\begin{proof}
Let $\mathcal{CV}' = \mathcal{CV} \setminus \{V_1, \ldots, V_k\}$. We have:
\begin{footnotesize}
\begin{align}
    &\Pr(V_1, \ldots V_k \text{ are returned}) \geq \Pr(V_1, \ldots, V_k \text{ are closed w.r.t. } \widehat{\gamma}(\cdot)) \times \nonumber \\
    &\quad \Pr \left( \left(\bigwedge_{U \in \{V_1, \ldots, V_k\}} \!\!\!\!\!\! \widehat{\gamma}(U) > mid  \right)\wedge \left(\bigwedge_{U \in \mathcal{CV}'} \!\!\!\!\!\! \widehat{\gamma}(U) < mid \right) \right) \label{eq:pr1_nucleus}
\end{align}
\end{footnotesize}
Similar to Lemma \ref{lem:unbiased}, it can be proved that $\mathbbm{E}[\widehat{\gamma}(U)] = \gamma(U)$ $\forall U \subseteq V$. Using the union bound and Hoeffding's inequality,
\begin{footnotesize}
\begin{align}
    &\Pr \left( \left(\bigwedge_{U \in \{V_1, \ldots, V_k\}} \widehat{\gamma}(U) > mid \right)\wedge \left(\bigwedge_{U \in \mathcal{CV}'} \widehat{\gamma}(U) < mid \right) \right) \nonumber \\
    &= 1 - \Pr \left( \left(\bigvee_{U \in \{V_1, \ldots, V_k\}} \!\!\! \widehat{\gamma}(U) \leq mid \right) \vee \left(\bigvee_{U \in \mathcal{CV}'}  \widehat{\gamma}(U) \geq mid \right)\right) \nonumber \\
    &\geq 1 - \sum_{U \in \{V_1, \ldots, V_k\}}  \Pr \left( \widehat{\gamma}(U) \leq mid \right) - \sum_{U \in \mathcal{CV}'} \!\! \Pr \left( \widehat{\gamma}(U) \geq mid \right) \nonumber \\
    &= 1 - \!\!\!\! \sum_{U \in \{V_1, \ldots, V_k\}} \!\!\!\!\!\!\!\!\!\!\!\!\! \Pr \left( \widehat{\gamma}(U) - \gamma(U) \leq -d_U \right) - \!\! \sum_{U \in \mathcal{CV}'} \!\!\!\!\! \Pr \left( \widehat{\gamma}(U) - \gamma(U) \geq d_U \right) \nonumber \\
    &\geq 1 - \sum_{U \in \mathcal{CV}} \exp \left( -2d_U^2\theta \right) \label{eq:pr2_nucleus}
\end{align}
\end{footnotesize}
Finally, plugging \eqref{eq:cv_nucleus} and \eqref{eq:pr2_nucleus} into \eqref{eq:pr1_nucleus}, we obtain \eqref{eq:prf_nucleus}.
\end{proof}
From Theorem \ref{th:topk_nucleus}, if the densest subgraph containment probabilities of the true top-$k$ node sets are reasonably larger than the other closed node sets w.r.t. $\widehat{\gamma}(\cdot)$ of size at least $l_m$,
and if the sample size $\theta$ is sufficiently large, the true top-$k$ sets are returned by Algorithm \ref{alg:exact_nucleus} with a high probability.
}

\vspace{-1mm}
\section{Related Work}
\label{sec:related}

\revise{In this section, we first revisit the densest subgraph discovery problem in deterministic graphs. Then, we discuss the existing attempt that extends the edge density to uncertain graphs by considering the maximum expected density. Finally, we state several close notions for cohesive and dense substructures in uncertain graphs, including core, truss, maximal cliques, clustering, and highly reliable subgraphs.}

\vspace{-1mm}
\spara{\revise{Densest subgraph in a deterministic graph.}} \revise{Given an undirected, unweighted, deterministic graph, the original densest subgraph (DS) problem \cite{G84} finds a subgraph with the highest edge-density, exact solutions to which are based on min-cuts in flow-networks \cite{G84, FYCLX19} and linear programming \cite{G84, Charikar00}. Since the computation of maximum flow has a high time complexity, researchers proposed approximate algorithms with theoretical guarantees, e.g., \cite{Charikar00, FYCLX19, BKV12, SW20}. Variants of the edge-density-based DS problem were also studied, such as triangle-density, clique-density, pattern-density, and edge-surplus based DS \cite{FYCLX19, MPPTX15, SDCS20, T15, TsourakakisBGGT13}, densest $k$-connected subgraph \cite{BGMT22} and size-bounded DS \cite{AC09}, top-$k$ overlapping DS \cite{GGT16}, top-$k$ DS maintenance on dynamic graphs \cite{ NGMG17, BKV12, SW20}, locally DS \cite{QLCZ15}, robust DS \cite{miyauchi2018robust}, density-friendly graph decomposition \cite{DCS17,T19}, DS on directed \cite{ MFCLZL20}, bipartite \cite{ HSBSSF16}, multilayer graphs \cite{GBBL20} etc.  For surveys and tutorials, we refer to \cite{FangLM22, GT15}. Enumerating all the densest subgraphs based on edge-density in a deterministic graph has been recently studied in \cite{CQ20}, which we use as a subroutine to find the \textsf{MPDS} of an uncertain graph in our problem. We notice that the problems of enumerating all clique-DS and pattern-DS in a deterministic graph were not studied in the literature. Thus, as additional technical contributions, we develop novel, exact algorithms for efficiently enumerating all clique-DS and all pattern-DS in a deterministic graph, and use them as subroutines to respectively find the \textsf{$h$-Clique-MPDS} and the \textsf{Pattern-MPDS} of an uncertain graph in our problem.}

\vspace{-1mm}
\spara{Expected edge densest subgraph.}
The expected edge density of a node set $U$ in an uncertain graph is the expectation of the edge density of the subgraph induced by $U$ across all possible worlds.
Zou \cite{Zou13} designed a polynomial-time algorithm to find the
subgraph with the maximum expected edge density in an uncertain graph using maximum flow techniques. In a graph where each edge weight distribution follows a given mean (reward) and variance (risk), Tsourakakis et al. \cite{tsourakakis2020novel} find a node set whose induced subgraph has high average reward (i.e., expected density) and low average risk.
As shown in Example~\ref{exam:expected}
and our experiments (\S \ref{sec:comp_others} and \S \ref{sec:case_club}), the expected edge densest subgraph is different from the {\sf MPDS}:
a subgraph of an uncertain graph having the maximum expected edge density may induce densest subgraphs
{\em only in a few} possible worlds of the graph. Such a subgraph can be large with many
low-probability edges or loosely connected nodes.

\vspace{-1mm}
\spara{Core and truss decompositions in an uncertain graph.}
As cohesive and dense substructures finding, core and truss decompositions are popular. The $k$-core (resp. $k$-truss) of a graph is a maximal subgraph
in which every node is connected to at least $k$ other nodes (resp. each edge participates in at least $(k-2)$ triangles) within that subgraph. These notions have been
extended to uncertain graphs, returning those subgraphs satisfying the above conditions with probability at least a threshold \cite{BonchiGKV14, EsfahaniW0T019, SMPV17, Dai21, HuangLL16, Sun0XB21, ZouZ17}. The innermost cores and trusses
have been used in applications such as task-driven team formation due to their higher densities \cite{HuangLL16,BonchiGKV14}. However,
they are different from the {\sf MPDS}
(\S \ref{sec:comp_others} and \S \ref{sec:case_club}).
Unlike the {\sf MPDS}, they do not find the node set most likely to induce a densest subgraph in the uncertain graph.

\vspace{-1mm}
\spara{Top-$k$ maximal cliques in an uncertain graph.} A clique is a set of nodes with each pair connected by an edge.
\cite{MukherjeeXT17, ZouLGZ10, LiDWMQY19, DLLCW22} study enumeration of maximal cliques in uncertain graphs. Densest subgraphs are not necessarily cliques.

\vspace{-1mm}
\spara{Node clustering in an uncertain graph.} $k$-means and $k$-medians approaches have been extended to node clustering in uncertain
graphs \cite{LJAS12,HanGXTHCH19,CeccarelloFPPV17}: Partition the nodes into $k$ clusters such that some notion of path-based
connection probability (e.g., reliability) between each node and its cluster's center
(or, across each pair of nodes within every cluster) is maximized.
Notice that clustering methods optimize the path-based connection probability between nodes inside a cluster, and do not optimize
the density of a subgraph.

\vspace{-1mm}
\spara{Highly reliable subgraphs in an uncertain graph.} Jin et al. \cite{JinLA11} have developed a mining-based approach to
discover highly reliable subgraphs in an uncertain graph. A highly reliable subgraph denotes a node set that has
a high probability of remaining connected via some paths in possible worlds of the uncertain graph. This method deals with
path-based connectivity between nodes, and is hence not suitable to identify densest subgraphs, such as {\sf MPDS} in our work.

\begin{table}[t]
	\scriptsize
	\centering
	\begin{center}
		\caption{\small Characteristics of our datasets}
		\vspace{-2mm}
		\begin{tabular}{c|c|c|c|c}
			\textbf{Name}   & \textbf{$n$} & \textbf{$m$} & \textbf{Type} & \textbf{Edge Prob: Mean,} \\
			\textbf{}   & \textbf{} & \textbf{} & \textbf{} & \textbf{ St. Dev., Quart.} \\ \hline \hline
			\textbf{\sf Karate} & 34 & 78 & Social & 0.25, 0.09,\\
			\textbf{\sf Club}& & & & \{0.18, 0.26, 0.33\} \\
			\hline
			\textbf{\sf Intel} & 54 & 969 & Device & 0.33, 0.19, \\
			\textbf{\sf Lab}& & & & \{0.16, 0.27, 0.44\}\\
			\hline
			\textbf{\sf LastFM} & 6\,899 & 23\,696 & Social & 0.33, 0.19, \\
			& & & & \{0.16, 0.27, 0.44\}\\\hline
			\textbf{\sf Homo} & 18\,384 & 995\,916 & Bio & 0.32, 0.21,\\
			\textbf{\sf Sapiens} &  &  & & \{0.18, 0.24, 0.34\}\\
			\hline
			\textbf{\sf Biomine} & 1\,045\,414 & 6\,742\,939 & Bio & 0.27, 0.21,\\
			& & & &  \{0.12, 0.22, 0.36\}\\\hline
			\textbf{\sf Twitter} & 6\,294\,565 & 11\,063\,034 & Social &0.14, 0.10, \\
			& & & & \{0.10, 0.10, 0.19\} \\\hline
			\revise{\textbf{\sf Friendster}} & \revise{65\,608\,366} & \revise{1\,806\,067\,135} & \revise{Social} & \revise{0.005, 0.013,} \\
			& & & & \revise{\{0.001, 0.003, 0.005\}} \\
		\end{tabular}
	\vspace{-2mm}
		\label{tab:dataset}
	\end{center}
\end{table}

\begin{figure}
  \vspace{-2mm}
  \centering
  \includegraphics[scale=0.45]{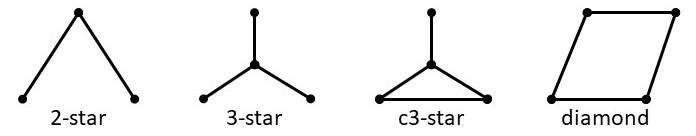}
   \vspace{-3mm}
  \captionof{figure}{\small Patterns used in our experiments}
  \label{fig:patterns}
    \vspace{-3mm}
\end{figure}

\vspace{-1mm}
\section{Experimental Results}
\label{sec:exp}
We run experiments to demonstrate the effectiveness and efficiency of our methods. Our C++ code \cite{code} is executed on one core 
of a 512GB,
2.4GHz Xeon server running Ubuntu.
\vspace{-2mm}
\subsection{Experimental Setup}
\label{sec:setup}
\vspace{-2mm}
\spara{Datasets.} We conduct experimental evaluations on six real-world graphs (Table \ref{tab:dataset}).
{\bf (1) \sf Karate Club} \cite{zachary1977information} is a social network of a university karate club. Nodes are club members and edges denote member interactions.
{\bf (2) \sf Intel Lab} \cite{Intel} is a collection of sensor communication data with 54 sensors deployed in the Intel Berkeley Research Lab between February 28 and April 5, 2004.
{\bf (3) \sf LastFM} \cite{lastfm} is a musical social network where users listen to music and share with friends. An edge between two users exists if they communicate at least once.
{\bf (4) \sf Homo Sapiens} \cite{szklarczyk2021string} is a protein-protein interaction network, where nodes are proteins and edges are interactions between proteins.
{\bf (5) \sf Biomine} \cite{BioMine} is constructed by integrating cross-references from several biological databases. Nodes represent biological concepts such as genes, proteins etc., and edges denote real-world phenomena between two nodes, e.g., a gene ``code'' for a protein. 
{\bf (6) \sf Twitter} \cite{snapnets} is a widely used social network where nodes are users and edges are retweets.
\revise{{\bf (7) \sf Friendster} \cite{Yang2012DefiningAE} is a social network where nodes are users and edges denote friendships.}

\vspace{-1mm}
\spara{Edge probability models.} We adopt various widely used models for generating the edge probabilities in our evaluation.
In {\bf (1)} {\sf Karate Club}, {\bf (6)} {\sf Twitter} \revise{and {\bf (7)} {\sf Friendster}}, we assign the probability of an edge as $1-e^{-\frac{t}{\mu}}$, which is an
exponential cdf of mean $\mu$ to the number $t$
of communications between the two users. We set $\mu=20$ \cite{KKHR20}.
In {\bf (2)} {\sf Intel Lab}, a (\revise{real}) edge probability denotes the fraction of messages from the sender that successfully reached the receiver \cite{Intel, SahaBVKB21}.
In {\bf (3)} {\sf LastFM}, the probability of any edge is the reciprocal of the larger of the out-degrees of its source and target nodes \cite{KKHR20}.
In {\bf (4)} {\sf Homo Sapiens}, an edge probability represents the confidence on the existence of the corresponding interaction, based on \revise{real} biological experiments \cite{szklarczyk2021string}.
In {\bf (5)} {\sf Biomine} \cite{BioMine}, an edge probability quantifies the existence of a phenomenon between the two endpoints, which was a combination of three criteria: relevance, informativeness, and confidence on the existence of a specific relationship \cite{PBGK10, BonchiGKV14, HuangLL16}.

\vspace{-1mm}
\spara{Methods compared.} We compare our \textsf{MPDS} and \textsf{NDS} algorithms with those for computing the expected densest subgraph \cite{Zou13}, $(k, \eta)$-core \cite{BonchiGKV14}, and $(k, \gamma)$-truss \cite{HuangLL16} (\S~\ref{sec:comp_others}). As discussed in \S~\ref{sec:problem}, we show the results for \textsf{MPDS} on the three smaller datasets and \textsf{NDS} on the three larger ones. Additionally, on some small synthetic graphs, we compare our \textsf{MPDS} approximation algorithms with the corresponding exact ones (\S~\ref{sec:comp_exact}). For sampling possible worlds, we compare our employed {\em Monte Carlo (MC)} method with {\em  Lazy Propagation} \cite{LFZT17} and {\em Recursive Stratified Sampling} \cite{RSS16} (\S~\ref{sec:parameter}).


\vspace{-1mm}
\spara{Parameters.} {\bf $\bullet$ $h$ for Clique-\textsf{MPDS/NDS}}: We vary $h \in \{3, 4, 5\}$ \cite{MPPTX15}. Notice that $h = 2$ denotes an edge.
{\bf $\bullet$ $\psi$ for Pattern-\textsf{MPDS/NDS}}: We vary $\psi \in$ \{2-star, 3-star, c3-star, diamond\} \cite{FYCLX19}, as shown in Figure \ref{fig:patterns}.
{\bf $\bullet$ Top-$k$ \textsf{MPDS}s}: We vary $k \in \{1, 5, 10\}$, with default value $1$.
\revise{{\bf $\bullet$ Top-$k$ \textsf{NDS}s}: We vary $k \in \{1, 5, 10, 50, 100\}$. 
{\bf $\bullet$ Minimum size threshold $l_m$}: We vary $l_m \in [1, 750]$. Beyond that range, no \textsf{NDS} is returned for any of our datasets.}
{\bf $\bullet$ Number of sampled possible worlds $\theta$}: We vary $\theta \in \{2^0\times10, 2^1\times10, \ldots, 2^8\times10\}$. The default value is chosen as in \S~\ref{sec:parameter}. 

\begin{table}[t]
    \vspace{-2mm}
	\footnotesize
	\centering
	\begin{center}
		\caption{\small Densest subgraph containment probabilities of the \textsf{NDS}, the expected densest subgraph (\textsf{EDS}), the innermost $\eta$-core and $\gamma$-truss ($\eta = \gamma = 0.1$); Expected densities of the \textsf{NDS} and \textsf{EDS}}
		\vspace{-1mm}	
		\begin{tabular}{c||c|c|c|c|c|c}
    	\multirow{2}{*}{\textbf{Dataset}} & \multicolumn{4}{|c}{\textbf{Containment Probability}} & \multicolumn{2}{|c}{\textbf{Expected Density}} \\ \cline{2-7}
			& \textbf{\textsf{NDS}} & \textbf{\textsf{EDS}} & \textbf{Core} & \textbf{Truss} & \textbf{\textsf{NDS}} & \textbf{\textsf{EDS}} \\ \hline \hline
			\textbf{\sf Homo Sapiens} & {\bf 1} & 0.05 & {\bf 1} & {\bf 1} & 54 & {\bf 54.62} \\ \hline
			\textbf{\sf Biomine} & {\bf 1} & 0.01 & 0.99 & 0 & 46.45 & {\bf 48.02} \\ \hline
			\textbf{\sf Twitter} & {\bf 1} & 0 & 0.95 & 0 & 37.65 & {\bf 38.64} 
		\end{tabular}
		\vspace{-4mm}
	\label{tab:existing_nds}
	\end{center}
\end{table}

\begin{table}[t]
	\footnotesize
	\centering
	\begin{center}
		\caption{\small Densest subgraph probabilities of the \textsf{MPDS}, the expected densest subgraph (\textsf{EDS}), the innermost $\eta$-core and $\gamma$-truss ($\eta = \gamma = 0.1$); Expected densities of the \textsf{MPDS} and \textsf{EDS}}
		\vspace{-1mm}
		\begin{tabular}{c||c|c|c|c|c|c}
        \multirow{2}{*}{\textbf{Dataset}} & \multicolumn{4}{|c}{\textbf{Densest Subgraph Probability}} & \multicolumn{2}{|c}{\textbf{Expected Density}} \\ \cline{2-7}
			& \textbf{\textsf{MPDS}} & \textbf{\textsf{EDS}} & \textbf{Core} & \textbf{Truss} & \textbf{\textsf{MPDS}} & \textbf{\textsf{EDS}} \\ \hline \hline
			\textbf{\sf Karate Club} & {\bf 0.012} & 0 & 0 & 0 & 0.703 & {\bf 0.75} \\ \hline
			\textbf{\sf Intel Lab} & {\bf 0.078} & 0.01 & 0.01 & 0 & 3.246 & {\bf 3.25} \\ \hline
			\textbf{\sf LastFM} & {\bf 0.075} & 0 & 0.04 & 0.02 & 0.667 & {\bf 0.86} 
		\end{tabular}
		\vspace{-5mm}
		\label{tab:existing_mpds}
	\end{center}
\end{table}

\begin{figure*}
    \vspace{-2mm}
    \centering
    \subcaptionbox{\small{\textsf{MPDS}} \label{fig:club_mpds}}[0.235\textwidth] {\includegraphics[scale=0.29]{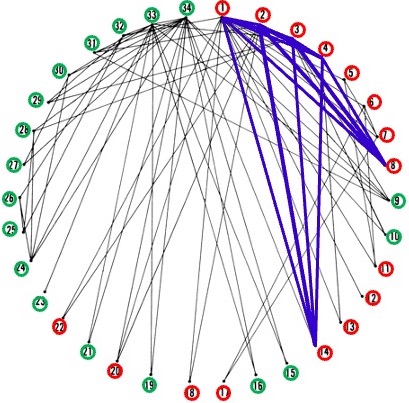}\vspace{-1mm}}
    \subcaptionbox{\small{Expected Densest Subgraph} \label{fig:club_eds}}[0.235\textwidth] {\includegraphics[scale=0.29]{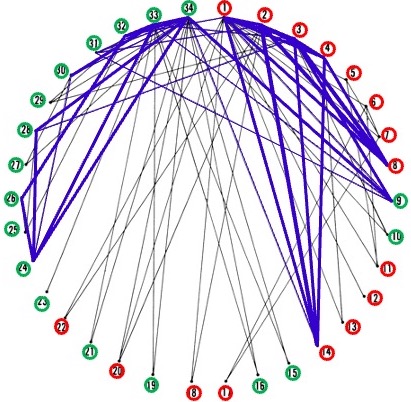}\vspace{-1mm}}
    \subcaptionbox{\small{Innermost core} \label{fig:club_core}}[0.235\textwidth] {\includegraphics[scale=0.29]{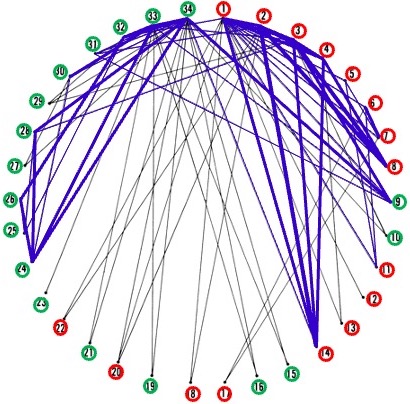}\vspace{-1mm}}
    \subcaptionbox{\small{Innermost truss} \label{fig:club_truss}}[0.235\textwidth] {\includegraphics[scale=0.29]{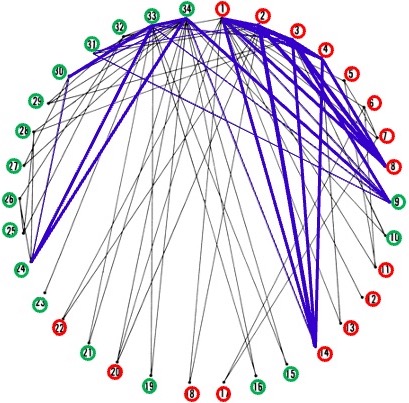}\vspace{-1mm}}
    \vspace{-2mm}
    \caption{\small Case studies to compare with existing dense subgraphs in uncertain graphs on the \textsf{Karate Club} dataset. The relevant subgraphs are in blue, while \revise{the colour of a node denotes its ground-truth community \cite{zachary1977information}. The thickness of each edge is proportional to its probability}.}
    \vspace{-5mm}
    \label{fig:case_club}
\end{figure*}

\vspace{-2mm}
\subsection{Comparison with Expected Density, Core, and Truss Decompositions in Uncertain Graphs}
\label{sec:comp_others}

For each of our larger datasets, we compare our \textsf{NDS} with some existing or close notions of densest subgraphs in uncertain graphs: expected densest subgraph \cite{Zou13}, innermost core \cite{BonchiGKV14}, and innermost truss \cite{HuangLL16}. Specifically, we compare the (approximate) densest subgraph containment probabilities of the following: $\bullet$ 
The (maximal) node set returned by Algorithm \ref{alg:exact_nucleus} with the highest frequency of being contained in the densest subgraphs of the generated possible worlds; $\bullet$ the expected densest subgraph (\textsf{EDS}); $\bullet$ for a given $\eta$, the innermost $\eta$-core, i.e., the $(k, \eta)$-core with the largest value of $k$; $\bullet$ for a given $\gamma$, the innermost $\gamma$-truss (analogous to cores). As shown in Table \ref{tab:existing_nds}, the containment probability of the $\eta$-core is comparable to (yet not greater than) that of the \textsf{NDS} for all datasets, in contrast to the \textsf{EDS} and the $\gamma$-truss. This makes sense for the following reason. The innermost $\eta$-core is likely to be an innermost core (and hence a reasonably dense subgraph \cite{FYCLX19}) of a possible world of the input graph. 
However, the same cannot be said about the other subgraphs. Table \ref{tab:existing_nds} demonstrates that our solution produces the most optimal node set with the highest densest subgraph containment probability compared to the other approaches.

In addition, for each of our smaller datasets, we compare the densest subgraph probability of the \textsf{MPDS} with those of the \textsf{EDS}, innermost $\eta$-core and innermost $\gamma$-truss. As shown in Table \ref{tab:existing_mpds}, the \textsf{MPDS} outperforms the other subgraphs.

\begin{table}[t!]
	\footnotesize
	\centering
	\vspace{-2mm}
	\begin{center}
        \caption{\shepherd{\small Probabilistic density of our proposed subgraphs (\textsf{MPDS} for the two smaller datasets and \textsf{NDS} for the two larger ones) and of existing dense subgraphs in uncertain graphs}}
		\vspace{-1.5mm}
		\shepherd{
		\begin{tabular}{c||c|c|c|c}
		\multirow{2}{*}	{\textbf{Dataset}} & \multicolumn{4}{|c}{\textbf{Probabilistic Density}}\\ \cline{2-5}
      & \textbf{MPDS/NDS} & \textbf{EDS} & \textbf{Core} & \textbf{Truss}  \\ \hline \hline
			\textsf{Karate Club} & ${\bf 0.281}$ & $0.095$ & $0.073$ & $0.134$ \\ \hline
			\textsf{LastFM} & ${\bf 0.333}$ & $0.007$ & $0.008$ & $0.013$ \\ \hline
			\textsf{Biomine} & ${\bf 0.546}$ & $0.191$ & $0.212$ & $0.538$ \\ \hline
			\textsf{Twitter} & ${\bf 0.789}$ & $0.042$ & $0.121$ & $0.781$ \\ 
		\end{tabular}}
		\vspace{-6mm}
		\label{tab:pd}
	\end{center}
\end{table}

\begin{table}[t!]
	\footnotesize
	\centering
	\begin{center}
        \caption{\shepherd{\small Probabilistic clustering coefficient of our proposed subgraphs (\textsf{MPDS} for the two smaller datasets and \textsf{NDS} for the two larger ones) and of existing dense subgraphs in uncertain graphs}}
		\vspace{-1.5mm}
		\shepherd{
		\begin{tabular}{c||c|c|c|c}
			\multirow{2}{*}{\textbf{Dataset}} & \multicolumn{4}{|c}{\textbf{Probabilistic Clustering Coefficient}}\\ \cline{2-5}
            &\textbf{MPDS/NDS} & \textbf{EDS} & \textbf{Core} & \textbf{Truss}  \\ \hline \hline
			\textsf{Karate Club} & ${\bf 0.284}$ & $0.150$ & $0.094$ & $0.158$ \\ \hline
			\textsf{LastFM} & ${\bf 0.333}$ & $0.002$ & $0.022$ & $0.257$ \\ \hline
			\textsf{Biomine} & ${\bf 0.546}$ & $0.203$ & $0.217$ & $0.539$ \\ \hline
			\textsf{Twitter} & ${\bf 0.775}$ & $0.142$ & $0.253$ & $0.768$ \\ 
		\end{tabular}}
		\vspace{-7mm}
		\label{tab:pcc}
	\end{center}
\end{table}

Since the \textsf{EDS} performs very poorly for all datasets, for fairness, we compare the expected densities of the \textsf{EDS} and our \textsf{NDS}/\textsf{MPDS}. Tables \ref{tab:existing_nds} and \ref{tab:existing_mpds} show that our solutions produce subgraphs with expected densities comparable to the optimal values, thereby showing that our returned subgraphs are good even with respect to expected density.

\shepherd{
We also consider two {\em external} evaluation metrics:  
$\bullet$ {\em Probabilistic Density} $PD(U)$ \cite{HuangLL16} for capturing the cohesiveness of a probabilistic subgraph $U$, which is defined as the weighted sum of existing edges divided by the maximum number of possible edges this subgraph can have (Equation~\ref{eq:pd}); and $\bullet$ {\em Probabilistic Clustering Coefficient} $PCC(U)$ \cite{JJ11} for measuring how well the nodes in a probabilistic subgraph $U$ cluster together, which is computed as three times the weighted sum of all possible triangles divided by the weighted sum of all neighboring edge pairs. The weights are existence probabilities of edges/triangles/neighboring edge pairs, assuming independence among edges (Equation~\ref{eq:pcc}). $V_U$, $E_U$, and $\Delta_U$ denote the set of nodes, edges, and triangles in $U$, respectively.
\begin{small}
\begin{align}
	PD(U)&=\frac{2\sum_{e\in E_U} p(e)}{|V_U|(|V_U|-1)}
\label{eq:pd} \\
    PCC(U)&=\frac{3\sum_{\bigtriangleup_{uvw}\in \Delta_U} p(u,v)p(u,w)p(v,w)}{\sum_{(u,v),(u,w)\in E_U,v\ne w}p(u,v)p(u,w)} 
\label{eq:pcc}    
\end{align}
\end{small}
Tables~\ref{tab:pd}-\ref{tab:pcc} demonstrate that our \textsf{NDS/MPDS} significantly outperforms other dense subgraph notions based on both probabilistic density and probabilistic clustering coefficient, implying that the \textsf{NDS/MPDS} is much more cohesive, i.e., most of the possible edges induced by the \textsf{NDS/MPDS} node set tend to exist, and the nodes in the \textsf{NDS/MPDS} cluster together. Only the innermost truss achieves slightly lower results on the two larger datasets. 
}

\subsection{\revise{Comparison between  \textsf{MPDS} and the Densest Subgraph in the Deterministic Version of Input Graph}}

\begin{table}[t!]
	\footnotesize
	\centering
	\begin{center}
		\caption{\revise{\small Densest subgraph probabilities of the \textsf{MPDS} and the deterministic densest subgraph (\textsf{DDS})}}
		\vspace{-1mm}
		\revise{
		\begin{tabular}{c||c|c|c}
			\textbf{Subgraph} & \textsf{Karate Club} & \textsf{Intel Lab} & \textsf{LastFM}  \\ \hline \hline
			\textsf{MPDS} & $\bf 0.012$ & $\bf 0.078$ & $\bf 0.075$ \\ \hline
			\textsf{DDS} & $\approx 0$ & $0.044$ & $\approx 0$ 
		\end{tabular}}
		\vspace{-4mm}
		\label{tab:mpdsvsdet}
	\end{center}
\end{table}

\begin{figure}[t!]
	\centering
	\subcaptionbox{\revise{\small{3 of the top-$20$ \textsf{MPDS}s}} \label{fig:case_topk}} [0.235\textwidth]{\includegraphics[scale=0.29]{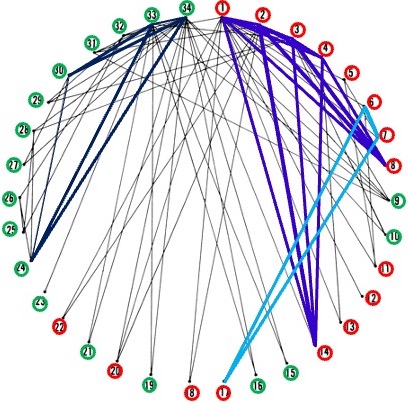}\vspace{-2mm}}
	\subcaptionbox{\revise{\small{Det. densest subgraph}} \label{fig:case_det}} [0.235\textwidth]{\includegraphics[scale=0.29]{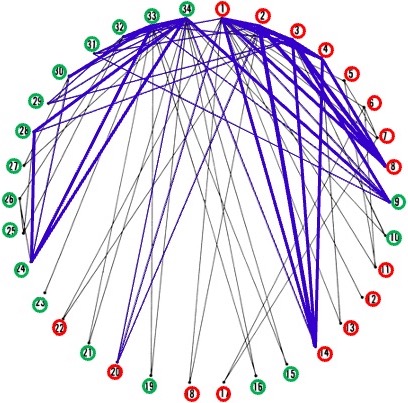}\vspace{-2mm}}
	\vspace{-1mm}
	\caption{\revise{\small 3 \textsf{MPDS}s (different colors) and the densest subgraph in the deterministic version of the \textsf{Karate Club} dataset. The color of a node denotes the community to which it belongs \cite{zachary1977information}. We only show $3$ of the top-$20$ \textsf{MPDS}s for better visualization; the others have node sets that are subsets of those shown. The thickness of each edge is proportional to its probability.}}
	\label{fig:case_topkvsdet}
	\vspace{-4mm}
\end{figure}

\revise{ 
As shown in Table \ref{tab:mpdsvsdet}, the (estimated) densest subgraph probability of the \textsf{MPDS} is much higher than that of the densest subgraph in the deterministic version of an uncertain graph (denoted by \textsf{DDS}). The reason is that the rich information encoded in the edge probabilities is ignored by the \textsf{DDS}. In practice, nodes may be densely connected via low-probability edges, e.g., due to noise.
Our proposed \textsf{MPDS} can capture and filter out this, leading to more useful results in uncertain graphs compared to the \textsf{DDS}.}

\subsection{\revise{Considering {\em All}  vs. One Densest Subgraph(s) in Each Sampled World}} 
\label{sec:allvs1}

\begin{table}[t!]
	\centering
	\begin{center}
		\caption{\revise{\small Distribution (mean, standard deviation, quartiles) of the number of densest subgraphs (edge, 3-clique, diamond) across all sampling rounds in our \textsf{MPDS} algorithm}}
 		\vspace{-1mm}
        \revise{
	    \footnotesize
		\begin{tabular}{c||c|c}
		    \textbf{Notion} & \textbf{\sf Karate Club} & \textbf{\sf LastFM} \\ \hline \hline
		    Edge & 1.12, 0.54, \{1, 1, 1\} & 2613.24, 22825.66, \{15, 127, 1023\} \\ \hline
		    3-Clique & 1.35, 0.91, \{1, 1, 1\} & 1880.74, 22134, \{31, 127, 511\} \\ \hline
		    Diamond & 1.18, 0.71, \{1, 1, 1\} & 3.52, 9.6, \{1, 1, 3\}
		\end{tabular}}
 		\vspace{-4mm}
		\label{tab:dist}
	\end{center}
\end{table}

\begin{table}[t!]
	\footnotesize
	\centering
	\begin{center}
		\caption{\revise{\small Average densest subgraph probabilities of the top-$10$ \textsf{MPDS}s returned by computing all vs only one densest subgraph in each sampled world}}
		\vspace{-1mm}
		\revise{
		\begin{tabular}{c||c|c|c|c|c|c}
    	\multirow{2}{*}{\textbf{Dataset}} & \multicolumn{2}{|c|}{\textbf{Edge}} & \multicolumn{2}{|c|}{\textbf{3-Clique}} & \multicolumn{2}{|c}{\textbf{Diamond}} \\ \cline{2-7}
			& \textbf{All} & \textbf{One} & \textbf{All} & \textbf{One} & \textbf{All} & \textbf{One} \\ \hline \hline
			\textsf{Karate Club} & $\bf 0.006$ & $0.005$ & $\bf 0.019$ & $0.018$ & $\bf 0.011$ & $0.01$ \\ \hline
			\textsf{LastFM} & $\bf 0.054$ & $0.004$ & $\bf 0.08$ & $0.004$ & $\bf 0.009$ & $0.007$ 
		\end{tabular}}
 		\vspace{-5mm}
		\label{tab:allvs1}
	\end{center}
\end{table}

\begin{table}[t!]
	\footnotesize
	\centering
	\begin{center}
		\caption{\shepherd{\small Purity (\S~\ref{sec:case_club}) of the node sets in the top-$k$ \textsf{MPDS}s and in the existing notions of dense subgraphs over \textsf{Karate Club} uncertain graph. There are only two cores and two trusses in this graph; thus the entries for $k > 2$ are empty for those subgraphs.}}
		\vspace{-1.5mm}
		\shepherd{
		\begin{tabular}{c||c|c|c|c}
			\multirow{2}{*}{\textbf{Top-$k$}} & \multicolumn{4}{|c}{\textbf{Purity}} \\ \cline{2-5} 
             & \textbf{MPDS} & \textbf{EDS} & \textbf{Core} & \textbf{Truss}  \\ \hline \hline
			\textsf{$1$} & ${\bf 1}$ & $0.6$ & $0.5$ & $0.538$ \\ \hline
			\textsf{$2$} & ${\bf 1}$ & $0.6$ & $0.515$ & $0.536$ \\ \hline
			\textsf{$5$} & ${\bf 1}$ & $0.749$ & - & - \\ \hline
			\textsf{$10$} & ${\bf 1}$ & $0.699$ & - & - \\ 
		\end{tabular}}
		\vspace{-7mm}
		\label{tab:purity}
	\end{center}
\end{table}

\revise{As shown in Table~\ref{tab:dist}, the number of densest subgraphs in a deterministic sample can be very large in practice (e.g., in {\sf LastFM}). 
Thus, if we compute only one densest subgraph (instead of all such subgraphs) in each world, the frequency of a particular subgraph in the candidate set (and hence its estimated densest subgraph probability) will reduce. Table \ref{tab:allvs1} presents that the average (estimated) densest subgraph probability of the top-$10$ results reduces if we compute only one densest subgraph. This gap can be up to 20$\times$ when the number of densest subgraphs is huge ({\sf LastFM} as per Table~\ref{tab:dist}).
Therefore, computing all densest subgraphs in each sampled world is critical in finding {\sf MPDS}s.}

\vspace{-1mm}
\subsection{Case Studies: \textsf{Karate Club} Network}
\label{sec:case_club}

\vspace{-1mm}
Densest subgraphs in social networks can correspond to communities \cite{DourisboureGP09}, filter bubbles and echo chambers \cite{asatani2021dense, L22}.

\spara{\revise{Comparison of the \textsf{MPDS} with the deterministic densest subgraph.}}
\revise{Figure \ref{fig:case_topkvsdet} shows that the densest subgraph given by the deterministic version of the \textsf{Karate Club} dataset has a much larger size with many low-probability edges and contains nodes from both ground-truth communities \cite{zachary1977information}. On the contrary, each of the \textsf{MPDS}s contains nodes from only one ground-truth community and has edges with higher probabilities. This demonstrates that \textsf{MPDS} is more powerful in detecting communities in uncertain graphs than simply considering the densest subgraph in the deterministic version.}

\begin{figure}[t!]
	\centering
	\subcaptionbox{\small{Typically developed} \label{fig:case_brain_1}} [0.235\textwidth]{\includegraphics[scale=0.225]{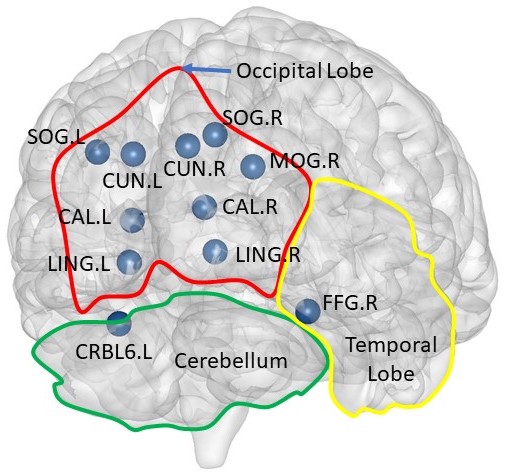}\vspace{-2mm}}
	\subcaptionbox{\small{ASD-affected} \label{fig:case_brain_2}} [0.235\textwidth]{\includegraphics[scale=0.225]{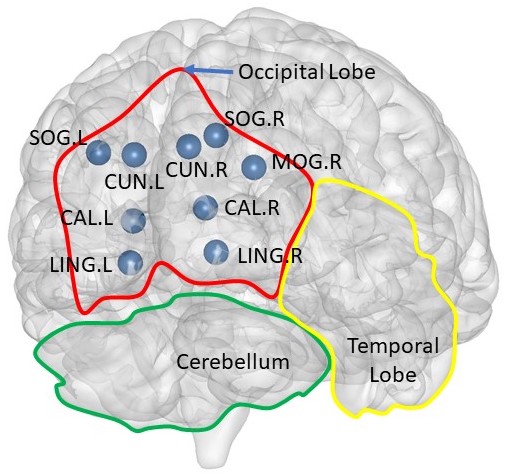}\vspace{-2mm}}
	\vspace{-1mm}
	\caption{\small Node sets of 3-clique \textsf{MPDS}s in brain networks. The colored boundaries denote the cerebellum, occipital, and temporal lobes.}
	\label{fig:case_brain}
	\vspace{-5mm}
\end{figure}

\begin{figure}[t!]
	\centering
	\subcaptionbox{\small{Typically developed} \label{fig:case_graph_1}} [0.235\textwidth]{\includegraphics[scale=0.07]{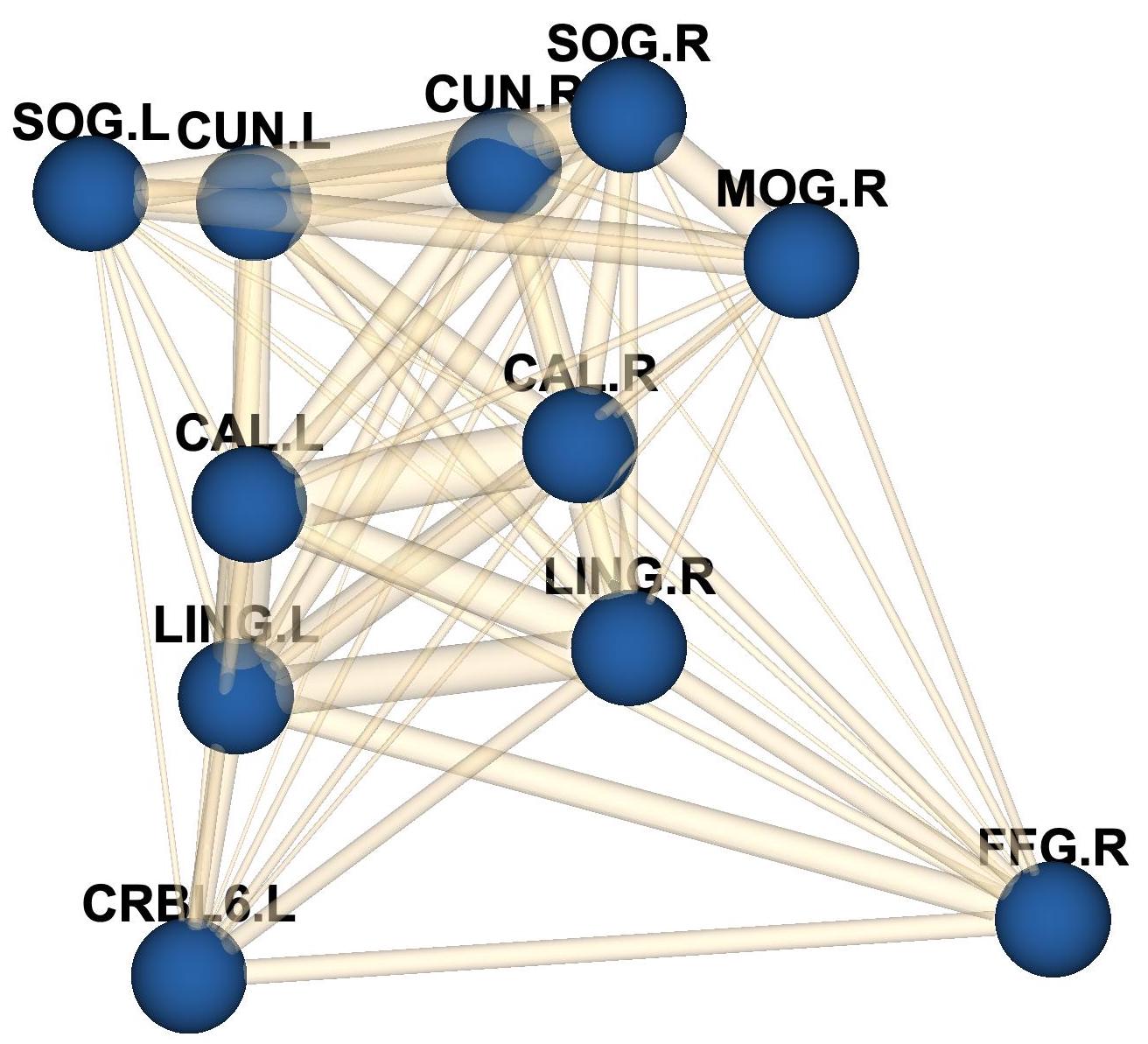}\vspace{-2mm}}
	\subcaptionbox{\small{ASD-affected} \label{fig:case_graph_2}} [0.235\textwidth]{\includegraphics[scale=0.07]{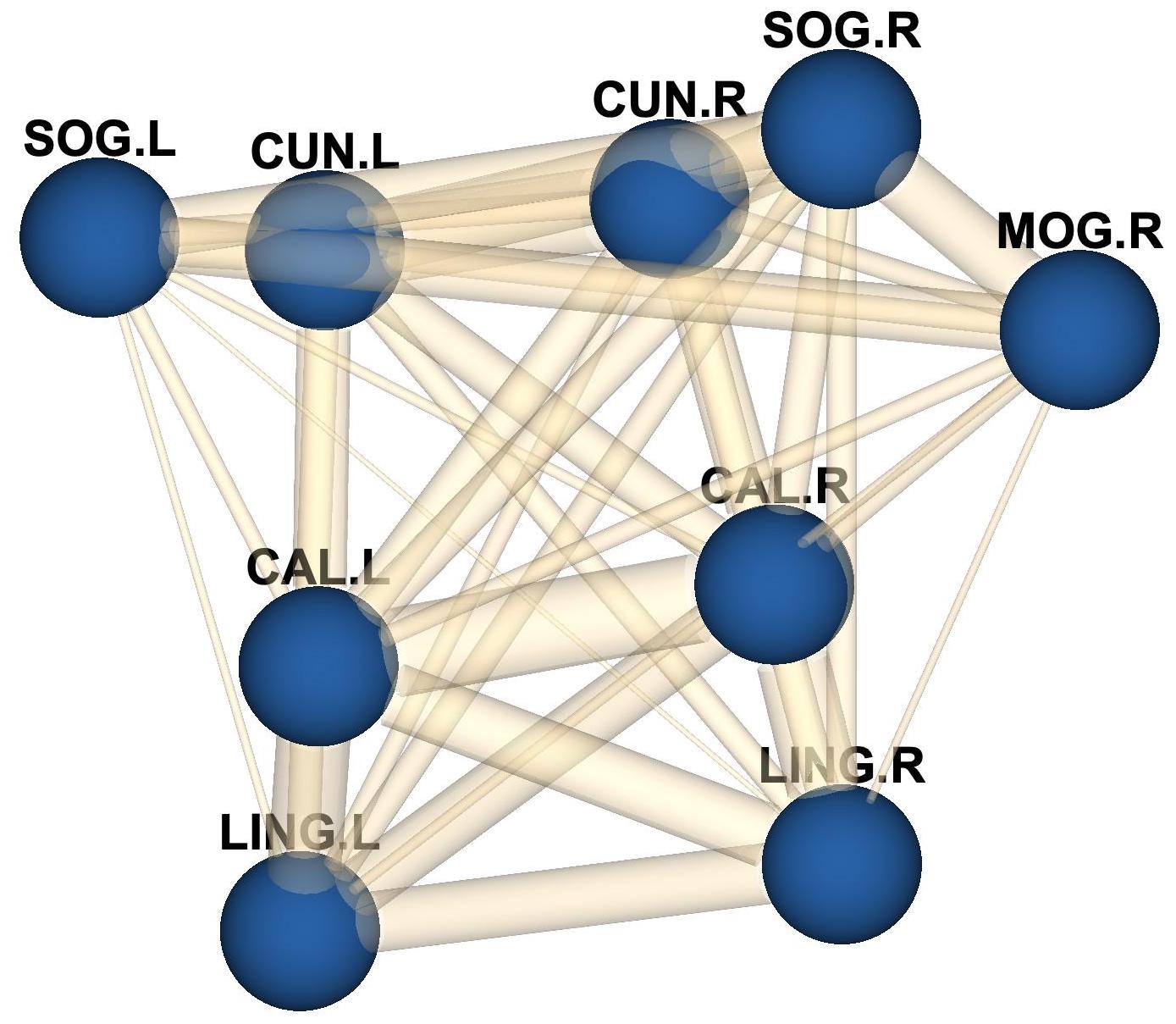}\vspace{-2mm}}
	\vspace{-1mm}
	\caption{\small 3-clique \textsf{MPDS}s in brain networks. The thickness of each edge is proportional to its probability.}
	\label{fig:case_graph}
	\vspace{-5mm}
\end{figure}

\spara{\revise{Comparison with other notions of densest subgraphs in uncertain graphs.}}
\shepherd{Table \ref{tab:pcc} already shows that the \textsf{MPDS} of the \textsf{Karate Club} dataset has a much higher probabilistic clustering coefficient than the other existing dense subgraph notions.} We show that the \textsf{MPDS} also represents a more meaningful and concise community (of club members) than the other subgraphs. All four subgraphs are shown in Figure \ref{fig:case_club}. Notice that the \textsf{MPDS} only contains nodes \revise{from one single ground-truth community \cite{zachary1977information} and has edges with higher probabilities}, in contrast to all the other subgraphs which contain nodes \revise{from both communities and have many low-probability edges}. 
\shepherd{Moreover, in Table~\ref{tab:purity}, we report the average purity (i.e., highest fraction of nodes from the same ground-truth community \cite{zachary1977information} in a node set) of the top-$k$ (up to $k=10$) subgraphs returned by each notion, and observe that MPDSs always achieve 100\% purity. Thus, users can retrieve the top-$k$ MPDSs to identify high-quality communities.}
This case study highlights the importance of computing \textsf{MPDS}s despite other notions of dense subgraphs in uncertain graphs.

\vspace{-1mm}
\subsection{Case Studies: Brain Networks}
\label{sec:case_brain}
\vspace{-1mm}

\begin{figure}[t!]
	\centering
	\subcaptionbox{\revise{\small{Typically developed}} \label{fig:case_brain_1_eds}} [0.235\textwidth]{\includegraphics[scale=0.375]{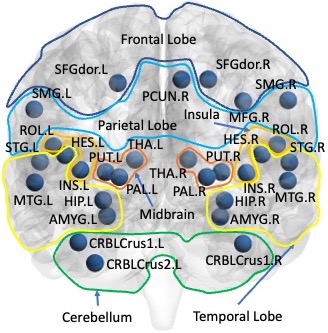}\vspace{-2mm}}
	\subcaptionbox{\revise{\small{ASD-affected}} \label{fig:case_brain_2_eds}} [0.235\textwidth]{\includegraphics[scale=0.375]{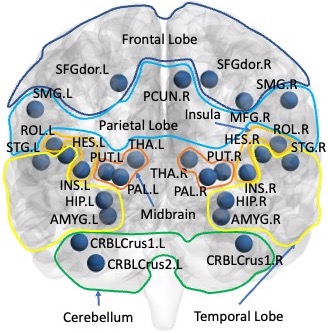}\vspace{-2mm}}
	\vspace{-1mm}
	\caption{\revise{\small Node sets of 3-clique \textsf{EDS}s in brain networks. The colored boundaries denote various brain regions as shown.}}
	\label{fig:case_brain_eds}
	\vspace{-5mm}
\end{figure}

\begin{figure}[t!]
	\centering
	\subcaptionbox{\revise{\small{Typically developed}} \label{fig:case_graph_1_eds}} [0.235\textwidth]{\includegraphics[scale=0.25]{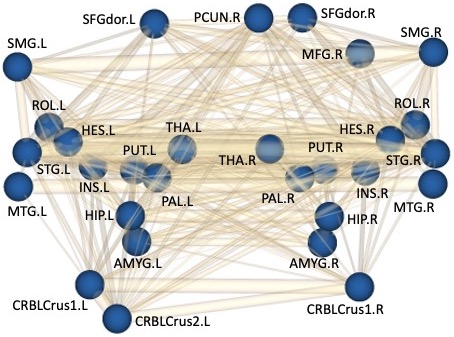}\vspace{-2mm}}
	\subcaptionbox{\revise{\small{ASD-affected}} \label{fig:case_graph_2_eds}} [0.235\textwidth]{\includegraphics[scale=0.25]{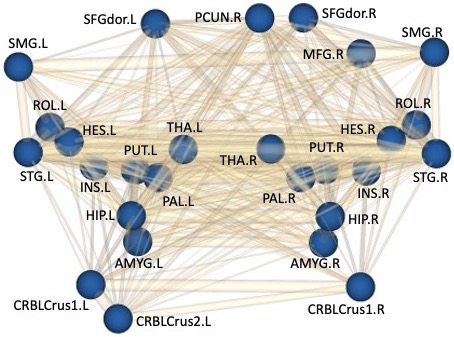}\vspace{-2mm}}
	\vspace{-1mm}
	\caption{\revise{\small 3-clique \textsf{EDS}s in brain networks. The thickness of each edge is proportional to its probability.}}
	\label{fig:case_graph_eds}
	\vspace{-5mm}
\end{figure}

A brain network can be defined as an uncertain graph where nodes are brain regions of interest (ROIs), an edge indicates co-activation between two ROIs, and an edge probability indicates the strength of the co-activation signal. 
Dense subgraphs in brain networks can represent brain regions responding together to stimuli \cite{legenstein_et_al:LIPIcs} or related to diseases \cite{wu2021extracting}.

The dataset we use \cite{conf.fninf.2013} contains data of 52 \textit{Typically Developed} (TD) children and 49 children suffering from \textit{Autism Spectrum Disorder} (ASD). 
Each subject is represented as a graph over 116 nodes (ROIs). $\mathcal{G}_{ASD}$ and $\mathcal{G}_{TD}$ are uncertain graphs, defined over the same set of nodes as the original ones, while the probability of each edge is the average of those of the same edge across all graphs in the ASD and TD groups.

Using BrainNet Viewer \cite{xia2013brainnet}, we show the 3-clique \textsf{MPDS}s for both $\mathcal{G}_{TD}$ and $\mathcal{G}_{ASD}$ in Figures \ref{fig:case_brain} and \ref{fig:case_graph}.
The \textsf{MPDS} in $\mathcal{G}_{ASD}$ lies entirely in the occipital lobe, in contrast to that in $\mathcal{G}_{TD}$, which also contains one node in the temporal lobe and one in the cerebellum. Besides, the \textsf{MPDS} in $\mathcal{G}_{ASD}$ is more symmetrical than that in $\mathcal{G}_{TD}$, since the former has only one node (MOG.R) without its counterpart in the other hemisphere, while the latter has two more such nodes (CRBL6.L and FFG.R). This is consistent with the results of different works in neuroscience indicating that, in contrast to typically developed brains, those affected by ASD are characterized by under-connectivity between distant brain regions and over-connectivity between closer ones \cite{Neu1, Neu2}, and that the hemispheres of ASD-affected brains are more symmetrical than those of typically developed ones \cite{Neur4}. Our consistent findings underline the importance of finding \textsf{MPDS}s in uncertain brain networks that can differentiate healthy and autistic brains.


\begin{figure}[t!]
	\centering
	\subcaptionbox{\shepherd{\small{Typically developed}} \label{fig:case_brain_1_core}} [0.235\textwidth]{\includegraphics[scale=0.375]{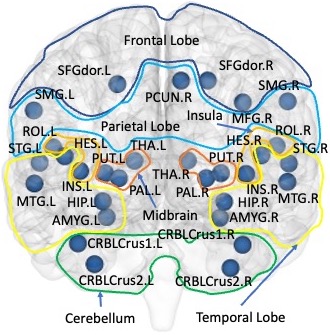}\vspace{-2mm}}
	\subcaptionbox{\shepherd{\small{ASD-affected}} \label{fig:case_brain_2_core}} [0.235\textwidth]{\includegraphics[scale=0.375]{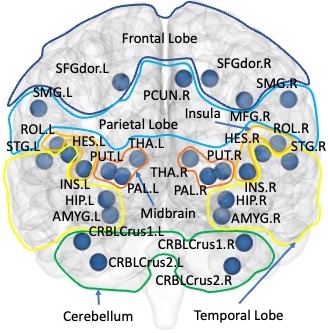}\vspace{-2mm}}
	\vspace{-1mm}
	\caption{\shepherd{\small Node sets of innermost cores in brain networks. The colored boundaries denote various brain regions as shown. Not all nodes are shown due to space constraints in the 2D projection of the 3D brain.}}
	\label{fig:case_brain_core}
	\vspace{-5mm}
\end{figure}

\begin{figure}[t!]
	\centering
	\subcaptionbox{\shepherd{\small{Typically developed}} \label{fig:case_graph_1_core}} [0.235\textwidth]{\includegraphics[scale=0.275]{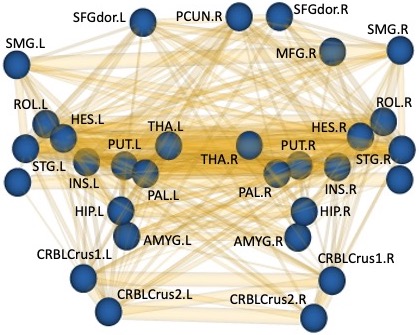}\vspace{-2mm}}
	\subcaptionbox{\shepherd{\small{ASD-affected}} \label{fig:case_graph_2_core}} [0.235\textwidth]{\includegraphics[scale=0.275]{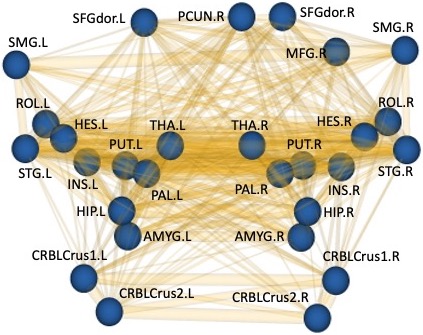}\vspace{-2mm}}
	\vspace{-1mm}
	\caption{\shepherd{\small Innermost cores in brain networks. The thickness of each edge is proportional to its probability. Not all nodes are shown due to space constraints in the 2D projection of the 3D brain.}}
	\label{fig:case_graph_core}
	\vspace{-5mm}
\end{figure}

\begin{figure}[t!]
	\centering
	\subcaptionbox{\shepherd{\small{Typically developed}} \label{fig:case_brain_1_truss}} [0.235\textwidth]{\includegraphics[scale=0.375]{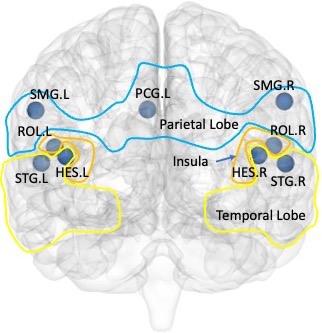}\vspace{-2mm}}
	\subcaptionbox{\shepherd{\small{ASD-affected}} \label{fig:case_brain_2_truss}} [0.235\textwidth]{\includegraphics[scale=0.375]{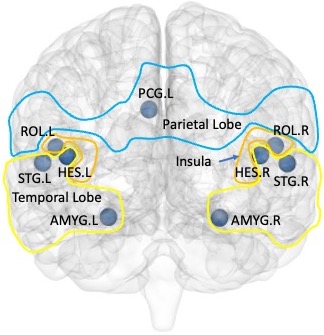}\vspace{-2mm}}
	\vspace{-1mm}
	\caption{\shepherd{\small Node sets of innermost trusses in brain networks. The colored boundaries denote various brain regions as shown.}}
	\label{fig:case_brain_truss}
	\vspace{-5mm}
\end{figure}

\begin{figure}[t!]
	\centering
	\subcaptionbox{\shepherd{\small{Typically developed}} \label{fig:case_graph_1_truss}} [0.235\textwidth]{\includegraphics[scale=0.35]{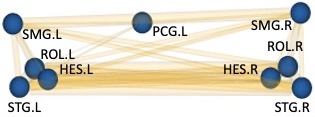}\vspace{-2mm}}
	\subcaptionbox{\shepherd{\small{ASD-affected}} \label{fig:case_graph_2_truss}} [0.235\textwidth]{\includegraphics[scale=0.35]{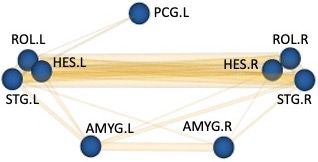}\vspace{-2mm}}
	\vspace{-1mm}
	\caption{\shepherd{\small Innermost trusses in brain networks. The thickness of each edge is proportional to its probability.}}
	\label{fig:case_graph_truss}
	\vspace{-3mm}
\end{figure}

\revise{We now show that the expected (3-clique) densest subgraph (\textsf{EDS}) cannot characterize and distinguish autistic brains, unlike our proposed \textsf{MPDS}. Note that in the existing literature, the \textsf{EDS} notion was defined only for edge density \cite{Zou13}; however, we show that it can be extended to clique and pattern densities (in Appendix \ref{sec:exp_ext}). Figures \ref{fig:case_brain_eds} and \ref{fig:case_graph_eds} present 3-clique \textsf{EDS}s in both $\mathcal{G}_{TD}$ and $\mathcal{G}_{ASD}$. Both of them span as many as 9 brain regions and are similar (w.r.t. symmetry) since both only have the same 3 nodes (PCUN.R, MFG.R, CRBLCrus2.L) without their counterparts in the other hemispheres. 
These contradict the characteristics of autistic brains in the biological literature \cite{Neu1, Neu2, Neur4} and fail to distinguish them from normal ones. This is consistent with our observation in Example \ref{exam:expected} that \textsf{EDS}s tend to be very large and span several unimportant nodes, and hence can be less meaningful in real-world applications. Our \textsf{MPDS} is more powerful than the existing expectation-based notion \textsf{EDS} in analyzing autistic brains.} \shepherd{Also, as shown in Figures \ref{fig:case_brain_core}-\ref{fig:case_graph_core} (resp. \ref{fig:case_brain_truss}-\ref{fig:case_graph_truss}), the innermost cores (resp. trusses) for both brains span multiple brain regions and are similar (w.r.t. symmetry) since they have the same nodes PCUN.R, MFG.R (resp. PCG.L) without their counterparts in the other hemispheres. Thus, even the innermost cores and trusses cannot well-characterize and distinguish autistic brains.}

\begin{figure*}
    \centering
    \vspace{1mm}
    \subcaptionbox{\small{Edge and Clique-\textsf{MPDS}} \label{fig:mpds_time_clique}}[0.235\textwidth] {
\begin{tikzpicture}[thick,xscale=0.575,yscale=0.45, every node/.style={transform shape}]

\definecolor{darkgray176}{RGB}{176,176,176}
\definecolor{lightgray204}{RGB}{204,204,204}
\definecolor{yellow}{RGB}{255,255,0}

\begin{axis}[
legend cell align={left},
legend columns=2,
legend style={
  fill opacity=0.8,
  draw opacity=1,
  text opacity=1,
  at={(0.03,0.97)},
  anchor=north west,
  draw=lightgray204
},
log basis y={10},
tick align=outside,
tick pos=left,
x grid style={darkgray176},
xlabel={Dataset},
xmin=-0.24, xmax=2.84,
xtick style={color=black},
xtick={0.3,1.3,2.3},
xticklabels={KarateClub,IntelLab,LastFM},
y grid style={darkgray176},
ylabel={Running time (seconds)},
ymin=0.01, ymax=1000,
ymode=log,
ytick style={color=black}
]
\draw[draw=black,fill=black] (axis cs:-0.1,0.01) rectangle (axis cs:0.1,0.0324132);
\addlegendimage{ybar,area legend,draw=black,fill=black}
\addlegendentry{edge}

\draw[draw=black,fill=black] (axis cs:0.9,0.01) rectangle (axis cs:1.1,0.847842);
\draw[draw=black,fill=black] (axis cs:1.9,0.01) rectangle (axis cs:2.1,26.7116);
\draw[draw=black,fill=yellow] (axis cs:0.1,0.01) rectangle (axis cs:0.3,0.0374249);
\addlegendimage{ybar,area legend,draw=black,fill=yellow}
\addlegendentry{3-clique}

\draw[draw=black,fill=yellow] (axis cs:1.1,0.01) rectangle (axis cs:1.3,1.17446);
\draw[draw=black,fill=yellow] (axis cs:2.1,0.01) rectangle (axis cs:2.3,31.2467);
\draw[draw=black,fill=white,postaction={pattern=north west lines}] (axis cs:0.3,0.01) rectangle (axis cs:0.5,0.041801);
\addlegendimage{ybar,area legend,draw=black,fill=white,postaction={pattern=north west lines}}
\addlegendentry{4-clique}

\draw[draw=black,fill=white,postaction={pattern=north west lines}] (axis cs:1.3,0.01) rectangle (axis cs:1.5,3.86074);
\draw[draw=black,fill=white,postaction={pattern=north west lines}] (axis cs:2.3,0.01) rectangle (axis cs:2.5,29.8532);
\draw[draw=black,fill=white,postaction={pattern=crosshatch}] (axis cs:0.5,0.01) rectangle (axis cs:0.7,0.0339983);
\addlegendimage{ybar,area legend,draw=black,fill=white,postaction={pattern=crosshatch}}
\addlegendentry{5-clique}

\draw[draw=black,fill=white,postaction={pattern=crosshatch}] (axis cs:1.5,0.01) rectangle (axis cs:1.7,4.5865);
\draw[draw=black,fill=white,postaction={pattern=crosshatch}] (axis cs:2.5,0.01) rectangle (axis cs:2.7,28.9051);
\end{axis}

\end{tikzpicture}}\vspace{-2mm}}
    \subcaptionbox{\small{Pattern-\textsf{MPDS}} \label{fig:mpds_time_pattern}}[0.235\textwidth] {
\begin{tikzpicture}[thick,xscale=0.575,yscale=0.45, every node/.style={transform shape}]

\definecolor{darkgray176}{RGB}{176,176,176}
\definecolor{lightgray204}{RGB}{204,204,204}
\definecolor{yellow}{RGB}{255,255,0}

\begin{axis}[
legend cell align={left},
legend columns=2,
legend style={
  fill opacity=0.8,
  draw opacity=1,
  text opacity=1,
  at={(0.03,0.97)},
  anchor=north west,
  draw=lightgray204
},
log basis y={10},
tick align=outside,
tick pos=left,
x grid style={darkgray176},
xlabel={Dataset},
xmin=-0.24, xmax=2.84,
xtick style={color=black},
xtick={0.3,1.3,2.3},
xticklabels={KarateClub,IntelLab,LastFM},
y grid style={darkgray176},
ylabel={Running time (seconds)},
ymin=0.01, ymax=1000,
ymode=log,
ytick style={color=black}
]
\draw[draw=black,fill=black] (axis cs:-0.1,0.01) rectangle (axis cs:0.1,0.0838423);
\addlegendimage{ybar,area legend,draw=black,fill=black}
\addlegendentry{2-star}

\draw[draw=black,fill=black] (axis cs:0.9,0.01) rectangle (axis cs:1.1,4.99751);
\draw[draw=black,fill=black] (axis cs:1.9,0.01) rectangle (axis cs:2.1,34.6813);
\draw[draw=black,fill=yellow] (axis cs:0.1,0.01) rectangle (axis cs:0.3,0.0482538);
\addlegendimage{ybar,area legend,draw=black,fill=yellow}
\addlegendentry{3-star}

\draw[draw=black,fill=yellow] (axis cs:1.1,0.01) rectangle (axis cs:1.3,6.7791);
\draw[draw=black,fill=yellow] (axis cs:2.1,0.01) rectangle (axis cs:2.3,31.3884);
\draw[draw=black,fill=white,postaction={pattern=north west lines}] (axis cs:0.3,0.01) rectangle (axis cs:0.5,0.03546925);
\addlegendimage{ybar,area legend,draw=black,fill=white,postaction={pattern=north west lines}}
\addlegendentry{c3-star}

\draw[draw=black,fill=white,postaction={pattern=north west lines}] (axis cs:1.3,0.01) rectangle (axis cs:1.5,5.247);
\draw[draw=black,fill=white,postaction={pattern=north west lines}] (axis cs:2.3,0.01) rectangle (axis cs:2.5,30.8374);
\draw[draw=black,fill=white,postaction={pattern=crosshatch}] (axis cs:0.5,0.01) rectangle (axis cs:0.7,0.0337282);
\addlegendimage{ybar,area legend,draw=black,fill=white,postaction={pattern=crosshatch}}
\addlegendentry{diamond}

\draw[draw=black,fill=white,postaction={pattern=crosshatch}] (axis cs:1.5,0.01) rectangle (axis cs:1.7,8.90483);
\draw[draw=black,fill=white,postaction={pattern=crosshatch}] (axis cs:2.5,0.01) rectangle (axis cs:2.7,32.5197);
\end{axis}

\end{tikzpicture}}\vspace{-2mm}}
    \subcaptionbox{\small{Edge and Clique-\textsf{NDS}} \label{fig:nds_time_clique}}[0.235\textwidth] {
\begin{tikzpicture}[thick,xscale=0.575,yscale=0.45, every node/.style={transform shape}]

\definecolor{darkgray176}{RGB}{176,176,176}
\definecolor{lightgray204}{RGB}{204,204,204}
\definecolor{yellow}{RGB}{255,255,0}

\begin{axis}[
legend cell align={left},
legend columns=2,
legend style={
  fill opacity=0.8,
  draw opacity=1,
  text opacity=1,
  at={(0.03,0.97)},
  anchor=north west,
  draw=lightgray204
},
log basis y={10},
tick align=outside,
tick pos=left,
x grid style={darkgray176},
xlabel={Dataset},
xmin=-0.24, xmax=3.84,
xtick style={color=black},
xtick={0.3,1.3,2.3,3.3},
xticklabels={HS,Biomine,Twitter,Friendster},
y grid style={darkgray176},
ylabel={Running time (seconds)},
ymin=1000, ymax=1000000,
ymode=log,
ytick style={color=black}
]
\draw[draw=black,fill=black] (axis cs:-0.1,1000) rectangle (axis cs:0.1,2573.721282);
\addlegendimage{ybar,area legend,draw=black,fill=black}
\addlegendentry{edge}

\draw[draw=black,fill=black] (axis cs:0.9,1000) rectangle (axis cs:1.1,3360.81395507);
\draw[draw=black,fill=black] (axis cs:1.9,1000) rectangle (axis cs:2.1,15366.79786586);
\draw[draw=black,fill=black] (axis cs:2.9,1000) rectangle (axis cs:3.1,18890.72);
\draw[draw=black,fill=yellow] (axis cs:0.1,1000) rectangle (axis cs:0.3,45831.4862004224);
\addlegendimage{ybar,area legend,draw=black,fill=yellow}
\addlegendentry{3-clique}

\draw[draw=black,fill=yellow] (axis cs:1.1,1000) rectangle (axis cs:1.3,20018.8822351);
\draw[draw=black,fill=yellow] (axis cs:2.1,1000) rectangle (axis cs:2.3,128537.807066);
\draw[draw=black,fill=yellow] (axis cs:3.1,1000) rectangle (axis cs:3.3,32176.64);
\draw[draw=black,fill=white,postaction={pattern=north west lines}] (axis cs:0.3,1000) rectangle (axis cs:0.5,70584.4584456415);
\addlegendimage{ybar,area legend,draw=black,fill=white,postaction={pattern=north west lines}}
\addlegendentry{4-clique}

\draw[draw=black,fill=white,postaction={pattern=north west lines}] (axis cs:1.3,1000) rectangle (axis cs:1.5,8370.7240541);
\draw[draw=black,fill=white,postaction={pattern=north west lines}] (axis cs:2.3,1000) rectangle (axis cs:2.5,79569.606196);
\draw[draw=black,fill=white,postaction={pattern=north west lines}] (axis cs:3.3,1000) rectangle (axis cs:3.5,32401.28);
\draw[draw=black,fill=white,postaction={pattern=crosshatch}] (axis cs:0.5,1000) rectangle (axis cs:0.7,86475.9407575804);
\addlegendimage{ybar,area legend,draw=black,fill=white,postaction={pattern=crosshatch}}
\addlegendentry{5-clique}

\draw[draw=black,fill=white,postaction={pattern=crosshatch}] (axis cs:1.5,1000) rectangle (axis cs:1.7,3438.36012006);
\draw[draw=black,fill=white,postaction={pattern=crosshatch}] (axis cs:2.5,1000) rectangle (axis cs:2.7,59923.02928928);
\draw[draw=black,fill=white,postaction={pattern=crosshatch}] (axis cs:3.5,1000) rectangle (axis cs:3.7,32293.44);
\end{axis}

\end{tikzpicture}}\vspace{-2mm}}
    \subcaptionbox{\small{Heuristic Pattern-\textsf{NDS}} \label{fig:nds_time_pattern}}[0.235\textwidth] {
\begin{tikzpicture}[thick,xscale=0.575,yscale=0.45, every node/.style={transform shape}]

\definecolor{darkgray176}{RGB}{176,176,176}
\definecolor{lightgray204}{RGB}{204,204,204}
\definecolor{yellow}{RGB}{255,255,0}

\begin{axis}[
legend cell align={left},
legend columns=2,
legend style={
  fill opacity=0.8,
  draw opacity=1,
  text opacity=1,
  at={(0.03,0.97)},
  anchor=north west,
  draw=lightgray204
},
log basis y={10},
tick align=outside,
tick pos=left,
x grid style={darkgray176},
xlabel={Dataset},
xmin=-0.24, xmax=3.84,
xtick style={color=black},
xtick={0.3,1.3,2.3,3.3},
xticklabels={HS,Biomine,Twitter,Friendster},
y grid style={darkgray176},
ylabel={Running time (seconds)},
ymin=100, ymax=100000,
ymode=log,
ytick style={color=black}
]
\draw[draw=black,fill=black] (axis cs:-0.1,100) rectangle (axis cs:0.1,955.019785882);
\addlegendimage{ybar,area legend,draw=black,fill=black}
\addlegendentry{2-star}

\draw[draw=black,fill=black] (axis cs:0.9,100) rectangle (axis cs:1.1,2681.84);
\draw[draw=black,fill=black] (axis cs:1.9,100) rectangle (axis cs:2.1,3347.235383988);
\draw[draw=black,fill=black] (axis cs:2.9,100) rectangle (axis cs:3.1,34808.06117801541);
\draw[draw=black,fill=yellow] (axis cs:0.1,100) rectangle (axis cs:0.3,850.289022446);
\addlegendimage{ybar,area legend,draw=black,fill=yellow}
\addlegendentry{3-star}

\draw[draw=black,fill=yellow] (axis cs:1.1,100) rectangle (axis cs:1.3,2710.67);
\draw[draw=black,fill=yellow] (axis cs:2.1,100) rectangle (axis cs:2.3,3370.14299202);
\draw[draw=black,fill=yellow] (axis cs:3.1,100) rectangle (axis cs:3.3,35046.56);
\draw[draw=black,fill=white,postaction={pattern=north west lines}] (axis cs:0.3,100) rectangle (axis cs:0.5,11615.90522194);
\addlegendimage{ybar,area legend,draw=black,fill=white,postaction={pattern=north west lines}}
\addlegendentry{c3-star}

\draw[draw=black,fill=white,postaction={pattern=north west lines}] (axis cs:1.3,100) rectangle (axis cs:1.5,3209.21);
\draw[draw=black,fill=white,postaction={pattern=north west lines}] (axis cs:2.3,100) rectangle (axis cs:2.5,4811.86632444);
\draw[draw=black,fill=white,postaction={pattern=north west lines}] (axis cs:3.3,100) rectangle (axis cs:3.5,50056.82715358301);
\draw[draw=black,fill=white,postaction={pattern=crosshatch}] (axis cs:0.5,100) rectangle (axis cs:0.7,1282.567218304);
\addlegendimage{ybar,area legend,draw=black,fill=white,postaction={pattern=crosshatch}}
\addlegendentry{diamond}

\draw[draw=black,fill=white,postaction={pattern=crosshatch}] (axis cs:1.5,100) rectangle (axis cs:1.7,9273.12);
\draw[draw=black,fill=white,postaction={pattern=crosshatch}] (axis cs:2.5,100) rectangle (axis cs:2.7,9136.1439638);
\draw[draw=black,fill=white,postaction={pattern=crosshatch}] (axis cs:3.5,100) rectangle (axis cs:3.7,95078.3379757092);
\end{axis}

\end{tikzpicture}}\vspace{-2mm}}
    \caption{\small Running times of our proposed methods; \textsf{MPDS} for the smaller datasets and \textsf{NDS} for the larger ones; HS denotes \textsf{HomoSapiens}}
    \vspace{-3mm}
    \label{fig:runtime}
\end{figure*}
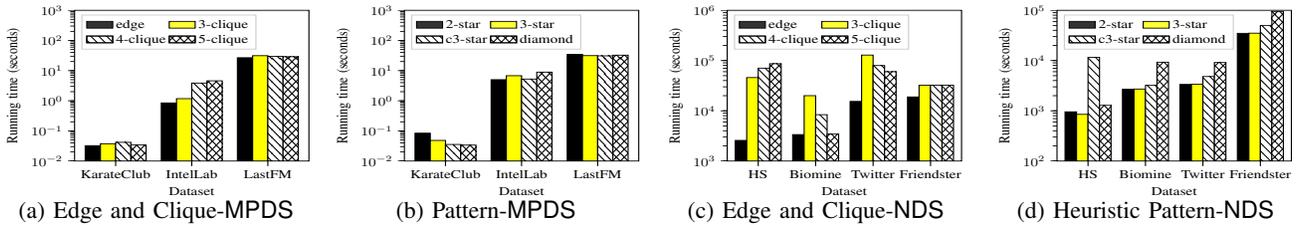

\vspace{-1mm}
\subsection{Efficiency}
\label{sec:large_graph_results}
\vspace{-1mm}
We report the running times of our methods in Figure~\ref{fig:runtime}. As shown in Figures \ref{fig:mpds_time_clique} and \ref{fig:nds_time_clique}, the running times for edge density are smaller than those for clique density. This is because the flow networks involved in computing edge densest subgraphs are much smaller; they only contain nodes for each node in the sampled possible worlds (\S~\ref{sec:top_k_mpds_algorithm}), in contrast to those for $h$-clique densest subgraphs which also contain nodes for each $(h-1)$-clique contained in $h$-cliques in the sampled possible worlds (\S~\ref{sec:clique}). However, there is no clear winner among 3-clique, 4-clique, and 5-clique. This is because, even if larger cliques take a longer time to enumerate, smaller cliques can be more in number, thereby increasing the size of the flow network and hence the running time. Similar arguments can be made for the four patterns (Figures \ref{fig:mpds_time_pattern} and \ref{fig:nds_time_pattern}).

\begin{table}[t]
	\footnotesize
	\centering
	\begin{center}
		\caption{\small Densest subgraph containment probabilities and running times of approximate and heuristic \textsf{Pattern-NDS}; \textsf{Karate Club}}
		\vspace{-1mm}
		\begin{tabular}{c||c|c|c|c}
	\multirow{2}{*}{\textbf{Pattern}} & \multicolumn{2}{|c|}{\textbf{Containment Probability}} & \multicolumn{2}{|c}{\textbf{Running Time (seconds)}} \\ \cline{2-5}
			& \textbf{Approx.} & \textbf{Heuristic} & \textbf{Approx.} & \textbf{Heuristic} \\ \hline \hline
			\textbf{2-star} & {\bf 0.625} & 0.6 & 0.0561 & {\bf 0.0129} \\ \hline
			\textbf{3-star} & {\bf 0.55} & 0.525 & 0.0242 & {\bf 0.0101} \\ \hline
			\textbf{c3-star} & {\bf 0.3313} & 0.262 & 0.0244 & {\bf 0.0109} \\ \hline
			\textbf{diamond} & {\bf 0.8} & 0.7687 & 0.0212 & {\bf 0.0093} 
		\end{tabular}
 		\vspace{-4mm}
		\label{tab:heuristic_lf}
	\end{center}
\end{table}

\begin{table}[t!]
	\centering
	\begin{center}
		\caption{\revise{\small Densest subgraph containment probabilities and running times of approximate and heuristic \textsf{Edge-NDS}; \textsf{Friendster}}}
        \revise{
	    \footnotesize
		\begin{tabular}{c||c|c}
		    \textbf{Method} & \textbf{Containment Probability} & \textbf{Running time (hours)} \\ \hline \hline
		    \textbf{Approximate} & {\bf 0.025} & 21.216 \\ \hline
		    \textbf{Heuristic} & 0.021 & {\bf 4.97} 
		\end{tabular}}
 		\vspace{-4mm}
		\label{tab:heuristic_friendster}
	\end{center}
\end{table}

\begin{table}[t]
	\footnotesize
	\centering
	\begin{center}
		\caption{\small Sampling methods' comparison; {\sf MPDS; Intel Lab}}
		\vspace{-1mm}
		\begin{tabular}{c||c|c|c}
			\textbf{Method} & \textbf{$\theta$}  & \textbf{Running Time (seconds)} & \textbf{Memory (MB)}  \\ \hline \hline
			\textbf{MC} & 160  & 2.233 &  {\bf 2.016} \\ \hline
			\textbf{LP} & 160  & 2.164 &  2.656 \\ \hline
			\textbf{RSS} & {\bf 120}  & {\bf 2.111} &  3.281 
		\end{tabular}
		\vspace{-4mm}
		\label{tab:sampling_methods_intel}
	\end{center}
\end{table}

\begin{table}[t]
	\footnotesize
	\centering
	\begin{center}
		\caption{\small Sampling methods' comparison; {\sf NDS; Biomine}}
		\vspace{-1mm}
		\begin{tabular}{c||c|c|c}
			\textbf{Method} & \textbf{$\theta$}  & \textbf{Running Time (seconds)} & \textbf{Memory (MB)}  \\ \hline \hline
			\textbf{MC} & 640  & 2248 &  {\bf 781} \\ \hline
			\textbf{LP} & 640  & 2178 &  1029 \\ \hline
			\textbf{RSS} & {\bf 600}  & {\bf 2027} &  1516 
		\end{tabular}
		\vspace{-4mm}
		\label{tab:sampling_methods_biomine}
	\end{center}
\end{table}

For patterns on the larger graphs (Figure \ref{fig:nds_time_pattern}), we use our heuristic \textsf{Pattern-NDS} method in place of our approximate one (\S~\ref{sec:nucleus}). In fact, the approximate \textsf{Pattern-NDS} method takes more than 3 days to run on these datasets. To show the quality of the results returned by our heuristic method, Table \ref{tab:heuristic_lf} compares the running times and the (approximate) densest subgraph containment probabilities (as in \S~\ref{sec:comp_others}) of the solutions returned by both methods on the smaller \textsf{Karate Club} dataset. Clearly, for all patterns, the heuristic method returns solutions of comparable quality to those returned by the approximate one, while having a lower running time.

\revise{For our largest dataset \textsf{Friendster}, our approximate methods require higher running times. Thus, we adopt a similar heuristic as in \S~\ref{sec:pattern} to edge and clique densities (which returns, in each sampled world, all subgraphs denser than the innermost core), and run these heuristic methods on \textsf{Friendster}. As shown in Table \ref{tab:heuristic_friendster}, this heuristic method yields reasonably good results, while significantly reducing the running time.}

\vspace{-1mm}
\spara{Varying sampling strategies.} For sampling possible worlds, we compare our employed {\em Monte Carlo (MC)} with {\em  Lazy Propagation (LP)} and {\em Recursive Stratified Sampling (RSS)} (\S~\ref{sec:top_k_mpds_algorithm}).
Tables \ref{tab:sampling_methods_intel} and \ref{tab:sampling_methods_biomine} show that all strategies result in {\em similar sample sizes} $\theta$ at convergence (\S~\ref{sec:parameter}) and have {\em comparable running times}, while MC consumes much less memory.
The other datasets exhibit similar trends and are omitted due to space constraints. Our \textsf{MPDS} and \textsf{NDS} solutions require sampling {\em all} edges in the graph. When applying LP, the visit frequencies of {\em all} edges need to be stored and updated. This reduces the speedup
while increasing the memory usage. RSS has a lower variance than MC; the difference depends on the variances of the estimates in various strata \cite{RSS16}. Unlike reliability queries, our solutions can hardly benefit from the BFS-based edge selection strategy (starting from high-degree nodes), since we need to consider the states of {\em all} edges. Ergo, the estimation variance is not reduced much, leading to a similar (but slightly smaller) sample size $\theta$ \cite{KeKQ19}.
Thus, there is limited speedup on top of the memory overhead due to recursion.
Hence, we adopt MC as the default strategy.


\begin{table}[t!]
	\footnotesize
	\centering
	\begin{center}
		\caption{\small Running times (seconds) of the exact and our approx. \textsf{MPDS} methods on synthetic graphs with number of edges $m$}
		\vspace{-1mm}
		\begin{tabular}{c||c|c|c|c|c|c|c}
			\multirow{2}{*}{\textbf{Graph}} & \multirow{2}{*}{\textbf{$m$}} & \multicolumn{2}{|c|}{\textbf{Edge}} & \multicolumn{2}{|c|}{\textbf{3-Clique}} & \multicolumn{2}{|c}{\textbf{Diamond}} \\ \cline{3-8}
			& & \textbf{Exact} & \textbf{Ours} & \textbf{Exact} & \textbf{Ours} & \textbf{Exact} & \textbf{Ours} \\ \hline \hline
			\textbf{BA\_7} & 13 & 0.172 & {\bf 0.02} & 0.225 & {\bf 0.025} & 0.349 & {\bf 0.025} \\ \hline
			\textbf{BA\_9} & 21 & 58.08 & {\bf 0.04} & 77.264 & {\bf 0.042} & 93.095 & {\bf 0.045} \\ \hline
			\textbf{ER\_7} & 20 & 71.39 & {\bf 0.033} & 78.919 & {\bf 0.036} & 140.361 & {\bf 0.04} \\ \hline
			\textbf{ER\_9} & 30 & 97413 & {\bf 0.048} & 123253 & {\bf 0.054} & 273557 & {\bf 0.064} 
		\end{tabular}
		\vspace{-3mm}
		\label{tab:exact_time}
	\end{center}
\end{table}

\vspace{-1mm}
\subsection{Comparison with the Exact Algorithms}
\label{sec:comp_exact}
%
We compare the effectiveness and efficiency of our proposed approximate top-$k$ \textsf{MPDS} algorithms with those of the corresponding exact algorithms. \revise{Since the computation is $\sharpP$-hard (\ref{sec:hardness})},
the exact top-$k$ \textsf{MPDS} methods need to compute the densest subgraphs in all $2^m$ possible worlds, which
is infeasible for the datasets in Table \ref{tab:dataset}. Thus, we generate four small synthetic graphs according to the Erd{\" o}s-R{\'e}nyi (ER) \cite{erdHos1960evolution} and Barab{\'a}si-Albert (BA) \cite{barabasi1999emergence} models, and assign edge probabilities uniformly at random. We denote by ER\_$n$ (resp. BA\_$n$) the ER (resp. BA) graph with $n$ nodes. The number $m$ of edges of the generated datasets, along with the running times of the exact and our proposed algorithms for edge, 3-clique, and diamond \textsf{MPDS}, are shown in Table \ref{tab:exact_time}. Clearly, for each notion of density on each dataset, the exact method takes some orders of magnitude longer time than our method.

We next show the effectiveness. For $k = 1$, in all cases, our method returns the same result as the exact one. For $k \in \{5, 10\}$, we compute, for each rank, the F1-score of the node set returned by our method with respect to the exact one as the ground truth. Figure \ref{fig:exact_prec} shows these scores averaged across all ranks from $1$ to $k$. Clearly, the scores are reasonably high in all cases, which implies a high accuracy for each method.

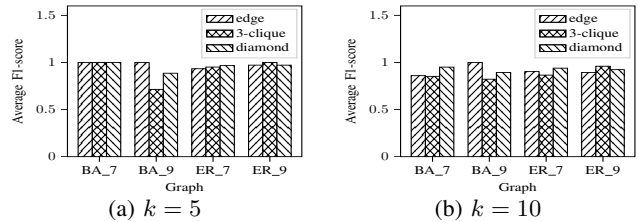
\begin{figure}[t]
	\centering
	\subcaptionbox{\small{$k = 5$} \label{fig:exact_prec_5}} [0.238\textwidth]{
\begin{tikzpicture}[thick,xscale=0.575,yscale=0.45, every node/.style={transform shape}]

\definecolor{darkgray176}{RGB}{176,176,176}
\definecolor{lightgray204}{RGB}{204,204,204}

\begin{axis}[
legend cell align={left},
legend style={fill opacity=0.8, draw opacity=1, text opacity=1, draw=lightgray204},
tick align=outside,
tick pos=left,
x grid style={darkgray176},
xlabel={Graph},
xmin=-0.3125, xmax=3.8125,
xtick style={color=black},
xtick={0.25,1.25,2.25,3.25},
xticklabels={BA\_7,BA\_9,ER\_7,ER\_9},
y grid style={darkgray176},
ylabel={Average F1-score},
ymin=0, ymax=1.6,
ytick style={color=black}
]
\draw[draw=black,fill=white,postaction={pattern=north east lines}] (axis cs:-0.125,0) rectangle (axis cs:0.125,1);
\addlegendimage{ybar,area legend,draw=black,fill=white,postaction={pattern=north east lines}}
\addlegendentry{edge}

\draw[draw=black,fill=white,postaction={pattern=north east lines}] (axis cs:0.875,0) rectangle (axis cs:1.125,1);
\draw[draw=black,fill=white,postaction={pattern=north east lines}] (axis cs:1.875,0) rectangle (axis cs:2.125,0.933333333333);
\draw[draw=black,fill=white,postaction={pattern=north east lines}] (axis cs:2.875,0) rectangle (axis cs:3.125,0.971428571429);
\draw[draw=black,fill=white,postaction={pattern=crosshatch}] (axis cs:0.125,0) rectangle (axis cs:0.375,1);
\addlegendimage{ybar,area legend,draw=black,fill=white,postaction={pattern=crosshatch}}
\addlegendentry{3-clique}

\draw[draw=black,fill=white,postaction={pattern=crosshatch}] (axis cs:1.125,0) rectangle (axis cs:1.375,0.714285714286);
\draw[draw=black,fill=white,postaction={pattern=crosshatch}] (axis cs:2.125,0) rectangle (axis cs:2.375,0.951282051282);
\draw[draw=black,fill=white,postaction={pattern=crosshatch}] (axis cs:3.125,0) rectangle (axis cs:3.375,1);
\draw[draw=black,fill=white,postaction={pattern=north west lines}] (axis cs:0.375,0) rectangle (axis cs:0.625,1);
\addlegendimage{ybar,area legend,draw=black,fill=white,postaction={pattern=north west lines}}
\addlegendentry{diamond}

\draw[draw=black,fill=white,postaction={pattern=north west lines}] (axis cs:1.375,0) rectangle (axis cs:1.625,0.885714285714);
\draw[draw=black,fill=white,postaction={pattern=north west lines}] (axis cs:2.375,0) rectangle (axis cs:2.625,0.966666666667);
\draw[draw=black,fill=white,postaction={pattern=north west lines}] (axis cs:3.375,0) rectangle (axis cs:3.625,0.971428571429);
\end{axis}

\end{tikzpicture}}\vspace{-2mm}}
	\subcaptionbox{\small{$k = 10$} \label{fig:exact_prec_10}} [0.238\textwidth]{
\begin{tikzpicture}[thick,xscale=0.575,yscale=0.45, every node/.style={transform shape}]

\definecolor{darkgray176}{RGB}{176,176,176}
\definecolor{lightgray204}{RGB}{204,204,204}

\begin{axis}[
legend cell align={left},
legend style={fill opacity=0.8, draw opacity=1, text opacity=1, draw=lightgray204},
tick align=outside,
tick pos=left,
x grid style={darkgray176},
xlabel={Graph},
xmin=-0.3125, xmax=3.8125,
xtick style={color=black},
xtick={0.25,1.25,2.25,3.25},
xticklabels={BA\_7,BA\_9,ER\_7,ER\_9},
y grid style={darkgray176},
ylabel={Average F1-score},
ymin=0, ymax=1.6,
ytick style={color=black}
]
\draw[draw=black,fill=white,postaction={pattern=north east lines}] (axis cs:-0.125,0) rectangle (axis cs:0.125,0.86);
\addlegendimage{ybar,area legend,draw=black,fill=white,postaction={pattern=north east lines}}
\addlegendentry{edge}

\draw[draw=black,fill=white,postaction={pattern=north east lines}] (axis cs:0.875,0) rectangle (axis cs:1.125,1);
\draw[draw=black,fill=white,postaction={pattern=north east lines}] (axis cs:1.875,0) rectangle (axis cs:2.125,0.90303030303);
\draw[draw=black,fill=white,postaction={pattern=north east lines}] (axis cs:2.875,0) rectangle (axis cs:3.125,0.893214285714);
\draw[draw=black,fill=white,postaction={pattern=crosshatch}] (axis cs:0.125,0) rectangle (axis cs:0.375,0.852380952381);
\addlegendimage{ybar,area legend,draw=black,fill=white,postaction={pattern=crosshatch}}
\addlegendentry{3-clique}

\draw[draw=black,fill=white,postaction={pattern=crosshatch}] (axis cs:1.125,0) rectangle (axis cs:1.375,0.822177822178);
\draw[draw=black,fill=white,postaction={pattern=crosshatch}] (axis cs:2.125,0) rectangle (axis cs:2.375,0.86554001554);
\draw[draw=black,fill=white,postaction={pattern=crosshatch}] (axis cs:3.125,0) rectangle (axis cs:3.375,0.96);
\draw[draw=black,fill=white,postaction={pattern=north west lines}] (axis cs:0.375,0) rectangle (axis cs:0.625,0.950505050505);
\addlegendimage{ybar,area legend,draw=black,fill=white,postaction={pattern=north west lines}}
\addlegendentry{diamond}

\draw[draw=black,fill=white,postaction={pattern=north west lines}] (axis cs:1.375,0) rectangle (axis cs:1.625,0.892857142857);
\draw[draw=black,fill=white,postaction={pattern=north west lines}] (axis cs:2.375,0) rectangle (axis cs:2.625,0.939393939394);
\draw[draw=black,fill=white,postaction={pattern=north west lines}] (axis cs:3.375,0) rectangle (axis cs:3.625,0.923736263736);
\end{axis}

\end{tikzpicture}}\vspace{-2mm}}
	\vspace{-2mm}
	\caption{\small F1-score (averaged across all ranks from 1 to $k$) of the top-$k$ node sets returned by our \textsf{MPDS} methods w.r.t. the exact ones}
	\label{fig:exact_prec}
	\vspace{-3mm}
\end{figure}

\vspace{-0.7mm}
\spara{\revise{Varying edge probabilities.}}
\revise{To show the impact of the edge probabilities on the effectiveness of our method, we generate 3 synthetic graphs by assigning normally distributed edge probabilities with means 0.2, 0.5 and 0.8 to ER\_7. For each graph, we run both the exact and our approximate \textsf{MPDS} methods. Figure \ref{fig:prob} above shows that our method returns reasonably good results for all edge probability distributions. The running time of our method is longer for larger values of the mean, since the edge probabilities and hence the possible worlds sampled in each round are larger.}

\begin{figure}[t]
	\centering
	\subcaptionbox{\revise{\small{Running Time}} \label{fig:time}} [0.235\textwidth]{
	\small
	\revise{
	\begin{tabular}{c||c}
	    Mean & Running Time (sec.) \\ \hline\hline
	    0.2 & 0.0899594 \\ \hline
	    0.5 & 0.102649 \\ \hline
	    0.8 & 0.121513 
	\end{tabular}}}
	\subcaptionbox{\revise{\small{Average F1-score}} \label{fig:f1}} [0.235\textwidth]{
\begin{tikzpicture}[thick,xscale=0.575,yscale=0.45, every node/.style={transform shape}]

\definecolor{darkgray176}{RGB}{176,176,176}
\definecolor{lightgray204}{RGB}{204,204,204}

\begin{axis}[
legend cell align={left},
legend style={fill opacity=0.8, draw opacity=1, text opacity=1, draw=lightgray204},
tick align=outside,
tick pos=left,
x grid style={darkgray176},
xlabel={Mean edge probability},
xmin=-0.3125, xmax=2.8125,
xtick style={color=black},
xtick={0.25,1.25,2.25},
xticklabels={0.2,0.5,0.8},
y grid style={darkgray176},
ylabel={Average F1-score},
ymin=0, ymax=1.6,
ytick style={color=black}
]
\draw[draw=black,fill=white,postaction={pattern=north east lines}] (axis cs:-0.125,0) rectangle (axis cs:0.125,1);
\addlegendimage{ybar,area legend,draw=black,fill=white,postaction={pattern=north east lines}}
\addlegendentry{$k=1$}

\draw[draw=black,fill=white,postaction={pattern=north east lines}] (axis cs:0.875,0) rectangle (axis cs:1.125,1);
\draw[draw=black,fill=white,postaction={pattern=north east lines}] (axis cs:1.875,0) rectangle (axis cs:2.125,1);
\draw[draw=black,fill=white,postaction={pattern=crosshatch}] (axis cs:0.125,0) rectangle (axis cs:0.375,0.981818181818);
\addlegendimage{ybar,area legend,draw=black,fill=white,postaction={pattern=crosshatch}}
\addlegendentry{$k=5$}

\draw[draw=black,fill=white,postaction={pattern=crosshatch}] (axis cs:1.125,0) rectangle (axis cs:1.375,1);
\draw[draw=black,fill=white,postaction={pattern=crosshatch}] (axis cs:2.125,0) rectangle (axis cs:2.375,1);
\draw[draw=black,fill=white,postaction={pattern=north west lines}] (axis cs:0.375,0) rectangle (axis cs:0.625,0.941818181818);
\addlegendimage{ybar,area legend,draw=black,fill=white,postaction={pattern=north west lines}}
\addlegendentry{$k=10$}

\draw[draw=black,fill=white,postaction={pattern=north west lines}] (axis cs:1.375,0) rectangle (axis cs:1.625,0.96);
\draw[draw=black,fill=white,postaction={pattern=north west lines}] (axis cs:2.375,0) rectangle (axis cs:2.625,0.90303030303);
\end{axis}

\end{tikzpicture}}\vspace{-2mm}}
	\vspace{-2mm}
	\caption{\revise{\small Running time of our \textsf{MPDS} method and F1-score (averaged across all ranks from 1 to $k$) of the top-$k$ node sets returned by our \textsf{MPDS} method w.r.t. the exact ones for synthetic edge probabilities}}
	\label{fig:prob}
	\vspace{-5mm}
\end{figure}
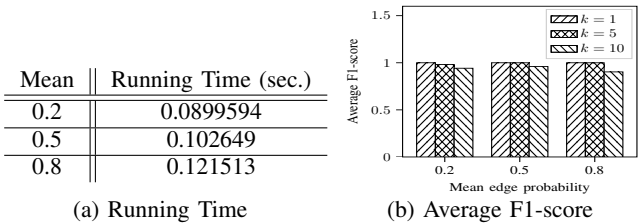

\vspace{-1mm}
\subsection{Parameter Sensitivity}
\label{sec:parameter}
\vspace{-1mm}


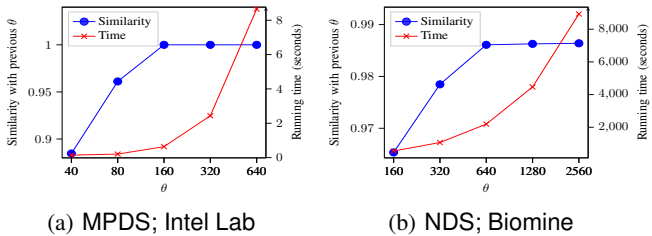
\begin{figure}[t!]
	\centering
	\subcaptionbox{\small{\textsf{MPDS; Intel Lab}} \label{fig:theta_mpds}} [0.235\textwidth]{
\begin{tikzpicture}[thick,xscale=0.5,yscale=0.45, every node/.style={transform shape}]

\definecolor{darkgray176}{RGB}{176,176,176}
\definecolor{lightgray204}{RGB}{204,204,204}

\begin{axis}[
legend cell align={left},
legend style={
  fill opacity=0.8,
  draw opacity=1,
  text opacity=1,
  at={(0.03,0.97)},
  anchor=north west,
  draw=lightgray204
},
tick align=outside,
tick pos=left,
x grid style={darkgray176},
xlabel={\(\displaystyle \theta\)},
xmin=-0.2, xmax=4.2,
xtick style={color=black},
xtick={0,1,2,3,4},
xticklabels={40,80,160,320, 640},
y grid style={darkgray176},
ylabel={Similarity with previous $\theta$},
ymin=0.88, ymax=1.04,
ytick style={color=black}
]
\addplot [semithick, blue, mark=*, mark size=3, mark options={solid}]
table {%
0 0.8846153846153846
1 0.9611320754716981
2 1
3 1
4 1
};
\label{mpds_similarity}
\end{axis}

\begin{axis}[
axis y line=right,
legend cell align={left},
legend style={
  fill opacity=0.8,
  draw opacity=1,
  text opacity=1,
  at={(0.03,0.97)},
  anchor=north west,
  draw=lightgray204
},
tick align=outside,
x grid style={darkgray176},
xmin=-0.2, xmax=4.2,
xtick pos=left,
xtick style={color=black},
xtick={0,1,2,3,4},
xticklabels={40,80,160,320,640},
y grid style={darkgray176},
ylabel={Running time (seconds)},
ymin=-0.01582355000000001, ymax=8.76257255,
ytick pos=right,
ytick style={color=black},
yticklabel style={anchor=west}
]
\addlegendimage{/pgfplots/refstyle=mpds_similarity}\addlegendentry{Similarity}
\addplot [semithick, red, mark=x, mark size=3, mark options={solid}]
table {%
0 0.143197
1 0.203725
2 0.638251
3 2.43623
4 8.65375
};
\label{mpds_time}
\addlegendentry{Time}
\end{axis}

\end{tikzpicture}}}
	\subcaptionbox{\small{\textsf{NDS; Biomine}} \label{fig:theta_nds}} [0.235\textwidth]{
\begin{tikzpicture}[thick,xscale=0.5,yscale=0.45, every node/.style={transform shape}]

\definecolor{darkgray176}{RGB}{176,176,176}
\definecolor{lightgray204}{RGB}{204,204,204}

\begin{axis}[
legend cell align={left},
legend style={
  fill opacity=0.8,
  draw opacity=1,
  text opacity=1,
  at={(0.03,0.97)},
  anchor=north west,
  draw=lightgray204
},
tick align=outside,
tick pos=left,
x grid style={darkgray176},
xlabel={\(\displaystyle \theta\)},
xmin=-0.2, xmax=4.2,
xtick style={color=black},
xtick={0,1,2,3,4},
xticklabels={160,320,640,1280,2560},
y grid style={darkgray176},
ylabel={Similarity with previous $\theta$},
ymin=0.96429147651045, ymax=0.9933648151655,
ytick style={color=black}
]
\addplot [semithick, blue, mark=*, mark size=3, mark options={solid}]
table {%
0 0.965338976738
1 0.978472262113
2 0.986100051489
3 0.986281134433
4 0.986388981289
};
\label{nds_similarity}
\end{axis}

\begin{axis}[
axis y line=right,
legend cell align={left},
legend style={
  fill opacity=0.8,
  draw opacity=1,
  text opacity=1,
  at={(0.03,0.97)},
  anchor=north west,
  draw=lightgray204
},
tick align=outside,
x grid style={darkgray176},
xmin=-0.2, xmax=4.2,
xtick pos=left,
xtick style={color=black},
xtick={0,1,2,3,4},
xticklabels={160,320,640,1280,2560},
y grid style={darkgray176},
ylabel={Running time (seconds)},
ymin=127.8710238453, ymax=9350.3656705007,
ytick pos=right,
ytick style={color=black},
yticklabel style={anchor=west}
]
\addlegendimage{/pgfplots/refstyle=nds_similarity}\addlegendentry{Similarity}
\addplot [semithick, red, mark=x, mark size=3, mark options={solid}]
table {%
0 547.075325966
1 1067.48044801
2 2192.31315804
3 4465.58068419
4 8931.16136838
};
\label{nds_time}
\addlegendentry{Time}
\end{axis}

\end{tikzpicture}}}
	\vspace{-2mm}
	\caption{\small Variation, w.r.t. $\theta$, of the running time and the similarity of the returned node sets to those for the previous value of $\theta$}
	\label{fig:theta}
	\vspace{-5mm}
\end{figure}


\begin{figure}[t!]
	\centering
	\subcaptionbox{\revise{\small Varying $k$} \label{fig:vary_k}} [0.235\textwidth]{
\begin{tikzpicture}[thick,xscale=0.575,yscale=0.45, every node/.style={transform shape}]

\definecolor{darkgray176}{RGB}{176,176,176}
\definecolor{darkorange25512714}{RGB}{255,127,14}
\definecolor{forestgreen4416044}{RGB}{44,160,44}
\definecolor{lightgray204}{RGB}{204,204,204}
\definecolor{steelblue31119180}{RGB}{31,119,180}

\begin{axis}[
legend cell align={left},
legend style={fill opacity=0.8, draw opacity=1, text opacity=1, draw=lightgray204, at={(0.01, 0.01)}, anchor=south west},
tick align=outside,
tick pos=left,
x grid style={darkgray176},
xlabel={\(\displaystyle k\)},
xmin=-0.2, xmax=4.2,
xtick style={color=black},
xtick={0, 1, 2, 3, 4},
xticklabels={1, 5, 10, 50, 100},
y grid style={darkgray176},
ylabel={Avg. Est. Containment Prob.},
ymin=0.965, ymax=1.003,
ytick style={color=black}
]
\addplot [semithick, steelblue31119180, mark=*, mark size=3, mark options={solid}]
table {%
0 1
1 0.999218830313
2 0.998671803281
3 0.991765672078
4 0.979281246234
};
\addlegendentry{HomoSapiens}
\addplot [semithick, darkorange25512714, mark=triangle*, mark size=3, mark options={solid,rotate=180}]
table {%
0 1
1 0.9984376125
2 0.9976564
3 0.9958439775
4 0.995062759375
};
\addlegendentry{Biomine}
\addplot [semithick, forestgreen4416044, mark=x, mark size=3, mark options={solid}]
table {%
0 1
1 0.98750016
2 0.98125009875
3 0.97149992625
4 0.968999871375
};
\addlegendentry{Twitter}
\end{axis}

\end{tikzpicture}}}
	\subcaptionbox{\revise{\small Varying $l_m$; \textsf{Homo Sapiens}} \label{fig:vary_lm_hs}} [0.235\textwidth]{
\begin{tikzpicture}[thick,xscale=0.575,yscale=0.45, every node/.style={transform shape}]

\definecolor{darkgray176}{RGB}{176,176,176}
\definecolor{darkorange25512714}{RGB}{255,127,14}
\definecolor{forestgreen4416044}{RGB}{44,160,44}
\definecolor{lightgray204}{RGB}{204,204,204}
\definecolor{steelblue31119180}{RGB}{31,119,180}

\begin{axis}[
legend cell align={left},
legend style={fill opacity=0.8, draw opacity=1, text opacity=1, draw=lightgray204, at={(0.01, 0.01)}, anchor=south west},
tick align=outside,
tick pos=left,
x grid style={darkgray176},
xlabel={\(\displaystyle l_m\)},
xmin=-0.2, xmax=4.2,
xtick style={color=black},
xtick={0, 1, 2, 3, 4},
xticklabels={200, 210, 220, 230, 240},
y grid style={darkgray176},
ylabel={Avg. Est. Containment Prob.},
ymin=0, ymax=1,
ytick style={color=black}
]
\addplot [semithick, steelblue31119180, mark=*, mark size=3, mark options={solid}]
table {%
0 0.998437420313
1 0.998437420313
2 0.830546858203
3 0.0886718808281
4 0
};
\end{axis}

\end{tikzpicture}}}
	\vspace{-2mm}
	\caption{\revise{\small Variation of the average estimated densest subgraph containment probability with $k$ and $l_m$ for \textsf{NDS} queries}}
	\label{fig:vary_nds}
	\vspace{-5mm}
\end{figure}
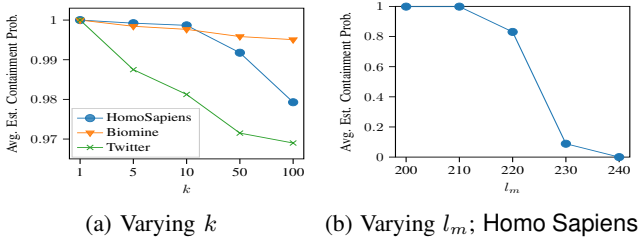

\vspace{-1mm}
\spara{Varying $\theta$.} We study the effects of the variation of $\theta$ in Figure \ref{fig:theta}.
As in \S~\ref{sec:large_graph_results}, although we show the results for only two datasets, the other datasets exhibit similar trends.
For \textsf{MPDS} on the \textsf{Intel Lab} dataset (Figure \ref{fig:theta_mpds}), increasing $\theta$ steadily increases the similarity of the returned node sets to those for the previous value of $\theta$ till a certain point ($\theta = 160$), after which it converges,
while the running time keeps increasing. Similar effects can be observed at $\theta = 640$ for \textsf{NDS} on the \textsf{Biomine} dataset (Figure \ref{fig:theta_nds}). We choose such values of $\theta$ as default for the respective datasets in our experiments.

\vspace{-1mm}
\spara{\revise{Varying $k$ for top-$k$ \textsf{NDS}s.}}
\revise{Figure \ref{fig:vary_k} shows that increasing the value of $k$ results in the reduction of the average estimated densest subgraph containment probability. This shows that increasing the value of $k$ too much results in the returned node sets being of lower quality.}

\spara{\revise{Varying $l_m$.}}
\revise{Figure \ref{fig:vary_lm_hs} shows the variation of the average estimated densest subgraph containment probability with $l_m$. Initially, till a certain value of $l_m$, the probability remains constant; since the returned node sets should be closed, our algorithm avoids reporting too small node sets even for smaller values of $l_m$. After that, the probability keeps decreasing with $l_m$ till a certain value beyond which it remains $0$ as there is no larger closed node set. This helps to choose a feasible upper bound on $l_m$. Note that Figure \ref{fig:vary_lm_hs} shows the results for the \textsf{Homo Sapiens} dataset; the other datasets exhibit similar trends but over different ranges of values of $l_m$, and are omitted due to space constraints.}



\section{Conclusions}
\label{sec:conclusions}
We studied the novel problem of finding the Most Probable Densest Subgraph (\textsf{MPDS}) in an uncertain graph, according to edge, clique and pattern densities. \revise{We proved that computing the densest subgraph probability for any given node set is $\sharpP$-hard.} We proposed a solution which returns the most frequent densest subgraphs from some sampled possible worlds, with theoretical accuracy guarantees. As building blocks, we designed novel algorithms to compute all clique- and pattern-densest subgraphs in a deterministic graph. We then extended our algorithm to compute the Nucleus Densest Subgraphs (\textsf{NDS}) via reduction to \revise{closed} frequent itemset mining. Our experiments on large real-world graphs showed that our methods are efficient. We depicted that our methods are reasonably accurate compared to the exact ones, while being orders of magnitude faster, using some small synthetic graphs. Moreover, we showed that our \textsf{MPDS} is significantly different from existing notions of dense subgraphs in uncertain graphs. Our case studies showcased the usefulness of the \textsf{MPDS} in differentiating autistic brains from healthy ones and in detecting useful communities in social networks.


\appendices

\section{Proof of Correctness of Algorithm \ref{alg:clique_densest}}
\label{sec:clique_proof}

\begin{table}[t]
	\footnotesize
	\centering
	\begin{center}
		\caption{\small Table of notations in appendix}
		\vspace{-1mm}
		\begin{tabular}{c||c}
		    \hline
		    {\bf Notation} & {\bf Description} \\ \hline \hline
			$\Lambda$ & Set of all $(h-1)$-cliques contained in $h$-cliques in $G$ \\ \hline
			$\mu_h(G)$ & Number of $h$-cliques in $G$ \\ \hline
			$\rho_h^*$ & Maximum $h$-clique density of any subgraph of $G$ \\ \hline
			$\mathcal{H}$ & The critical flow network (Algorithm \ref{alg:flow}) \\ \hline
			$c(\mathcal{S}, \mathcal{T})$ & Capacity of $s$-$t$ cut $(\mathcal{S}, \mathcal{T})$ in $\mathcal{H}$ \\ \hline
			$f^*$ &  A maximum flow in $\mathcal{H}$ \\ \hline
			$\mathcal{H}_{f^*}$ & Residual graph of $\mathcal{H}$ under $f^*$ \\ \hline
			$\mathcal{H}^C$ & Graph of strongly connected components of $\mathcal{H}_{f^*}$ \\ \hline
		\end{tabular}
		\vspace{-5mm}
		\label{tab:notations}
	\end{center}
\end{table}

We prove that Algorithm \ref{alg:clique_densest} correctly computes all $h$-clique densest subgraphs of $G$ as below:

{\bf (1) } Lines \ref{line:lb}-\ref{line:core} are justified by Lemma \ref{lem:core} and the fact that $\Tilde{\rho}$ is a lower bound on $\rho_h^*$, the optimal $h$-clique density of a subgraph in $G$ \cite{FYCLX19}.

{\bf (2)} Line \ref{line:max_flow} of Algorithm \ref{alg:clique_densest} computes a maximum flow $f^*$ in $\mathcal{H}$. We show that the capacity of such flow can capture the number of $h$-cliques in $G$ (Lemma~\ref{lem:cut} and Corollary~\ref{cor:mincut}). Thus, it facilitates the densest subgraph finding (Lemma~\ref{lem:densest_cut}).

{\bf (3)} We show the properties of $\mathcal{H}_{f^*}$ regarding the source node $s$ (Lemma \ref{lem:sv}) and the sink node $t$ (Lemmas \ref{lem:tv}, \ref{lem:tc}), along with their SCCs (Lemma \ref{lem:prop}). These illustrate why we do not consider the SCCs of $s$ and $t$ in Line \ref{line:scc} of Algorithm \ref{alg:clique_densest}.

{\bf (4)} We prove that all densest subgraphs can be determined exactly once via enumerating all independent sets (Corollary~\ref{cor:ic_densest}), along with several definitions (Definitions~\ref{def:nt}-\ref{def:dcc}) and auxiliary lemmas (Lemmas~\ref{lem:ic_dcc}-\ref{lem:dcc_densest}).

\begin{lem}
[\cite{FYCLX19}]
\label{lem:core}
Given a graph $G = (V, E)$ and an integer $h > 0$, the $h$-clique densest subgraph of $G$ is contained in its $(\lceil \rho_h^* \rceil, h)$-core, where $\rho_h^* = \max_{S \subseteq V} \rho_h(S)$.
\end{lem}


Line \ref{line:begin_flow} of Algorithm \ref{alg:clique_densest} constructs a flow network by the method in \cite{MPPTX15}. For the sake of completeness, the pseudocode is shown in Algorithm \ref{alg:flow}. Specifically, $\mathcal{H}$ has a source node $s$, a sink node $t$, and one node for each $v \in V$ and each $\lambda \in \Lambda$. Each edge $e$ in $\mathcal{H}$ has a non-negative capacity $c(e)$. The edges can be classified into four categories, as shown in Line 3, Algorithm \ref{alg:flow}. A flow in $\mathcal{H}$ is a function $f : E_\mathcal{H} \rightarrow \mathbbm{R}$ that satisfies the following properties:
\begin{itemize}
    \item \textbf{Capacity:} $f(e) \leq c(e) \; \forall e \in E_\mathcal{H}$
    \item \textbf{Antisymmetry:} $f(v, u) = - f(u, v) \; \forall (u, v) \in E_\mathcal{H}$
    \item \textbf{Conservation:} For each $u \in V_\mathcal{H} \setminus \{s, t\}$, $\sum_{(u, v) \in E_\mathcal{H}} f(u, v) = \sum_{(v, u) \in E_\mathcal{H}} f(v, u) = 0$
\end{itemize}
The value of a flow $f$ is defined as the total flow on the edges leaving $s$, i.e., $\sum_{v \in V_\mathcal{H} \, : \, (s, v) \in E_\mathcal{H}} f(s, v)$.

\begin{algorithm}[t]
	\renewcommand{\algorithmicrequire}{\textbf{Input:}}
	\renewcommand{\algorithmicensure}{\textbf{Output:}}
	\scriptsize
	\caption{\small Construct the flow network for clique density \cite{MPPTX15}}
    \label{alg:flow}
	\begin{algorithmic}[1]
		\REQUIRE Graph $G = (V,E)$, set of $(h-1)$-cliques $\Lambda$, parameter $\alpha$
		\ENSURE Flow network $\mathcal{H} = (V_\mathcal{H}, E_\mathcal{H}, c)$
		\STATE $V_\mathcal{H} \gets V \cup \Lambda \cup \{s, t\}$, $E_\mathcal{H} \gets \emptyset$
		\FORALL{$v \in V$}
		\STATE $E_{\mathcal{H}} \gets E_{\mathcal{H}} \cup \{(s, v), (v, t), (v, s), (t, v)\}$
		\STATE Set $c(s, v) = deg_G(v, h)$, $c(v, t) = h \alpha$, $c(v, s) = c(t, v) = 0$
		\ENDFOR
		\FORALL{$\lambda \in \Lambda$} \label{line:begin_flow}
		\FORALL{$v \in \lambda$}
		\STATE $E_{\mathcal{H}} \gets E_{\mathcal{H}} \cup \{(\lambda, v), (v, \lambda)\}$
		\STATE Set $c(\lambda, v) = +\infty$, $c(v, \lambda) = 0$
		\ENDFOR
		\FORALL{$v \in V$}
		\IF{$\lambda$ and $v$ form an $h$-clique}
		\STATE $E_{\mathcal{H}} \gets E_{\mathcal{H}} \cup \{(v, \lambda), (\lambda, v)\}$
		\STATE Set $c(v, \lambda) = 1$, $c(\lambda, v) = 0$
		\ENDIF
		\ENDFOR
		\ENDFOR \label{line:end_flow}
		\RETURN $(V_\mathcal{H}, E_\mathcal{H}, c)$
	\end{algorithmic}
\end{algorithm}

Line \ref{line:max_flow} of Algorithm \ref{alg:clique_densest} computes a maximum flow $f^*$ in $\mathcal{H}$. We show in Lemma \ref{lem:densest_cut} that this facilitates the $h$-clique densest subgraph detection. To this end, we first compute the value of a maximum flow (and hence the capacity of a minimum $s$-$t$ cut) in $\mathcal{H}$. We prove this in Corollary \ref{cor:mincut} using Lemma \ref{lem:cut} as a building block. Lemma \ref{lem:cut} computes, for a given node set $V_1 \subseteq V$, the minimum capacity of an $s$-$t$ cut in $\mathcal{H}$ with all nodes in $V_1$ (resp. $V \setminus V_1$) on the same side as $s$ (resp. $t$).

\begin{lem}
\label{lem:cut}
For a set $V_1 \subseteq V$, the minimum value of an $s$-$t$ cut $(\mathcal{S}, V_\mathcal{H} \setminus \mathcal{S})$ of $\mathcal{H}$ having $\mathcal{S} \cap V = V_1$ is
\begin{small}
\begin{equation*}
    \min_{\mathcal{S} : s \in \mathcal{S}, t \notin \mathcal{S}, \mathcal{S} \cap V = V_1} c \left( \mathcal{S}, V_\mathcal{H} \setminus \mathcal{S} \right) = h \left[ \mu_h(G) + \left[ \rho_h^* - \rho_h(V_1) \right] |V_1| \right]
\end{equation*}
\end{small}
Moreover, one value of $\mathcal{S}$ achieving the minimum above has $\Lambda_1 = \mathcal{S} \cap \Lambda$ as the set of all $(h-1)$-cliques in $\Lambda$ which consist only of nodes in $V_1$.
\end{lem}
\begin{proof}
Similar to that of Theorem 2 in \cite{MPPTX15}.
\end{proof}

\begin{corry}
\label{cor:mincut}
The capacity of a minimum $s$-$t$ cut $(\mathcal{S}^*, V_\mathcal{H} \setminus \mathcal{S}^*)$ of $\mathcal{H}$ is $h \mu_h(G)$.
\end{corry}
\begin{proof}
From Lemma \ref{lem:cut}, and noticing that $\rho_h(V_1) \leq \rho_h^*$ for any $V_1 \subseteq V$, the capacity of a minimum $s$-$t$ cut $(\mathcal{S}^*, V_\mathcal{H} \setminus \mathcal{S}^*)$ of $\mathcal{H}$ is
\begin{small}
\begin{align*}
    c(\mathcal{S}^*, V_\mathcal{H} \setminus \mathcal{S}^*) &= \min_{V_1 \subseteq V} \min_{\mathcal{S} : s \in \mathcal{S}, t \notin \mathcal{S}, \mathcal{S} \cap V = V_1} c \left( \mathcal{S}, V_\mathcal{H} \setminus \mathcal{S} \right) \\
    &= h \left[ \mu_h(G) + \min_{V_1 \subseteq V} \left[ \rho_h^* - \rho_h(V_1) \right] |V_1| \right] \\
    &= h \mu_h(G)
\end{align*}
\end{small}
\end{proof}

We now show the relationship between a maximum flow in $\mathcal{H}$ (Line \ref{line:max_flow} of Algorithm \ref{alg:clique_densest}) and the densest subgraphs in $G$.

\begin{lem}
\label{lem:densest_cut}
Let $V_1$ be a non-empty subset of $V$. Denote by $\Lambda_1$ the set of all $(h-1)$-cliques in $\Lambda$ consisting only of nodes in $V_1$, and let $\mathcal{S} = \{s\} \cup V_1 \cup \Lambda_1$. The following statements are equivalent.
\begin{enumerate}
    \item $V_1$ induces a densest subgraph in $G$.
    \item $\left( \mathcal{S}, V_\mathcal{H} \setminus \mathcal{S} \right)$ is a minimum $s$-$t$ cut in $\mathcal{H}$.
    \item There is no edge from $\mathcal{S}$ to $V_\mathcal{H} \setminus \mathcal{S}$ in $\mathcal{H}_{f^*}$.
\end{enumerate}
\end{lem}
\begin{proof}
Similar to that of Lemma 3 in \cite{CQ20}.
\end{proof}

To enumerate the densest subgraphs in $G$ following Lemma \ref{lem:densest_cut}, we decompose the directed graph $\mathcal{H}_{f^*}$ into strongly connected components (SCCs). By contracting each SCC of $\mathcal{H}_{f^*}$ into a super-node, we obtain $\mathcal{H}^C$, the SCC graph of $\mathcal{H}_{f^*}$, which is a directed acyclic graph. However, as mentioned in Line \ref{line:scc} of Algorithm \ref{alg:clique_densest}, we do not consider the SCCs of $s$ and $t$. In the following, we justify this based on some properties of $\mathcal{H}_{f^*}$ regarding the source node $s$ (Lemma \ref{lem:sv}) and the sink node $t$ (Lemmas \ref{lem:tv} and \ref{lem:tc}), along with their SCCs (Lemma \ref{lem:prop}).

\begin{lem}
\label{lem:sv}
In the residual graph $\mathcal{H}_{f^*}$, there is no edge from the source node $s$ to any $u \in V_\mathcal{H} \setminus \{s\}$, and there is an edge from every $v \in V$ to $s$.
\end{lem}

\begin{proof}
Similar to that of Lemma 4 in \cite{CQ20}.
%
\end{proof}

\begin{lem}
\label{lem:tv}
In the residual graph $\mathcal{H}_{f^*}$, there is a path from the sink node $t$ to every $v \in V$.
\end{lem}

\begin{proof}
Similar to that of Lemma 5 in \cite{CQ20}.
\end{proof}

\begin{lem}
\label{lem:tc}
In the residual graph $\mathcal{H}_{f^*}$, there is a path from the sink node $t$ to every $\lambda \in \Lambda$.
\end{lem}
\begin{proof}
Consider any arbitrary $\lambda \in \Lambda$. Suppose, by way of contradiction, that there is no path from $t$ to $\lambda$ in $\mathcal{H}_{f^*}$. In that case, there is no edge in $\mathcal{H}_{f^*}$ entering $\lambda$, i.e., the corresponding edge in $\mathcal{H}$ is saturated, as such edges can only be from nodes $v \in V$, and from Lemma \ref{lem:tv}, there is a path from $t$ to every node $v \in V$ in $\mathcal{H}_{f^*}$. Now, for every $v \in \lambda$, $c(v, \lambda) = 0$, implying that $f^*(v, \lambda) = 0$. Also, for every $v \in V$ that forms an $h$-clique with $\lambda$, $c(v, \lambda) = 1$, implying that $f^*(v, \lambda) = 1$. (Note that there must exist such $v \in V$, since $\Lambda$ contains only those $(h-1)$-cliques that are contained in $h$-cliques.) This means the total flow entering $\lambda$ is positive, which violates the conservation property of a flow, leading to a contradiction.
\end{proof}

\begin{lem}
\label{lem:prop}
Let $scc(s)$ and $scc(t)$ denote the SCCs of $\mathcal{H}_{f^*}$ containing $s$ and $t$, respectively. The SCC graph $\mathcal{H}^C$ of $\mathcal{H}_{f^*}$ has the following properties.
\begin{itemize}
    \item $scc(s) = \{s\}$
    \item $scc(s)$ has no outgoing edge
    \item $scc(t)$ has no incoming edge
\end{itemize}
\end{lem}
\begin{proof}
The first two properties can be proved directly from Lemma \ref{lem:sv}. We prove the third by contradiction. Suppose $scc(t)$ has an incoming edge in $\mathcal{H}^C$ from another SCC $C$. Consider any node $u \in C$. Then $u \neq t$ and there is a path from $u$ to $t$ in $\mathcal{H}_{f^*}$. Clearly, $u \neq s$, which means $u \in V \cup \Lambda$. From Lemmas \ref{lem:tv} and \ref{lem:tc}, there is a path from $t$ to $u$ in $\mathcal{H}_{f^*}$. In that case, $u$ and $t$ are in the same SCC, which is a contradiction.
\end{proof}

A densest subgraph of $G$ cannot have any node in the SCC of $s$ (Lemma \ref{lem:prop}) or $t$ (Lemma \ref{lem:densest_cut}). Moreover, removing the SCCs of $s$ and $t$ from $\mathcal{H}^C$ does not affect the connections between the other SCCs (Lemma \ref{lem:prop}), which can contain a densest subgraph. Therefore, we can focus entirely on the other SCCs, as in line \ref{line:scc} of Algorithm \ref{alg:clique_densest}. We denote them as non-trivial components (Definition \ref{def:nt}).

\begin{defn}
[Non-trivial Component]
\label{def:nt}
A non-trivial component is a strongly connected component of $\mathcal{H}_{f^*}$ which does not include the source node $s$ or the target node $t$.
\end{defn}

To enumerate all densest subgraphs of $G$, we define the descendants and ancestors of non-trivial components (Definition \ref{def:desc_anc}) and introduce the independence among such components (Definition \ref{def:ic}).

\begin{defn}
[Descendants and Ancestors]
\label{def:desc_anc}
For a non-trivial component $C$, the set of descendants (resp. ancestors) of $C$, denoted by $des(C)$ (resp. $anc(C)$), is defined as the set of those non-trivial components $C'$ such that there exists a directed path from $C$ to $C'$ (resp. from $C'$ to $C$) in $\mathcal{H}^C$. For a set $\mathcal{C}$ of non-trivial components, the sets of descendants and ancestors of $\mathcal{C}$, denoted by $des(\mathcal{C})$ and $anc(\mathcal{C})$, respectively, are defined as $des(\mathcal{C}) = \cup_{C \in \mathcal{C}} des(C)$ and $anc(\mathcal{C}) = \cup_{C \in \mathcal{C}} anc(C)$.
\end{defn}

\begin{defn}
[Independent Component Set]
\label{def:ic}
An independent component set is a set $\mathcal{C}$ of non-trivial components such that:
\begin{itemize}
    \item For all $C \in \mathcal{C}$, $C \cap V \neq \emptyset$.
    \item For all $C_1, C_2 \in \mathcal{C}$, $C_1 \notin des(C_2)$ and $C_2 \notin des(C_1)$.
\end{itemize}
\end{defn}

It turns out that enumerating all independent component sets is equivalent to enumerating all densest subgraphs in $G$ (Corollary \ref{cor:ic_densest}). To prove this, we first define the concept of d-closed component sets (Definition \ref{def:dcc}).

\begin{defn}
[d-Closed Component Set]
\label{def:dcc}
A d-closed component set is a set $\mathcal{C}$ of non-trivial components such that:
\begin{itemize}
    \item For all $C \in \mathcal{C}$ where $C \cap V = \emptyset$, there is at least one incoming edge in $\mathcal{H}^C$ to $C$ from a component $C' \in \mathcal{C}$.
    \item $des(\mathcal{C}) \subseteq \mathcal{C}$.
\end{itemize}
\end{defn}


We now show the existence of a bijection between independent sets and d-closed component sets (Lemma \ref{lem:ic_dcc}), as well as one between d-closed component sets and densest subgraphs in $G$ (Lemma \ref{lem:dcc_densest}), thereby proving Corollary \ref{cor:ic_densest}.

\begin{lem}
\label{lem:ic_dcc}
For an independent component set $\mathcal{C}$, define $g(\mathcal{C}) = \mathcal{C} \cup des(\mathcal{C})$. Then $g$ is a bijection between the set of independent component sets and the set of d-closed component sets of $\mathcal{H}^C$.
\end{lem}

\begin{proof}
Notice that the function $g$ is well-defined, since $g(\mathcal{C})$ is a d-closed component set uniquely determined by the independent component set $\mathcal{C}$. It suffices to show that, for each d-closed component set $\mathcal{C}_1$, there is a unique independent component set $\mathcal{C}$ such that $\mathcal{C}_1 = g(\mathcal{C}) = \mathcal{C} \cup des(\mathcal{C})$.

Let $\mathcal{C}_2$ be the set of components in $\mathcal{C}_1$ with no incoming edges in $\mathcal{H}^C$ from any component in $\mathcal{C}_1$. We shall show that $\mathcal{C}_2$ is our required set by proving that {\bf (1)} $g(\mathcal{C}_2) = \mathcal{C}_1$; {\bf (2)} $\mathcal{C}_2$ is independent; and {\bf (3)} $\mathcal{C}_2$ is the unique set satisfying the above.

To see why $g(\mathcal{C}_2) = \mathcal{C}_2 \cup des(\mathcal{C}_2) = \mathcal{C}_1$, note that $\mathcal{C}_2 \subseteq \mathcal{C}_1$ and $\mathcal{C}_2 \cap des(\mathcal{C}_2) = \emptyset$ by construction. Thus, it is enough to show that $des(\mathcal{C}_2) = \mathcal{C}_1 \setminus \mathcal{C}_2$. If $C \in des(\mathcal{C}_2)$, then $C \in \mathcal{C}_1 \setminus \mathcal{C}_2$ since $\mathcal{C}_1$ is d-closed; thus $des(\mathcal{C}_2) \subseteq \mathcal{C}_1 \setminus \mathcal{C}_2$. Also, if $C \in \mathcal{C}_1 \setminus \mathcal{C}_2$, then $C$ has an incoming edge from a component $C' \in \mathcal{C}_1$. If $C' \in \mathcal{C}_2$, then clearly $C \in des(\mathcal{C}_2)$; otherwise $C'$ has an incoming edge from a component $C'' \in \mathcal{C}_1$. Continuing in this way, since $\mathcal{H}^C$ is acyclic, $C$ has a path from a component in $\mathcal{C}_2$, which means $C \in des(\mathcal{C}_2)$, and hence $ \mathcal{C}_1 \setminus \mathcal{C}_2 \subseteq des(\mathcal{C}_2)$.

Now we prove that $\mathcal{C}_2$ is independent. First note that $\mathcal{C}_2$ only consists of non-trivial components $C$ satisfying $C \cap V \neq \emptyset$. To see why, consider $C \in \mathcal{C}_1$ such that $C \cap V = \emptyset$. Since $\mathcal{C}_1$ is d-closed, there is at least one incoming edge in $\mathcal{H}^C$ to $C$ from a component in $\mathcal{C}_1$, which means $C \notin \mathcal{C}_2$ by construction.
Now suppose, by way of contradiction, that there exist $C', C'' \in \mathcal{C}_2 \subseteq \mathcal{C}_1$ where $C' \in des(C_2)$. Then there is a directed path from $C''$ to $C'$ in $\mathcal{H}^C$. Let $C'''$ be the immediate predecessor of $C'$ on that path. Since $\mathcal{C}_1$ is d-closed, $C''' \in \mathcal{C}_1$. In that case, the edge from $C'''$ to $C'$ means that, by construction, $C' \notin \mathcal{C}_2$, which is a contradiction.

It remains to be proved that $\mathcal{C}_2$ is the unique set satisfying the above properties. Let $\mathcal{C}_3$ be another independent component set such that $g(\mathcal{C}_3) = \mathcal{C}_1$. We shall show that $\mathcal{C}_2 = \mathcal{C}_3$.

If $C \in \mathcal{C}_2 \subseteq \mathcal{C}_1$, suppose, by way of contradiction, that $C \notin \mathcal{C}_3$. Since $\mathcal{C}_1 = \mathcal{C}_3 \cup des(\mathcal{C}_3)$, $C \in des(\mathcal{C}_3)$. By definition, there is a path in $\mathcal{H}^C$ to $C$ from a component $C' \in \mathcal{C}_3 \subseteq \mathcal{C}_1$. Let $C''$ be the immediate predecessor of $C$ on that path. Since $\mathcal{C}_3$ is independent, $C'' \notin \mathcal{C}_3$. Since $\mathcal{C}_1$ is d-closed, $C'' \in \mathcal{C}_1$. Thus, there exists an edge in $\mathcal{H}^C$ to $C \in \mathcal{C}_2$ from $C'' \in \mathcal{C}_1$, which contradicts the construction of $\mathcal{C}_2$. Thus $C \in \mathcal{C}_3$, and hence $\mathcal{C}_2 \subseteq \mathcal{C}_3$.

If $C \notin \mathcal{C}_2$, suppose, by way of contradiction, that $C \in \mathcal{C}_3 \subseteq \mathcal{C}_1$. By definition, $C$ has an incoming edge in $\mathcal{H}^C$ from a component $C' \in \mathcal{C}_1$. Clearly, $C'$ is an ancestor of $C$. Since $\mathcal{C}_1 = \mathcal{C}_3 \cup des(\mathcal{C}_3)$, $C' \in \mathcal{C}_3$, which contradicts the assumption that $\mathcal{C}_3$ is independent. Thus $C \notin \mathcal{C}_3$, and hence $\mathcal{C}_3 \subseteq \mathcal{C}_2$.
\end{proof}

\begin{lem}
\label{lem:dcc_densest}
For a d-closed component set $\mathcal{C}$
, define $g(\mathcal{C}) = \bigcup_{C \in \mathcal{C}} C \cap V$. Then $g$ is a bijection between the set of d-closed component sets of $\mathcal{H}^C$ and the set of densest subgraphs of $G$.
\end{lem}

\begin{proof}
For any node $u \in V_\mathcal{H}$, let $scc(u)$ denote the unique strongly connected component in $\mathcal{H}^C$ containing $u$.

We first prove that the function $g$ is well defined. Clearly, $g(\mathcal{C})$ is uniquely determined by $\mathcal{C}$. Since $\mathcal{C}$ is d-closed, there is no edge in $\mathcal{H}^C$ from a non-trivial component in $\mathcal{C}$ to any non-trivial component not in $\mathcal{C}$. Also, from Lemma \ref{lem:prop}, $scc(t)$ has no incoming edge and $scc(s)$ has no outgoing edge. In other words, there is no edge in $\mathcal{H}_{f^*}$ from any node in $\{s\} \cup \left( \bigcup_{C \in \mathcal{C}} C \right)$ to any other node. Thus, from Lemma \ref{lem:densest_cut}, $g(\mathcal{C})$ induces a densest subgraph in $G$.

Next we show that, for any $V_1 \subseteq V$ inducing a densest subgraph in $G$, there is a unique d-closed component set $\mathcal{C}$ such that $V_1 = g(\mathcal{C}) = \bigcup_{C \in \mathcal{C}} C \cap V$. Denote by $\Lambda_1$ the set of all $(h-1)$-cliques in $\Lambda$ consisting only of nodes in $V_1$. Define $\mathcal{C}_1 = \bigcup_{u \in V_1 \cup \Lambda_1} scc(u)$, and $\mathcal{C}_2 \subseteq \mathcal{C}_1$ as the set $\mathcal{C}_1$ excluding those components $C \in \mathcal{C}_1$ with $C \cap V = \emptyset$ and no edge in $\mathcal{H}^C$ to $C$ from any component in $\mathcal{C}_1$. We shall show that $\mathcal{C}_2$ is our required component set by proving that {\bf (1)} $g(\mathcal{C}_2) = V_1$; {\bf (2)} $\mathcal{C}_2$ is d-closed; and {\bf (3)} $\mathcal{C}_2$ is the unique set satisfying the above properties.

We first prove that $g(\mathcal{C}_2) = V_1$. Notice that every $v \in V_1$ is contained in some component in $\mathcal{C}_1$. Also, no node $v \in V \setminus V_1$ is contained in any component in $\mathcal{C}_1$ (see Claim \ref{clm:all_v}). Since $\mathcal{C}_2 \subseteq \mathcal{C}_1$ by construction, the same holds for $\mathcal{C}_2$ as well. Now it is easy to see that $\bigcup_{C \in \mathcal{C}_2} C \cap V = V_1$.

Second, we prove that $\mathcal{C}_2$ is d-closed. Note that by construction, for all $C \in \mathcal{C}_2$ where $C \cap V = \emptyset$, there is at least one incoming edge in $\mathcal{H}^C$ to $C$ from a component $C' \in \mathcal{C}_2$. Now, for showing that $des(\mathcal{C}_2) \subseteq \mathcal{C}_2$, it suffices to show that there is no edge (and hence no path) in $\mathcal{H}^C$ from a component in $C \in \mathcal{C}_2$ to a non-trivial component $C' \notin \mathcal{C}_2$. All components in $\mathcal{C}_1$ (and hence $\mathcal{C}_2$, since $\mathcal{C}_2 \subseteq \mathcal{C}_1$) consist only of nodes in $V_1 \cup \Lambda_1$, and from Lemma \ref{lem:densest_cut}, there is no edge from such nodes to those contained in components not in $\mathcal{C}_1$. Also, by construction, no component in $\mathcal{C}_1$ (and hence $\mathcal{C}_2$) has any outgoing edge to any component in $\mathcal{C}_1 \setminus \mathcal{C}_2$.

Finally, we prove that $\mathcal{C}_2$ is the unique set satisfying the above two properties. Let $\mathcal{C}_3$ be another set such that $g(\mathcal{C}_3) = V_1$ and $\mathcal{C}_3$ is d-closed; we shall show that $\mathcal{C}_2 = \mathcal{C}_3$. For $i \in \{2, 3\}$, define $\mathcal{C}_i^{(V)} = \{C \in \mathcal{C}_i : C \cap V \neq \emptyset\}$ and $\mathcal{C}_i^{(\Lambda)} = \mathcal{C}_i \setminus \mathcal{C}_i^{(V)} = \{C \in \mathcal{C}_i : C \cap V = \emptyset\}$. Clearly, it suffices to prove that $\mathcal{C}_2^{(V)} = \mathcal{C}_3^{(V)}$ and $\mathcal{C}_2^{(\Lambda)} = \mathcal{C}_3^{(\Lambda)}$. The first one is trivial, since $\bigcup_{C \in \mathcal{C}_2} C \cap V = \bigcup_{C \in \mathcal{C}_3} C \cap V = V_1$. For the second, we prove that $\mathcal{C}_2^{(\Lambda)} \subseteq \mathcal{C}_3^{(\Lambda)}$ and $\mathcal{C}_3^{(\Lambda)} \subseteq \mathcal{C}_2^{(\Lambda)}$.

If $C \in \mathcal{C}_2^{(\Lambda)}$, then $C$ must have an incoming edge from a component $C' \in \mathcal{C}_2$ (since $\mathcal{C}_2$ is d-closed) with $C' \cap V \neq \emptyset$ (by the construction of the flow network $\mathcal{H}$). In that case, $C' \in \mathcal{C}_2^{(V)} = \mathcal{C}_3^{(V)} \subseteq \mathcal{C}_3$. Since $\mathcal{C}_3$ is d-closed, $C \in \mathcal{C}_3$; and as $C \notin \mathcal{C}_2^{(V)} = \mathcal{C}_3^{(V)}$, $C \in \mathcal{C}_3^{(\Lambda)}$. Hence $\mathcal{C}_2^{(\Lambda)} \subseteq \mathcal{C}_3^{(\Lambda)}$.

If $C \notin \mathcal{C}_2^{(\Lambda)}$, suppose, by way of contradiction, that $C \in \mathcal{C}_3^{(\Lambda)}$. By definition, $C \subseteq \mathcal{C}_3$ and $C \subseteq \Lambda$. Since $\mathcal{C}_3$ is d-closed, $C$ must have an incoming edge from a component $C' \in \mathcal{C}_3$. Also, by the construction of the flow network $\mathcal{H}$, $C' \cap V \neq \emptyset$, which means $C' \in \mathcal{C}_3^{(V)} = \mathcal{C}_2^{(V)} \subseteq \mathcal{C}_2$. Since $\mathcal{C}_2$ is d-closed, $C \in \mathcal{C}_2$. As $C \cap V = \emptyset$, $C \in \mathcal{C}_2^{(\Lambda)}$, which is a contradiction. Thus $C \notin \mathcal{C}_3^{(\Lambda)}$, and hence $\mathcal{C}_3^{(\Lambda)} \subseteq \mathcal{C}_2^{(\Lambda)}$.
\end{proof}

\begin{clm}
\label{clm:all_v}
For $V_1 \subseteq V$ inducing a densest subgraph in $G$, denote by $\Lambda_1$ the set of all $(h-1)$-cliques in $\Lambda$ consisting only of nodes in $V_1$. For any non-trivial component $C$, either $C \subseteq V_1 \cup \Lambda_1$ or $C \cap (V_1 \cup \Lambda_1) = \emptyset$.
\end{clm}
\begin{proof}
Suppose, by way of contradiction, that $C$ has nodes $u_1 \in V_1 \cup \Lambda_1$ and $u_2 \in V_\mathcal{H} \setminus (V_1 \cup \Lambda_1 \cup \{s, t\})$. Since $C$ is strongly connected, there must exist a path from $u_1$ to $u_2$ in $\mathcal{H}_{f^*}$. Then there is an edge from a node in $V_1 \cup \Lambda_1$ to one in $V_\mathcal{H} \setminus (V_1 \cup \Lambda_1 \cup \{s, t\})$, which is not possible (Lemma \ref{lem:densest_cut}).
\end{proof}

\begin{corry}
\label{cor:ic_densest}
For an independent component set $\mathcal{C}$, define $g(\mathcal{C}) = \bigcup_{C \in \mathcal{C} \cup des(\mathcal{C})} C \cap V$. Then $g$ is a bijection between the set of independent component sets and the set of densest subgraphs of $G$.
\end{corry}
\begin{proof}
Let $g_1$ and $g_2$ be the bijections in Lemmas \ref{lem:ic_dcc} and \ref{lem:dcc_densest} respectively. Clearly $g_2(g_1(\mathcal{C})) = g(\mathcal{C})$ is a bijection.
\end{proof}

Following Corollary \ref{cor:ic_densest}, we enumerate all densest subgraphs (Line \ref{line:end_enum_all} of Algorithm \ref{alg:clique_densest}) by exploring all independent component sets, which invokes Algorithm \ref{alg:enum} by taking an empty set $\emptyset$, the set of all non-trivial components in $\mathcal{H}^C$, and $V$ as inputs. Specifically, Algorithm \ref{alg:enum} recursively selects an independent component set $\mathcal{C}_1$ while storing in $\mathcal{C}_2$, the components that are independent with all components in $\mathcal{C}_1$. When a component $C$ is moved from $\mathcal{C}_2$ to $\mathcal{C}_1$, all descendants and ancestors of $C$ along with $C$ itself are removed from $\mathcal{C}_2$ and then a recursion is invoked (line \ref{line:rec}); this guarantees that  $\mathcal{C}_1$ is independent. Note that the removal of $C$ itself (line \ref{line:rem}) ensures that each independent component set, and hence each densest subgraph, is enumerated exactly once.

\section{Proof of Correctness of Algorithm \ref{alg:pattern_densest}}
\label{sec:pattern_proof}

The correctness of Algorithm \ref{alg:pattern_densest} can be proved in a way similar to that of Algorithm \ref{alg:clique_densest} (Appendix \ref{sec:clique_proof}). The main difference lies in the construction of the flow network $\mathcal{H}$, which is done using the method in \cite{FYCLX19}. For reference, its pseudocode is shown in Algorithm \ref{alg:flow_pat}. Note that this algorithm does not construct a node for every $\psi$-instance; rather, it constructs one node for each group $g$ of $\psi$-instances with the same set of nodes in $V$. Hereafter, we denote by $\Lambda'$ the set of all such $\psi$-instance groups. Since this algorithm is different from its analogue in Appendix \ref{sec:clique_proof}, so is the derivation of the value of a maximum flow in the constructed network $\mathcal{H}$. This is shown in Lemma \ref{lem:cut_pat}.

\begin{algorithm}[t]
	\renewcommand{\algorithmicrequire}{\textbf{Input:}}
	\renewcommand{\algorithmicensure}{\textbf{Output:}}
	\scriptsize
	\caption{\small Construct the flow network for pattern density \cite{FYCLX19}}
    \label{alg:flow_pat}
	\begin{algorithmic}[1]
		\REQUIRE Graph $G = (V,E)$, pattern $\psi = (V_\psi, E_\psi)$, set of $\psi$-instances $\Lambda$, parameter $\alpha$
		\ENSURE Flow network $\mathcal{H} = (V_\mathcal{H}, E_\mathcal{H}, c)$
		\STATE $V_\mathcal{H} \gets V \cup \{s, t\}$, $E_\mathcal{H} \gets \emptyset$
		\FORALL{$v \in V$}
		\STATE $E_{\mathcal{H}} \gets E_{\mathcal{H}} \cup \{(s, v), (v, t), (v, s), (t, v)\}$
		\STATE Set $c(s, v) = deg_G(v, \psi)$, $c(v, t) = | V_\psi | \alpha$, $c(v, s) = c(t, v) = 0$
		\ENDFOR
		\STATE $\Lambda' \gets$ Set of groups $g$ of $\psi$-instances in $\Lambda$ with the same node set
		\FORALL{$g' \in \Lambda$} \label{line:begin_flow_pat}
		\STATE $\lambda' \gets$ Node set of instances in $g$
		\STATE $V_\mathcal{H} \gets V_\mathcal{H} \cup \{\lambda'\}$
		\FORALL{$v' \in \lambda'$}
		\STATE $E_{\mathcal{H}} \gets E_{\mathcal{H}} \cup \{(\lambda', v'), (v', \lambda')\}$
		\STATE Set $c(\lambda', v') = |g| \left( | V_\psi | - 1 \right)$, $c(v', \lambda') = |g|$
		\ENDFOR
		\ENDFOR \label{line:end_flow_pat}
		\RETURN $(V_\mathcal{H}, E_\mathcal{H}, c)$
	\end{algorithmic}
\end{algorithm}

\begin{lem}
\label{lem:cut_pat}
For a set $V_1 \subseteq V$, the minimum value of an $s$-$t$ cut $(\mathcal{S}, V_\mathcal{H} \setminus \mathcal{S})$ of $\mathcal{H}$ having $\mathcal{S} \cap V = V_1$ is
\begin{small}
\begin{equation*}
    \min_{\mathcal{S} : s \in \mathcal{S}, t \notin \mathcal{S}, \mathcal{S} \cap V = V_1}\!\!\!\!\!\! c \left( \mathcal{S}, V_\mathcal{H} \setminus \mathcal{S} \right) = \left| V_\psi \right| \left[ \mu_\psi(G) + \left[ \rho_\psi^* - \rho_\psi(V_1) \right] |V_1| \right]
\end{equation*}
\end{small}
Moreover, one value of $\mathcal{S}$ achieving the minimum above has $\Lambda'_1 = \mathcal{S} \cap \Lambda'$ as the set of all $\psi$-instance groups in $\Lambda'$ which consist only of nodes in $V_1$.
\end{lem}
\begin{proof}
Similar to that of Lemma 13 in \cite{FYCLX19}.
\end{proof}

\revise{\section{Extension of Expected Densest Subgraphs to Clique and Pattern Densities}
\label{sec:exp_ext}
The notion of expected densest subgraphs in uncertain graphs, which has been studied for edge density \cite{Zou13}, can be extended to clique and pattern densities. The following discussion is for pattern density, and hence also holds for the special case of clique density.
\begin{defn}
[Expected Pattern Density]
\label{def:exp_pattern}
Given an uncertain graph $\mathcal{G} = (V, E, p)$ and a pattern $\psi$, the expected pattern density (w.r.t. $\psi$) of $\mathcal{G}$ is the expectation of the pattern density of all possible worlds of $\mathcal{G}$. Formally,
\begin{small}
\begin{equation}
    \overline{\rho}_\psi(\mathcal{G}) = \mathbbm{E} \left[ \rho_\psi(G) \right] = \sum_{G \sqsubseteq \mathcal{G}} \Pr(G) \times \rho_\psi(G)
\end{equation}
\end{small}
\end{defn}
As shown in Theorem \ref{th:exp_pattern}, the expected pattern density of a subgraph w.r.t. any pattern can be expressed as the weighted pattern density of that subgraph, which is ratio of the sum of the weights of all pattern instances in the deterministic version of the subgraph to the number of nodes in the subgraph. Here the weight of a pattern instance is its existence probability, i.e., the product of its edge probabilities.
\begin{theor}
\label{th:exp_pattern}
Given an uncertain graph $\mathcal{G} = (V, E, p)$ and a pattern $\psi$, the expected pattern density (w.r.t. $\psi$) of a subgraph $\mathcal{G}' = (V', E', p)$ can be computed as
\begin{small}
\begin{equation}
    \overline{\rho}_\psi(\mathcal{G}') = \frac{1}{|V'|} \sum_{\omega \in \Omega} \prod_{e \in E_\omega} p(e)
\end{equation}
\end{small}
where $\Omega$ is the set of all instances $\omega = \left( V_\omega, E_\omega \right)$ of $\psi$ in the deterministic version of $\mathcal{G}'$.
\end{theor}
\begin{proof}
Consider a possible world $G' \sqsubseteq \mathcal{G}'$. For each $\omega \in \Omega$, let $X_\omega$ be a binary random variable denoting whether $\omega$ exists in $G'$. Then $\mathbbm{E} \left[ X_\omega \right] = \Pr \left( X_\omega = 1 \right) = \prod_{e \in E_\omega} p(e)$ and the number of $\psi$-instances in $G'$ is $\sum_{\omega \in \Omega} X_\omega$. Thus the pattern density (w.r.t. $\psi$) of $G'$ is
\begin{small}
\begin{equation*}
    \rho_\psi(G') = \frac{1}{|V'|} \sum_{\omega \in \Omega} X_\omega
\end{equation*}
\end{small}
From Definition \ref{def:exp_pattern} and the linearity of expectation, we have
\begin{small}
\begin{equation*}
    \overline{\rho}_\psi(\mathcal{G}') = \mathbbm{E} \left[ \rho_\psi(G') \right] = \frac{1}{|V'|} \sum_{\omega \in \Omega} \mathbbm{E} \left[ X_\omega \right] = \frac{1}{|V'|} \sum_{\omega \in \Omega} \prod_{e \in E_\omega} p(e)
\end{equation*}
\end{small}
\end{proof}
Since the expected pattern density is equivalent to the weighted pattern density of a particular weighted deterministic graph, the expected pattern densest subgraph in an uncertain graph can be computed using existing methods like \cite{SDCS20}.}

\balance
\bibliographystyle{IEEEtran}
\bibliography{ref}

\begin{thebibliography}{100}
\providecommand{\url}[1]{#1}
\csname url@samestyle\endcsname
\providecommand{\newblock}{\relax}
\providecommand{\bibinfo}[2]{#2}
\providecommand{\BIBentrySTDinterwordspacing}{\spaceskip=0pt\relax}
\providecommand{\BIBentryALTinterwordstretchfactor}{4}
\providecommand{\BIBentryALTinterwordspacing}{\spaceskip=\fontdimen2\font plus
\BIBentryALTinterwordstretchfactor\fontdimen3\font minus
  \fontdimen4\font\relax}
\providecommand{\BIBforeignlanguage}[2]{{%
\expandafter\ifx\csname l@#1\endcsname\relax
\typeout{** WARNING: IEEEtran.bst: No hyphenation pattern has been}%
\typeout{** loaded for the language `#1'. Using the pattern for}%
\typeout{** the default language instead.}%
\else
\language=\csname l@#1\endcsname
\fi
#2}}
\providecommand{\BIBdecl}{\relax}
\BIBdecl

\bibitem{G84}
A.~V. Goldberg, \emph{Finding a Maximum Density Subgraph}.\hskip 1em plus 0.5em
  minus 0.4em\relax University of California Berkeley, 1984.

\bibitem{Charikar00}
M.~Charikar, ``Greedy approximation algorithms for finding dense components in
  a graph,'' in \emph{International Workshop on Approximation Algorithms for
  Combinatorial Optimization}.\hskip 1em plus 0.5em minus 0.4em\relax Berlin,
  Heidelberg: Springer, 2000, pp. 84--95.

\bibitem{TsourakakisBGGT13}
C.~Tsourakakis, F.~Bonchi, A.~Gionis, F.~Gullo, and M.~Tsiarli, ``Denser than
  the densest subgraph: Extracting optimal quasi-cliques with quality
  guarantees,'' in \emph{Proceedings of the 19th ACM SIGKDD International
  Conference on Knowledge Discovery and Data Mining}, 2013, p. 104–112.

\bibitem{GT15}
A.~Gionis and C.~E. Tsourakakis, ``Dense subgraph discovery,'' in
  \emph{Proceedings of the 21th ACM SIGKDD International Conference on
  Knowledge Discovery and Data Mining}, 2015, p. 2313–2314.

\bibitem{FYCLX19}
Y.~Fang, K.~Yu, R.~Cheng, L.~V. Lakshmanan, and X.~Lin, ``Efficient algorithms
  for densest subgraph discovery,'' \emph{Proceedings of the VLDB Endowment},
  vol.~12, no.~11, p. 1719–1732, 2019.

\bibitem{DourisboureGP09}
Y.~Dourisboure, F.~Geraci, and M.~Pellegrini, ``Extraction and classification
  of dense implicit communities in the web graph,'' \emph{ACM Transactions on
  the Web}, vol.~3, no.~2, pp. 1--36, 2009.

\bibitem{asatani2021dense}
K.~Asatani, H.~Yamano, T.~Sakaki, and I.~Sakata, ``Dense and influential core
  promotion of daily viral information spread in political echo chambers,''
  \emph{Scientific reports}, vol.~11, no.~1, pp. 1--10, 2021.

\bibitem{L22}
L.~V. Lakshmanan, ``On a quest for combating filter bubbles and
  misinformation,'' in \emph{Proceedings of the 2022 ACM SIGMOD International
  Conference on Management of Data}, 2022, p.~2.

\bibitem{legenstein_et_al:LIPIcs}
R.~Legenstein, W.~Maass, C.~H. Papadimitriou, and S.~S. Vempala, ``Long term
  memory and the densest k-subgraph problem,'' in \emph{9th Innovations in
  Theoretical Computer Science Conference (ITCS 2018)}, ser. Leibniz
  International Proceedings in Informatics (LIPIcs), vol.~94.\hskip 1em plus
  0.5em minus 0.4em\relax Dagstuhl, Germany: Schloss Dagstuhl--Leibniz-Zentrum
  fuer Informatik, 2018, pp. 57:1--57:15.

\bibitem{wu2021extracting}
Q.~Wu, X.~Huang, A.~J. Culbreth, J.~A. Waltz, L.~E. Hong, and S.~Chen,
  ``Extracting brain disease-related connectome subgraphs by adaptive dense
  subgraph discovery,'' \emph{Biometrics}, 2021.

\bibitem{DuJDLT09}
X.~Du, R.~Jin, L.~Ding, V.~E. Lee, and J.~H. Thornton, ``Migration motif: A
  spatial-temporal pattern mining approach for financial markets,'' in
  \emph{Proceedings of the 15th ACM SIGKDD International Conference on
  Knowledge Discovery and Data Mining}, 2009, p. 1135–1144.

\bibitem{BuehrerC08}
G.~Buehrer and K.~Chellapilla, ``A scalable pattern mining approach to web
  graph compression with communities,'' in \emph{Proceedings of the 2008
  International Conference on Web Search and Data Mining}, 2008, p. 95–106.

\bibitem{ZhangP12}
Y.~Zhang and S.~Parthasarathy, ``Extracting analyzing and visualizing triangle
  k-core motifs within networks,'' in \emph{IEEE 28th international conference
  on data engineering}, 2012, pp. 1049--1060.

\bibitem{ZhaoT12}
F.~Zhao and A.~K.~H. Tung, ``Large scale cohesive subgraphs discovery for
  social network visual analysis,'' \emph{Proceedings of the VLDB Endowment},
  vol.~6, no.~2, p. 85–96, 2012.

\bibitem{CohenHKZ03}
E.~Cohen, E.~Halperin, H.~Kaplan, and U.~Zwick, ``Reachability and distance
  queries via 2-hop labels,'' in \emph{Proceedings of the 13th Annual ACM-SIAM
  Symposium on Discrete Algorithms}.\hskip 1em plus 0.5em minus 0.4em\relax
  USA: Society for Industrial and Applied Mathematics, 2002, p. 937–946.

\bibitem{JinXRF09}
R.~Jin, Y.~Xiang, N.~Ruan, and D.~Fuhry, ``3-hop: A high-compression indexing
  scheme for reachability query,'' in \emph{Proceedings of the 2009 ACM SIGMOD
  International Conference on Management of Data}, 2009, p. 813–826.

\bibitem{GionisJLSW13}
A.~Gionis, F.~Junqueira, V.~Leroy, M.~Serafini, and I.~Weber, ``Piggybacking on
  social networks,'' \emph{Proceedings of the VLDB Endowment}, vol.~6, no.~6,
  p. 409–420, 2013.

\bibitem{FaragoM19}
A.~Farag{\'o} and Z.~R. Mojaveri, ``In search of the densest subgraph,''
  \emph{Algorithms}, vol.~12, no.~8, p. 157, 2019.

\bibitem{T15}
C.~E. Tsourakakis, ``The k-clique densest subgraph problem,'' in
  \emph{Proceedings of the 24th International Conference on World Wide
  Web}.\hskip 1em plus 0.5em minus 0.4em\relax Republic and Canton of Geneva,
  CHE: International World Wide Web Conferences Steering Committee, 2015, p.
  1122–1132.

\bibitem{MPPTX15}
M.~Mitzenmacher, J.~Pachocki, R.~Peng, C.~E. Tsourakakis, and S.~C. Xu,
  ``Scalable large near-clique detection in large-scale networks via
  sampling,'' in \emph{Proceedings of the 21st ACM SIGKDD International
  Conference on Knowledge Discovery and Data Mining}, 2015, p. 815–824.

\bibitem{YanagisawaH18}
H.~Yanagisawa and S.~Hara, ``Discounted average degree density metric and new
  algorithms for the densest subgraph problem,'' \emph{Networks}, vol.~71,
  no.~1, pp. 3--15, 2018.

\bibitem{Aggarwal09}
C.~C. Aggarwal, \emph{Managing and Mining Uncertain Data}, ser. Advances in
  Database Systems.\hskip 1em plus 0.5em minus 0.4em\relax Boston, MA, USA:
  Springer, 2009, vol.~35.

\bibitem{AdarR07}
E.~Adar and C.~R\'{e}, ``Managing uncertainty in social networks,''
  \emph{{IEEE} Data Engineering Bulletin}, vol.~30, no.~2, pp. 15--22, 2007.

\bibitem{Liben-NowellK07}
D.~Liben-Nowell and J.~Kleinberg, ``The link prediction problem for social
  networks,'' in \emph{Proceedings of the 12th International Conference on
  Information and Knowledge Management}, 2003, p. 556–559.

\bibitem{BoldiBGT12}
P.~Boldi, F.~Bonchi, A.~Gionis, and T.~Tassa, ``Injecting uncertainty in graphs
  for identity obfuscation,'' \emph{Proceedings of the VLDB Endowment}, vol.~5,
  no.~11, p. 1376–1387, 2012.

\bibitem{KYC18}
A.~Khan, Y.~Ye, and L.~Chen, \emph{{On Uncertain Graphs}}, ser. Synthesis
  Lectures on Data Management.\hskip 1em plus 0.5em minus 0.4em\relax Morgan
  {\&} Claypool Publishers, 2018, vol.~10.

\bibitem{KassianoGPT16}
V.~Kassiano, A.~Gounaris, A.~N. Papadopoulos, and K.~Tsichlas, ``Mining
  uncertain graphs: An overview,'' in \emph{International Workshop of
  Algorithmic Aspects of Cloud Computing}.\hskip 1em plus 0.5em minus
  0.4em\relax Cham: Springer, 2016, pp. 87--116.

\bibitem{KhanC15}
A.~Khan and L.~Chen, ``On uncertain graphs modeling and queries,''
  \emph{Proceedings of the VLDB Endowment}, vol.~8, no.~12, p. 2042–2043,
  2015.

\bibitem{SevonEHKT06}
P.~Sevon, L.~Eronen, P.~Hintsanen, K.~Kulovesi, and H.~Toivonen, ``Link
  discovery in graphs derived from biological databases,'' in
  \emph{International Workshop on Data Integration in the Life Sciences}.\hskip
  1em plus 0.5em minus 0.4em\relax Berlin, Heidelberg: Springer, 2006, pp.
  35--49.

\bibitem{ZhangRCSWW17}
C.~Zhang, C.~R\'{e}, M.~Cafarella, C.~De~Sa, A.~Ratner, J.~Shin, F.~Wang, and
  S.~Wu, ``Deepdive: Declarative knowledge base construction,''
  \emph{Communications of the ACM}, vol.~60, no.~5, p. 93–102, 2017.

\bibitem{ZouGL10}
Z.~Zou, H.~Gao, and J.~Li, ``Discovering frequent subgraphs over uncertain
  graph databases under probabilistic semantics,'' in \emph{Proceedings of the
  16th ACM SIGKDD International Conference on Knowledge Discovery and Data
  Mining}, 2010, p. 633–642.

\bibitem{KempeKT03}
D.~Kempe, J.~Kleinberg, and E.~Tardos, ``Maximizing the spread of influence
  through a social network,'' in \emph{Proceedings of the 9th ACM SIGKDD
  International Conference on Knowledge Discovery and Data Mining}, 2003, p.
  137–146.

\bibitem{HP10}
M.~Hua and J.~Pei, ``Probabilistic path queries in road networks: Traffic
  uncertainty aware path selection,'' in \emph{Proceedings of the 13th
  International Conference on Extending Database Technology}, 2010, p.
  347–358.

\bibitem{YalavarthiKK17}
V.~K. Yalavarthi, X.~Ke, and A.~Khan, ``Select your questions wisely: For
  entity resolution with crowd errors,'' in \emph{Proceedings of the 2017 ACM
  Conference on Information and Knowledge Management}, 2017, p. 317–326.

\bibitem{PBGK10}
M.~Potamias, F.~Bonchi, A.~Gionis, and G.~Kollios, ``k-nearest neighbors in
  uncertain graphs,'' \emph{Proceedings of the VLDB Endowment}, vol.~3, no.
  1–2, p. 997–1008, 2010.

\bibitem{SahaBVKB21}
A.~Saha, R.~Brokkelkamp, Y.~Velaj, A.~Khan, and F.~Bonchi, ``Shortest paths and
  centrality in uncertain networks,'' \emph{Proceedings of the VLDB Endowment},
  vol.~14, no.~7, p. 1188–1201, 2021.

\bibitem{MukherjeeXT17}
A.~P. Mukherjee, P.~Xu, and S.~Tirthapura, ``Enumeration of maximal cliques
  from an uncertain graph,'' \emph{IEEE Transactions on Knowledge and Data
  Engineering}, vol.~29, no.~3, pp. 543--555, 2016.

\bibitem{ZouLGZ10}
Z.~Zou, J.~Li, H.~Gao, and S.~Zhang, ``Finding top-k maximal cliques in an
  uncertain graph,'' in \emph{IEEE 26th International Conference on Data
  Engineering}, 2010, pp. 649--652.

\bibitem{LiDWMQY19}
R.-H. Li, Q.~Dai, G.~Wang, Z.~Ming, L.~Qin, and J.~X. Yu, ``Improved algorithms
  for maximal clique search in uncertain networks,'' in \emph{IEEE 35th
  International Conference on Data Engineering}, 2019, pp. 1178--1189.

\bibitem{BonchiGKV14}
F.~Bonchi, F.~Gullo, A.~Kaltenbrunner, and Y.~Volkovich, ``Core decomposition
  of uncertain graphs,'' in \emph{Proceedings of the 20th ACM SIGKDD
  International Conference on Knowledge Discovery and Data Mining}, 2014, p.
  1316–1325.

\bibitem{HuangLL16}
X.~Huang, W.~Lu, and L.~V. Lakshmanan, ``Truss decomposition of probabilistic
  graphs: Semantics and algorithms,'' in \emph{Proceedings of the 2016 ACM
  SIGMOD International Conference on Management of Data}, 2016, p. 77–90.

\bibitem{HanGXTHCH19}
K.~Han, F.~Gui, X.~Xiao, J.~Tang, Y.~He, Z.~Cao, and H.~Huang, ``Efficient and
  effective algorithms for clustering uncertain graphs,'' \emph{Proceedings of
  the VLDB Endowment}, vol.~12, no.~6, p. 667–680, 2019.

\bibitem{MaCLGF019}
C.~Ma, R.~Cheng, L.~V. Lakshmanan, T.~Grubenmann, Y.~Fang, and X.~Li, ``Linc: A
  motif counting algorithm for uncertain graphs,'' \emph{Proceedings of the
  VLDB Endowment}, vol.~13, no.~2, p. 155–168, 2019.

\bibitem{Zou13}
Z.~Zou, ``{Polynomial-time algorithm for finding densest subgraphs in uncertain
  graphs},'' in \emph{Proceedings of the 11th Workshop on Mining and Learning
  with Graphs}, 2013.

\bibitem{tsourakakis2020novel}
C.~E. Tsourakakis, T.~Chen, N.~Kakimura, and J.~Pachocki, ``Novel dense
  subgraph discovery primitives: Risk aversion and exclusion queries,'' in
  \emph{Joint European Conference on Machine Learning and Knowledge Discovery
  in Databases}.\hskip 1em plus 0.5em minus 0.4em\relax Springer, 2020, pp.
  378--394.

\bibitem{CQ20}
L.~Chang and M.~Qiao, ``Deconstruct densest subgraphs,'' in \emph{Proceedings
  of the 29th International Conference on World Wide Web}, 2020, p.
  2747–2753.

\bibitem{wang2005tfp}
J.~Wang, J.~Han, Y.~Lu, and P.~Tzvetkov, ``Tfp: an efficient algorithm for
  mining top-k frequent closed itemsets,'' \emph{IEEE Transactions on Knowledge
  and Data Engineering}, vol.~17, no.~5, pp. 652--663, 2005.

\bibitem{V79}
L.~G. Valiant, ``{The complexity of enumeration and reliability problems},''
  \emph{SIAM Journal on Computing}, vol.~8, no.~3, pp. 410--421, 1979.

\bibitem{B86}
M.~O. Ball, ``{Computational complexity of network reliability analysis: An
  overview},'' \emph{IEEE Transactions on Reliability}, vol.~35, no.~3, pp.
  230--239, 1986.

\bibitem{KBGN18}
A.~Khan, F.~Bonchi, F.~Gullo, and A.~Nufer, ``Conditional reliability in
  uncertain graphs,'' \emph{IEEE Transactions on Knowledge and Data
  Engineering}, vol.~30, no.~11, pp. 2078--2092, 2018.

\bibitem{JY12}
J.~Chen and Y.~Saad, ``Dense subgraph extraction with application to community
  detection,'' \emph{TKDE}, vol.~24, no.~7, pp. 1216--1230, 2012.

\bibitem{YinBLG17}
H.~Yin, A.~R. Benson, J.~Leskovec, and D.~F. Gleich, ``Local higher-order graph
  clustering,'' in \emph{Proceedings of the 23rd {ACM} {SIGKDD} International
  Conference on Knowledge Discovery and Data Mining}.\hskip 1em plus 0.5em
  minus 0.4em\relax New York, NY, USA: {ACM}, 2017, pp. 555--564.

\bibitem{BZ03}
V.~Batagelj and M.~Zaversnik, ``An o(m) algorithm for cores decomposition of
  networks,'' \emph{CoRR}, vol. cs.DS/0310049, 2003.

\bibitem{LFZT17}
Y.~Li, J.~Fan, D.~Zhang, and K.-L. Tan, ``Discovering your selling points:
  Personalized social influential tags exploration,'' in \emph{Proceedings of
  the 2017 ACM SIGMOD International Conference on Management of Data}, 2017, p.
  619–634.

\bibitem{RSS16}
R.-H. Li, J.~X. Yu, R.~Mao, and T.~Jin, ``Recursive stratified sampling: A new
  framework for query evaluation on uncertain graphs,'' \emph{IEEE Transactions
  on Knowledge and Data Engineering}, vol.~28, no.~2, pp. 468--482, 2016.

\bibitem{DBS18}
M.~Danisch, O.~D. Balalau, and M.~Sozio, ``Listing k-cliques in sparse
  real-world graphs,'' in \emph{Proceedings of the 2018 World Wide Web
  Conference}.\hskip 1em plus 0.5em minus 0.4em\relax Republic and Canton of
  Geneva, CHE: International World Wide Web Conferences Steering Committee,
  2018, p. 589–598.

\bibitem{SDCS20}
B.~Sun, M.~Danisch, T.-H.~H. Chan, and M.~Sozio, ``Kclist++: A simple algorithm
  for finding k-clique densest subgraphs in large graphs,'' \emph{Proceedings
  of the VLDB Endowment}, vol.~13, no.~10, p. 1628–1640, 2020.

\bibitem{QZC17}
M.~Qiao, H.~Zhang, and H.~Cheng, ``Subgraph matching: On compression and
  computation,'' \emph{Proceedings of the VLDB Endowment}, vol.~11, no.~2, p.
  176–188, 2017.

\bibitem{BBCGS15}
O.~D. Balalau, F.~Bonchi, T.-H.~H. Chan, F.~Gullo, and M.~Sozio, ``Finding
  subgraphs with maximum total density and limited overlap,'' in
  \emph{Proceedings of the Eighth ACM International Conference on Web Search
  and Data Mining}, 2015, p. 379–388.

\bibitem{BKV12}
B.~Bahmani, R.~Kumar, and S.~Vassilvitskii, ``Densest subgraph in streaming and
  mapreduce,'' \emph{Proc. VLDB Endow.}, vol.~5, no.~5, p. 454–465, jan 2012.

\bibitem{SW20}
S.~Sawlani and J.~Wang, ``Near-optimal fully dynamic densest subgraph,'' in
  \emph{Proceedings of the 52nd Annual ACM SIGACT Symposium on Theory of
  Computing}, ser. STOC 2020.\hskip 1em plus 0.5em minus 0.4em\relax New York,
  NY, USA: Association for Computing Machinery, 2020, p. 181–193.

\bibitem{BGMT22}
F.~Bonchi, D.~Garc\'{\i}a-Soriano, A.~Miyauchi, and C.~E. Tsourakakis,
  ``Finding densest k-connected subgraphs,'' \emph{Discrete Appl. Math.}, vol.
  305, no.~C, p. 34–47, dec 2022.

\bibitem{AC09}
R.~Andersen and K.~Chellapilla, ``Finding dense subgraphs with size bounds,''
  in \emph{Algorithms and Models for the Web-Graph}, K.~Avrachenkov, D.~Donato,
  and N.~Litvak, Eds.\hskip 1em plus 0.5em minus 0.4em\relax Berlin,
  Heidelberg: Springer Berlin Heidelberg, 2009, pp. 25--37.

\bibitem{GGT16}
E.~Galbrun, A.~Gionis, and N.~Tatti, ``Top-k overlapping densest subgraphs,''
  \emph{Data Min. Knowl. Discov.}, vol.~30, no.~5, p. 1134–1165, sep 2016.

\bibitem{NGMG17}
M.~A.~U. Nasir, A.~Gionis, G.~D.~F. Morales, and S.~Girdzijauskas, ``Fully
  dynamic algorithm for top-k densest subgraphs,'' in \emph{Proceedings of the
  2017 ACM on Conference on Information and Knowledge Management}, ser. CIKM
  '17.\hskip 1em plus 0.5em minus 0.4em\relax New York, NY, USA: Association
  for Computing Machinery, 2017, p. 1817–1826.

\bibitem{QLCZ15}
L.~Qin, R.-H. Li, L.~Chang, and C.~Zhang, ``Locally densest subgraph
  discovery,'' in \emph{Proceedings of the 21th ACM SIGKDD International
  Conference on Knowledge Discovery and Data Mining}, ser. KDD '15.\hskip 1em
  plus 0.5em minus 0.4em\relax New York, NY, USA: Association for Computing
  Machinery, 2015, p. 965–974.

\bibitem{miyauchi2018robust}
A.~Miyauchi and A.~Takeda, ``Robust densest subgraph discovery,'' in \emph{2018
  IEEE International Conference on Data Mining (ICDM)}.\hskip 1em plus 0.5em
  minus 0.4em\relax IEEE, 2018, pp. 1188--1193.

\bibitem{DCS17}
M.~Danisch, T.-H.~H. Chan, and M.~Sozio, ``Large scale density-friendly graph
  decomposition via convex programming,'' in \emph{Proceedings of the 26th
  International Conference on World Wide Web}, ser. WWW '17.\hskip 1em plus
  0.5em minus 0.4em\relax Republic and Canton of Geneva, CHE: International
  World Wide Web Conferences Steering Committee, 2017, p. 233–242.

\bibitem{T19}
N.~Tatti, ``Density-friendly graph decomposition,'' \emph{ACM Trans. Knowl.
  Discov. Data}, vol.~13, no.~5, sep 2019.

\bibitem{MFCLZL20}
C.~Ma, Y.~Fang, R.~Cheng, L.~V. Lakshmanan, W.~Zhang, and X.~Lin, ``Efficient
  algorithms for densest subgraph discovery on large directed graphs,'' in
  \emph{Proceedings of the 2020 ACM SIGMOD International Conference on
  Management of Data}, ser. SIGMOD '20.\hskip 1em plus 0.5em minus 0.4em\relax
  New York, NY, USA: Association for Computing Machinery, 2020, p. 1051–1066.

\bibitem{HSBSSF16}
B.~Hooi, H.~A. Song, A.~Beutel, N.~Shah, K.~Shin, and C.~Faloutsos, ``Fraudar:
  Bounding graph fraud in the face of camouflage,'' in \emph{Proceedings of the
  22nd ACM SIGKDD International Conference on Knowledge Discovery and Data
  Mining}, ser. KDD '16.\hskip 1em plus 0.5em minus 0.4em\relax New York, NY,
  USA: Association for Computing Machinery, 2016, p. 895–904.

\bibitem{GBBL20}
E.~Galimberti, F.~Bonchi, F.~Gullo, and T.~Lanciano, ``Core decomposition in
  multilayer networks: Theory, algorithms, and applications,'' \emph{ACM Trans.
  Knowl. Discov. Data}, vol.~14, no.~1, jan 2020.

\bibitem{FangLM22}
Y.~Fang, W.~Luo, and C.~Ma, ``Densest subgraph discovery on large graphs:
  Applications, challenges, and techniques,'' \emph{Proc. {VLDB} Endow.},
  vol.~15, no.~12, pp. 3766--3769, 2022.

\bibitem{EsfahaniW0T019}
F.~Esfahani, J.~Wu, V.~Srinivasan, A.~Thomo, and K.~Wu, ``Fast truss
  decomposition in large-scale probabilistic graphs,'' in \emph{Advances in
  Database Technology - 22nd International Conference on Extending Database
  Technology}.\hskip 1em plus 0.5em minus 0.4em\relax OpenProceedings.org,
  2019, pp. 722--725.

\bibitem{SMPV17}
D.~Seux, F.~Malliaros, A.~Papadopoulos, and M.~Vazirgiannis, ``Core
  decomposition of uncertain graphs using representative instances,'' in
  \emph{6th International Conference on Complex Networks and Their
  Applications}, 2017.

\bibitem{Dai21}
Q.~Dai, R.~Li, G.~Wang, R.~Mao, Z.~Zhang, and Y.~Yuan, ``Core decomposition on
  uncertain graphs revisited,'' \emph{IEEE Transactions on Knowledge and Data
  Engineering}, 2021.

\bibitem{Sun0XB21}
Z.~Sun, X.~Huang, J.~Xu, and F.~Bonchi, ``Efficient probabilistic truss
  indexing on uncertain graphs,'' in \emph{The Web Conference}.\hskip 1em plus
  0.5em minus 0.4em\relax {ACM} / {IW3C2}, 2021, pp. 354--366.

\bibitem{ZouZ17}
Z.~Zou and R.~Zhu, ``Truss decomposition of uncertain graphs,'' \emph{Knowledge
  and Information Systems}, vol.~50, no.~1, pp. 197--230, 2017.

\bibitem{DLLCW22}
Q.~Dai, R.-H. Li, M.~Liao, H.~Chen, and G.~Wang, ``Fast maximal clique
  enumeration on uncertain graphs: A pivot-based approach,'' in
  \emph{Proceedings of the 2022 ACM SIGMOD International Conference on
  Management of Data}, 2022, p. 2034–2047.

\bibitem{LJAS12}
L.~Liu, R.~Jin, C.~C. Aggarwal, and Y.~Shen, ``Reliable clustering on uncertain
  graphs,'' in \emph{{IEEE} International Conference on Data Mining}, 2012, pp.
  459--468.

\bibitem{CeccarelloFPPV17}
M.~Ceccarello, C.~Fantozzi, A.~Pietracaprina, G.~Pucci, and F.~Vandin,
  ``Clustering uncertain graphs,'' \emph{Proceedings of the VLDB Endowment},
  vol.~11, no.~4, pp. 472--484, 2017.

\bibitem{JinLA11}
R.~Jin, L.~Liu, and C.~C. Aggarwal, ``Discovering highly reliable subgraphs in
  uncertain graphs,'' in \emph{Proceedings of the 17th {ACM} International
  Conference on Knowledge Discovery and Data Mining}.\hskip 1em plus 0.5em
  minus 0.4em\relax {ACM}, 2011, pp. 992--1000.

\bibitem{code}
\BIBentryALTinterwordspacing
A.~Saha, X.~Ke, A.~Khan, and C.~Long, ``{Most probable densest subgraphs: Code
  and data},'' 2022. [Online]. Available:
  \url{https://github.com/ArkaSaha/MPDS}
\BIBentrySTDinterwordspacing

\bibitem{zachary1977information}
W.~W. Zachary, ``An information flow model for conflict and fission in small
  groups,'' \emph{Journal of anthropological research}, vol.~33, no.~4, pp.
  452--473, 1977.

\bibitem{Intel}
\BIBentryALTinterwordspacing
S.~Madden, ``Intel lab data,'' 2004. [Online]. Available:
  \url{http://db.csail.mit.edu/labdata/labdata.html}
\BIBentrySTDinterwordspacing

\bibitem{lastfm}
\BIBentryALTinterwordspacing
``Lastfm api,'' 2002. [Online]. Available: \url{https://www.last.fm}
\BIBentrySTDinterwordspacing

\bibitem{szklarczyk2021string}
D.~Szklarczyk, A.~L. Gable, K.~C. Nastou, D.~Lyon, R.~Kirsch, S.~Pyysalo, N.~T.
  Doncheva, M.~Legeay, T.~Fang, P.~Bork \emph{et~al.}, ``The string database in
  2021: Customizable protein--protein networks, and functional characterization
  of user-uploaded gene/measurement sets,'' \emph{Nucleic Acids Research},
  vol.~49, no.~D1, pp. D605--D612, 2021.

\bibitem{BioMine}
L.~Eronen and H.~Toivonen, ``Biomine: Predicting links between biological
  entities using network models of heterogeneous databases,'' \emph{BMC
  Bioinformatics}, vol.~13, no.~1, 2012.

\bibitem{snapnets}
J.~Leskovec and A.~Krevl, ``{SNAP Datasets}: {Stanford} large network dataset
  collection,'' \url{http://snap.stanford.edu/data}, 2014.

\bibitem{Yang2012DefiningAE}
J.~Yang and J.~Leskovec, ``Defining and evaluating network communities based on
  ground-truth,'' \emph{Knowledge and Information Systems}, vol.~42, pp.
  181--213, 2015.

\bibitem{KKHR20}
X.~Ke, A.~Khan, M.~Al~Hasan, and R.~Rezvansangsari, ``Reliability maximization
  in uncertain graphs,'' \emph{IEEE Transactions on Knowledge and Data
  Engineering}, vol.~34, no.~2, pp. 894--913, 2022.

\bibitem{JJ11}
J.~J. Pfeiffer~III and J.~Neville, ``Methods to determine node centrality and
  clustering in graphs with uncertain structure,'' in \emph{Fifth International
  AAAI Conference on Weblogs and Social Media}, 2011.

\bibitem{conf.fninf.2013}
C.~Craddock, Y.~Benhajali, C.~Chu, F.~Chouinard, A.~Evans, A.~Jakab, B.~S.
  Khundrakpam, J.~D. Lewis, Q.~Li, M.~Milham, C.~Yan, and P.~Bellec, ``The
  neuro bureau preprocessing initiative: Open sharing of preprocessed
  neuroimaging data and derivatives,'' \emph{Frontiers in Neuroinformatics},
  2013.

\bibitem{xia2013brainnet}
M.~Xia, J.~Wang, and Y.~He, ``Brainnet viewer: a network visualization tool for
  human brain connectomics,'' \emph{PloS one}, vol.~8, no.~7, p. e68910, 2013.

\bibitem{Neu1}
A.~Di~Martino, C.~Kelly, R.~Grzadzinski, X.-N. Zuo, M.~Mennes, M.~Mairena,
  C.~Lord, F.~Castellanos, and M.~Milham, ``Aberrant striatal functional
  connectivity in children with autism,'' \emph{Biological Psychiatry},
  vol.~69, no.~9, pp. 847--56, 12 2010.

\bibitem{Neu2}
S.~Noonan, F.~Haist, and R.-A. M\"{u}ller, ``Aberrant functional connectivity
  in autism: Evidence from low-frequency bold signal fluctuations,''
  \emph{Brain Research}, vol. 1262, pp. 48--63, 02 2009.

\bibitem{Neur4}
M.~Postema, D.~Van~Rooij, E.~Anagnostou, C.~Arango, G.~Auzias, M.~Behrmann,
  G.~Busatto, S.~Calderoni, R.~Calvo, E.~Daly, C.~Deruelle, A.~Di~Martino,
  I.~Dinstein, F.~Duran, S.~Durston, C.~Ecker, S.~Ehrlich, D.~Fair, J.~Fedor,
  and C.~Francks, ``Altered structural brain asymmetry in autism spectrum
  disorder in a study of 54 datasets,'' \emph{Nature Communications}, vol.~10,
  12 2019.

\bibitem{KeKQ19}
X.~Ke, A.~Khan, and L.~L.~H. Quan, ``An in-depth comparison of s-t reliability
  algorithms over uncertain graphs,'' \emph{Proceedings of the VLDB Endowment},
  vol.~12, no.~8, p. 864–876, 2019.

\bibitem{erdHos1960evolution}
P.~Erd{\"o}s and A.~R{\'e}nyi, ``{On random graphs},'' \emph{Publicationes
  Mathematicae Debrecen}, vol.~6, pp. 290--297, 1959.

\bibitem{barabasi1999emergence}
A.-L. Barab{\'a}si and R.~Albert, ``Emergence of scaling in random networks,''
  \emph{Science}, vol. 286, no. 5439, pp. 509--512, 1999.

\end{thebibliography}

\end{document}